\documentclass[a4paper,12pt]{article}

\usepackage{amsmath,amssymb,amsfonts,amsthm,mathtools} 
\usepackage{bm,bbm}
\usepackage{booktabs}
\usepackage{color}
\usepackage[dvips,letterpaper]{geometry}
\usepackage{fullpage}
\usepackage{here}
\usepackage{graphicx,psfrag,epsf}
\usepackage{indentfirst}
\usepackage{lineno}
\usepackage{natbib}
\usepackage{setspace}
\usepackage{subfigure} 
\usepackage{subfig}
\usepackage{times}
\usepackage{url}
\usepackage{verbatim}
\usepackage{multirow}
\usepackage[table,xcdraw]{xcolor}
\usepackage{adjustbox}
\usepackage{rotating}
\usepackage{algorithm,algcompatible,algpseudocode}
\usepackage{graphics}
\graphicspath{{figures/}}

 \usepackage{pdfpages} 
\usepackage{pdflscape}
\usepackage{scalefnt}

\usepackage{lipsum}

\usepackage{listings}
\usepackage{tikz}
\usetikzlibrary{matrix,shapes,arrows,positioning}

\usepackage{newunicodechar}
\newunicodechar{ﬁ}{fi}
\newunicodechar{ﬀ}{ff}
\newunicodechar{ﬁ}{fi}
\input cyracc.def

\newtheorem{Proposition}{Proposition}[section]

 \date{}

\title{\textbf{\Large Log-symmetric quantile regression models}}

 \author{\normalsize
{Helton Saulo}$^{1}$\thanks{Corresponding author. Helton Saulo. Department of Statistics, Universidade de Bras\'ilia, Bras\'ilia, Brazil. Email: heltonsaulo@gmail.com}\,\,, {Alan Dasilva}$^{1}$,\, {V\'ictor Leiva}$^{2}$\,\,\text{and}\,\,{Luis S\'{a}nchez}$^{3}$   \\
   {\small $^{1}$Department of Statistics, Universidade de Bras\'{i}lia, Bras\'{i}lia, Brazil}\\[-0.15cm]
  {\small $^{2}$Department of Statistics, Pontificia Universidad Cat\'{o}lica de  Chile, Santiago, Chile}\\[-0.15cm]
  {\small $^{3}$Department of Mathematics and Statistics, Universidad de La Frontera, Temuco, Chile}\\[-0.15cm]
 }

\begin{document}
	

\maketitle

\begin{abstract}

Regression models based on the log-symmetric family of distributions are particularly useful when the response is strictly positive and asymmetric. In this paper, we propose a class of quantile regression models based on reparameterized log-symmetric distributions, which have a quantile parameter. Two Monte Carlo simulation studies are carried out using the \texttt{R} software. The first one analyzes the performance of the maximum likelihood estimators, the information criteria AIC, BIC and AICc, and the generalized Cox-Snell and random quantile residuals. The second one evaluates the performance of the size and power of the Wald, likelihood ratio, score and gradient tests. A real box office data set is finally analyzed to illustrate the proposed approach.
 
\end{abstract}


\section{Introduction}\label{sec:01}

{Regression models with a continuous strictly positive and asymmetric dependent variable (response) have been widely applied in different fields, such as economics, environmental science, reliability and survival analysis, among others; see, for example, \cite{Lemonte-2009}, \cite{Leiva-2015a}, \cite{vslm:18}, \cite{vanegas2017log}, \cite{slgs:2020a,slgs:2020b} and \cite{slgs:2020c}.} Some distributions that have been used for regression are the log-normal, gamma, Weibull, Birnbaum-Saunders; see details on these models in \cite{jkb:94,jkb:95}. In a broader perspective, classes of distributions offer more flexibility for a regression model, since several distributions are obtained as special cases; the exponential family \citep{nm:83} is a good example. 

The class of log-symmetric distributions studied by \cite{vanegasp:16a} comprise several distributions that are generally used in the modeling of continuous, strictly positive and asymmetric data. This class also accommodates the possibility to model bimodal and/or light- and heavy-tailed data. Log-symmetric distributions are based on symmetric distributions \citep{fa:90}, that is, they arise
when the distribution of the logarithm of the random variable (RV) is symmetrical. Thus, if $Y$ is a symmetric random variable, then $T=\exp(Y)$ follows a log-symmetric distribution. Regression models based on log-symmetric distributions were studied by 
\cite{vanegas2015,vp:16a,vanegas2017log}, where it is allowed both median and skewness (or the relative dispersion) to be modeled. Other works addressing applications of log-symmetric models, mainly in the area of economics, can be seen {in \cite{saulo2017log}} and \cite{vslm:18}.

Quantile regression models play an important role in the regression literature. The advantage behind the use of these models lies in their flexibility to show different effects of the independent variables (covariates) on the response along the quantiles of the response. Moreover, quantile regression models are more robust to outliers; see \cite{k:05} and \cite{hn:07}.

The main objective of this paper is to propose a new class of quantile regression models based on the log-symmetric distributions. A reparameterization  of the log-symmetric distributions is proposed by inserting {the quantile function} as one of the parameters. We carry out two Monte Carlo simulation studies using the \texttt{R} software. The first study evaluates the performance of the maximum likelihood estimators, the generalized Cox-Snell (GCS) \citep{cs:68} and randomized quantile (RQ) residuals \citep{ds:96}, and the Akaike \citep{akaike1974}, Bayesian \citep{schwarz1978} and corrected Akaike \citep{bozdogan1987} information criteria. The second simulation study analyzes the performance of of the likelihood ratio, score, Wald and gradient tests.

The rest of this paper proceeds as follows. In Section \ref{sec:2}, we describe preliminary aspects regarding the family of log-symmetric distributions as well as the proposed reparameterized log-symmetric distribution. In Section \ref{sec:3}, we introduce the log-symmetric quantile regression model. In this section, we also describe the score function, Fisher information matrix, hypothesis tests based on the maximum likelihood estimators, model selection criteria and residuals. In Section \ref{sec:4}, we carry out two Monte Carlo simulation studies. In Section \ref{sec:5}, we apply the proposed model to a real box office data set, and finally in Section \ref{sec:6}, we provide some concluding remarks.

\section{Log-symmetric distributions}\label{sec:2}

{In this section, we describe briefly the classical log-symmetric distributions. We then introduce a quantile-based reparameterization of this distribution, and subsequently present some of its properties.}

\subsection{Classical log-symmetric distributions}

A random variable $Y$, strictly continuous with support on $(0,\infty)$, follows a log-symmetric distribution with scale parameter $\lambda > 0$ and power parameter $\phi > 0$, denoted by $Y \sim LS(\lambda,\phi,g)$, if its probability density function (PDF) and cumulative distribution function (CDF) are respectively given by 
\begin{equation}\label{eq:lg-pdf}
f_{Y}(y;\lambda,\phi)
=
\dfrac{\xi_{nc}}{\sqrt{\phi}\,y}\,
g\!\left(\dfrac{1}{\phi}\left[\log(y)-\log(\lambda)\right]^2\right), 
\quad y>0,
\end{equation}
and
\begin{equation}\label{eq:lg-cdf}
 F_{Y}(y;\lambda,\phi)=G\!\left(\dfrac{1}{\phi}\left[\log(y)-\log(\lambda)\right]^2 \right), \quad y>0,
\end{equation}
where $g(\cdot)$ is the density generator kernel possibly associated with an additional parameter $\vartheta$ (or parameter vector $\bm{\vartheta}$), $\xi_{nc}$ is a normalizing constant and $G(\omega) = \xi_{nc} \int_{-\infty}^{\omega} g(z^2) \textrm{d}z$, $\omega\in\mathbb{R}$; see \cite{vanegasp:16a}. The parameters $\lambda >0 $ and $\phi > 0$ represent the median and skewness (or relative dispersion), respectively, of the distribution of $Y$. The transformation $V = \log(Y)$ leads to an random variable $V$ following a symmetric distribution with location parameter $\mu = \log(\lambda) \in \mathbb{R}$, denoted $V \sim S(\mu, \phi, g)$, with PDF given by
\begin{equation}\label{eq:sym-pdf}
f_{V}(v;\mu,\phi)
=
\dfrac{\xi_{nc}}{\sqrt{\phi}}\,
g\!\left(\dfrac{1}{\phi}\left[v-\mu\right]^2\right),
\quad v\in\mathbb{R}.
\end{equation}
Note that $G$ in \eqref{eq:lg-cdf} is the CDF of a symmetric random variable $V \sim S(0, 1, g)$. The {100$q$-th quantile} of $Y \sim \textrm{LS}(\lambda,\phi,g)$ is given by
\begin{equation}\label{eq:lg-qf}
 Q_{Y}(q;\lambda,\phi)=\lambda\exp\big(\sqrt{\phi}\,z_{q}\big),
\end{equation}
where $z_{q}=G^{-1}(q)$ is the {100$q$-th quantile} of $V\sim\textrm{S}(\mu,\phi,g)$. Some log-symmetric distributions obtained from different $g$ function are presented in Table~\ref{tab:gfuncions}, {where $l$ is a real constant and  $\textrm{IGF}(x,l)= \int_{0}^{1}\exp(-xt)t^{l-1}\textrm{d} t$, for $l>0$ and $x \geq 0$, is the incomplete gamma function}. More details on the class of log-symmetric distributions, as well as other log-symmetric distributions can be seen in \cite{vanegasp:16a}.
\begin{table}[!ht] 
\centering
\small
\caption{Density generator $g(u)$ for some log-symmetric distributions.}\label{tab:gfuncions}
\vspace{0.15cm}
\begin{tabular}{llllll}
\hline {Distribution}                     &&             $g(u)$                                                              \\ 
\hline
Log-normal($\lambda,\phi$)                  && $\exp\left( -\frac{1}{2}u\right)$   
\\[1.5ex]
Log-Student-$t$($\lambda,\phi,\vartheta$)        && $\left(1+\frac{u}{\vartheta} \right)^{-\frac{\vartheta+1}{2}}$, $\vartheta>0$    
\\[1.5ex]  
Log-power-exponential($\lambda,\phi,\vartheta$)     && $\exp\left( -\frac{1}{2}u^{\frac{1}{1+\vartheta}}\right)$, $-1<{\vartheta}\leq{1}$   
\\[1.5ex]   
Log-hyperbolic($\lambda,\phi,\vartheta$) && $ \exp(-\vartheta\sqrt{1+u})$, $\vartheta>0$
\\[1.5ex]
Log-slash($\lambda,\phi,\vartheta$) && $\textrm{IGF}\left(\vartheta+\frac{1}{2} ,\frac{u}{2}\right)$, $\vartheta>0$
\\[1.5ex]
Log-contaminated-normal($\lambda,\phi,{\bm\vartheta}=(\vartheta_1,\vartheta_2)^\top$) &&
$
\sqrt{\vartheta_2}\exp\left(-\frac{1}{2}\vartheta_2 u\right)+\frac{(1-\vartheta_1)}{\vartheta_1}
\exp\left(-\frac{1}{2} u\right)$,
\\
&& $0<\vartheta_1,\vartheta_2<1$
\\[1.5ex]
Extended Birnbaum-Saunders($\lambda,\phi,\vartheta$)         && $\cosh(u^{1/2})\exp\left(-\frac{2}{\vartheta^2}\sinh^2(u^{1/2}) \right) $, $\vartheta>0$ 
\\[1.5ex] 
Extended Birnbaum-Saunders-$t$($\lambda,\phi,{\bm\vartheta}=(\vartheta_{1},\vartheta_{2})^{\top}$) && $\cosh(u^{1/2})\left(\vartheta_{2}\vartheta_{1}^2+4\sinh^2(u^{1/2})\right)^{-\frac{\vartheta_{2}+1}{2}} $,\\  
&& $\vartheta_{1},\vartheta_{2}>0$
\\\hline
\end{tabular}
\end{table}

\subsection{Quantile-based log-symmetric distributions}

Consider a fixed number $q \in (0,1)$ and the one-to-one transformation $(\lambda,\phi) \mapsto (Q,\phi)$, where $Q$ is the {100$q$-th quantile} of $Y$ given in \eqref{eq:lg-qf}. We introduce a reparametrization of the log-symmetric distribution based on $Q$, by writing $\lambda = Q/ \exp(\sqrt{\phi} z_p)$ in \eqref{eq:lg-pdf} and \eqref{eq:lg-cdf}. Then, we obtain the PDF and CDF of $Y$ as
\begin{equation}\label{eq:quant:cd}
 F_{Y}(y;Q,\phi)
=
G\!\left(\frac{1}{\phi} \left[ \log(y)-\log(\lambda) \right]^2 \right)=
G\!\left(\frac{1}{\phi} \left[ \log(y)-\log(Q)+\sqrt{\phi}\,z_{q} \right]^2 \right),\quad y>0,
\end{equation}
and
\begin{equation}\label{eq:quant:ft}
f_{Y}(y;Q,\phi)
=
\dfrac{\xi_{nc}}{\sqrt{\phi}\,y}\,
g\!\left(\frac{1}{\phi} \left[ \log(y)-\log(Q)+\sqrt{\phi}\,z_{q} \right]^2 \right), 
\quad y>0,
\end{equation}
respectively. Let us denote $Y\sim\textrm{QLS}(Q,\phi,g)$. Now, $V=\log(Y)\sim\textrm{QS}(\Psi,\phi,g)$, that is, $V$ has a symmetric distribution with PDF 
$f_{V}(v;\Psi,\phi)
=
{\xi_{nc}}/{\sqrt{\phi}}\,
g\!\left({\dfrac{1}{\phi}} \left[v-\Psi+\sqrt{\phi}\,z_{q}\right]^2\right),
\quad v\in\mathbb{R}$,
where $\Psi=\log(Q)$. Next, we present some properties of the log-symmetric distribution reparameterized by the quantile.

\begin{Proposition}\label{prop:01}
Let $Y\sim\textrm{QLS}(Q,\phi,g)$ . Then, $cY\sim\textrm{QLS}(cQ,\phi,g)$, with $c>0$.
\end{Proposition}
\begin{proof} If $Y\sim\textrm{QLS}(Q,\phi,g)$ with $c>0$, then 
\begin{eqnarray*}
 P(cY\leq{y})=P\left(Y\leq\frac{y}{c}\right)&=&
G\!\left(\frac{1}{\phi} \left[ \log(y)-\log(c)-\log(Q)+\sqrt{\phi}\,z_{q} \right]^2 \right)\\
&=&G\!\left(\frac{1}{\phi} \left[ \log(y)-\log(cQ)+\sqrt{\phi}\,z_{q} \right]^2 \right).
\end{eqnarray*}
\end{proof}

\begin{Proposition}\label{prop:02}
Let $Y\sim\textrm{QLS}(Q,\phi,g)$. Then, $Y^c\sim\textrm{QLS}(Q^c,c^2\phi,g)$, with $c>0$.
\end{Proposition}
\begin{proof} If $Y\sim\textrm{QLS}(Q,\phi,g)$ with $c > 0$,
Then,
\begin{eqnarray*}
 P(Y^c\leq{y})=P\left(Y\leq{y^{1/c}}\right)&=&
G\!\left(\frac{1}{\phi} \left[\frac{1}{c} \log(y)-\log(Q)+\sqrt{\phi}\,z_{q} \right]^2 \right)\\
&=&G\!\left(\frac{1}{c^2\phi} \left[ \log(y)-\log(Q^c)+\sqrt{c^2\phi}\,z_{q} \right]^2 \right).
\end{eqnarray*}
\end{proof}

\section{Log-symmetric quantile regression model}\label{sec:3}

Let $Y_1, \ldots, Y_n$, be a set of $n$ independent random variables such that $Y_i \sim\textrm{QLS}(Q_{i},\phi_{i},g)$, and ${\bf y}=(y_1,\ldots,y_n)^{\top}$ be the observations on $Y_1, \ldots, Y_n$. We then define
\begin{equation}\label{eq:logsymreg:quant:01}
Y_{i}=Q_{i}\,\epsilon_{i}^{\sqrt{\phi_i}},  \quad \epsilon_{i} \sim \textrm{QLS}(1, 1, g),
\end{equation}
which implies
\begin{eqnarray}\label{eq:logsymreg:quant:02}\nonumber
V_i
=
\log(Y_{i})
&=& 
\log(Q_{i})+\sqrt{\phi_i}\log(\epsilon_{i})\\ 
&=&\Psi_{i}+\sqrt{\phi_i}\varepsilon_{i} 
\quad i = 1, \ldots, n,
\end{eqnarray}
such that $V_{i} \sim \textrm{QS}(\Psi_{i},\phi_i,g)$, $\varepsilon_{i}=\log(\epsilon_{i})\sim \textrm{QS}(0,1,g)$,
$\log(Q_{i})=\Psi_{i}=\bm{x}_i^\top\bm \beta$ and
$\log(\phi_i) =\bm{w}^{\top}_{i}\bm{\tau}$, where $\bm{\beta} =(\beta_0,\ldots,\beta_{k})^\top$ and
$\bm{\tau}=(\tau_0,\ldots,{\tau_{l}})^\top$ are the vectors of unknown parameters to be estimated, 
${\bm{x}}^{\top}_{i}= (1,x_{i1},\ldots, x_{ik})^\top$ and
${\bm{w}}^{\top}_{i} = (1,w_{i1}, \ldots, w_{il})^\top$ 
are the sets of $k$ and $l$ covariates, respectively,

The estimation of the model parameters is performed using the maximum likelihood method. The log-likelihood function (without {the constant term}) for $\bm\theta=(\bm\beta^{\top},\bm\tau^{\top})^{\top}$ is given by  
\begin{equation}\label{mleequation}
\ell(\bm \theta)=\sum_{i=1}^{n} \left\{ \log(g(z^2_i))-\dfrac{1}{2}\log(\phi_i) \right \},
\end{equation}
where $z_i=[\log(y_i)-\log(Q_i)+\sqrt{\phi_i}\,z_{q}]/ \sqrt{\phi_i}$, for $i=1,\ldots,n$, with $z_{q}$ as given in \eqref{eq:lg-qf}. To obtain the maximum likelihood estimates, one must maximize the log-likelihood function in \eqref{mleequation} by equating the score vector $\dot{\bm \ell}(\bm{\theta})$, which contains the first derivatives of the log-likelihood function, to zero, providing the likelihood equations. The score vector $\dot{\bm \ell}(\bm{\theta})$ can be divided into two components
\begin{equation}\label{eq:score}
\dot{\bm \ell}(\bm{\theta}) = (\dot{\bm \ell}_{\bm \eta}(\bm{\theta}), \dot{\bm \ell}_{\bm \phi}(\bm{\theta}))^{\top},
\end{equation}
where $\dot{\bm \ell}_{\bm \eta}(\bm{\theta}) = (\dot{\bm \ell}_{ \beta_0}(\bm{\theta}), \cdots, \dot{\bm \ell}_{ \beta_k}(\bm{\theta}))^\top$ and $\dot{\bm \ell}_{\bm \phi}(\bm{\theta}) = (\dot{\bm \ell}_{ \tau_0}(\bm{\theta}), \cdots, \dot{\bm \ell}_{ \tau_l}(\bm{\theta}))^{\top}$, with 
{
\begin{equation*}
\dot{\bm \ell}_{\beta_j}(\bm{\theta}) = \sum_{i = 1}^{n} v(z_i) z_i \frac{x_{ij}}{\sqrt{\phi_i}},
\quad j = 0, \ldots, k,
\end{equation*}
\begin{equation}
\dot{\bm \ell}_{\tau_j}(\bm{\theta}) = \frac{1}{2} \sum_{i = 1}^{n} {\dfrac{w_{ij}}{\phi_i}} \left[  v(z_i) z_i (z_i - z_p) - 1 \right]
, \quad j = 0, \ldots, l,  \nonumber
\end{equation}} 
$v(z_i) = - 2 {g'(z_i^2)} / {g(z_i^2)}$. They must be solved by an iterative procedure for non-linear optimization

The solution of the system of likelihood equations cannot be obtained analytically, which makes it necessary to use an iterative procedure for non-linear optimization, such as the Broyden-Fletcher-Goldfarb-Shanno (BFGS) quasi-Newton method; see \cite{nw:06} and \cite{lange2010}. {The BFGS method is implemented in the \texttt{R} software 
\citep{rcodeteam:20} by the \texttt{maxLik} package.} Inference for $\bm \theta$ can be based on the asymptotic distribution of $\widehat{\bm \theta}$. Under some regularity conditions \citep{eh:78}, we have
{
\begin{equation}\label{norm:asym}
\sqrt{n}\,[\widehat{{\bm \theta}} -{\bm \theta}] \xrightarrow{D} \textrm{N}_{2+k+l}(\bm{0}, \mathcal{ I}({\bm \theta})^{-1}),
\end{equation}}
as $n \to \infty$, where {$\xrightarrow{D}$} denotes ``convergence in distribution'', $\bm{0}$ is a $(2+k+l) \times 1$ vector of zeros and $\mathcal{ I}({\bm \theta})$ is the expected Fisher information matrix, which can be obtained by computing $\mathcal{I}(\bm{\theta}) = \mathrm{E}[ - \partial^2 \ell(\bm{\theta})/\partial\bm{\theta}\partial\bm{\theta}^\top]$. The diagonal elements of $\mathcal{I}(\bm{\theta})$ are given by {$\mathcal{ I}_{\bm \beta \bm \beta}({\bm \theta})$ and $ \mathcal{ I}_{\bm \tau \bm \tau}({\bm \theta})$}, where
{
\begin{equation*}
\mathcal{ I}_{\bm \beta \bm \beta}({\bm \theta}) = \sum_{i=1}^{n} \frac{x_{ij}^2}{{\phi_i}} d_g(\zeta), \quad j = 0, \ldots, k,
\end{equation*}
and
\begin{equation}
\mathcal{ I}_{\bm \tau \bm \tau}({\bm \theta}) = \frac{1}{4} \sum_{i = 1}^{n} {\left(\dfrac{w_{ij}}{\phi_i}\right)^2} f_g(\zeta), \quad j = 0, \ldots, l,
\nonumber
\end{equation}
with $d_g(\zeta) = \mathrm{E}[v^2(Z) Z^2]$ and $f_g(\zeta) = \mathrm{E}[(v(Z)Z(Z - z_p) - 1)^2]$, where $Z = [\log(Y)-\log(Q)+\sqrt{\phi}\,z_{q}]/ \sqrt{\phi} \sim \textrm{S}(0,1,g)$.}

Note that in \eqref{mleequation} the extra parameter $(\vartheta)$ is assumed to be fixed. The motivation to leave it fixed comes from the works of \cite{l:97} and \cite{kanoetal93}. The former work shows that robustness against outliers under the Student-$t$ model holds only if the degrees of freedom parameter is fixed instead of estimated by the maximum likelihood method. The latter work reports difficulties in estimating the extra parameter of the power-exponential; see \cite{vanegasp:16a}. Thus, two steps are executed to estimate the extra parameter: Step 1) consider a grid of values for $\vartheta$, $\vartheta_1,\vartheta_2,\ldots,\vartheta_k$ say, and compute the maximum likelihood estimates of $\bm\theta$ for each value of $\vartheta_i$, $i=1,\ldots,k$. Also compute the value of the log-likelihood function; Step 2) the final estimates of $\vartheta$ of $\bm\theta$ are those that maximize the log-likelihood function.

\subsection{Hypothesis testing}
\label{sec:hyptest}

Suppose that our interest lies in testing the hypothesis $H_0: {\bm \theta} = {\bm \theta_0}$ against the alternative hypothesis $H_1: {\bm \theta} \neq {\bm \theta_0}$, where $\bm\theta=(\bm\beta^{\top},\bm\tau^{\top})^{\top}$ is a vector of parameters of interest.
In the literature on hypothesis testing for regression models, the Wald \citep{w:47}, score \citep{rao:48} and likelihood ratio \citep{wk:38} tests are widely used to test these hypotheses. 
The Wald, score and likelihood ratio statistics, denoted by $S_W$, $S_R$ e $S_{LR}$ are given respectively by:
\begin{eqnarray}
S_W & = & (\widehat{\bm \theta} - {\bm \theta_0})^{\top} \mathcal{ J}(\widehat{\bm \theta}) (\widehat{\bm \theta} - {\bm \theta_0}),
\nonumber \\
S_R & = & \dot{\bm \ell}(\bm{\theta_0})^{\top} \mathcal{ I}({\bm \theta_0})^{-1} \dot{\bm \ell}(\bm{\theta_0}), \nonumber \\
S_{LR} & = & -2 \left[  \ell(\bm \theta_0) - \ell(\widehat{\bm \theta}) \right], \nonumber
\end{eqnarray}
where $\mathcal{ J}(\bm{\theta}) = - \left\lbrace \partial^2 \ell(\bm{\theta}) / \partial \bm{\theta} \partial \bm{\theta}^\top \right\rbrace $ is the observed Fisher information. \cite{rao2005} suggests a modification to the $S_R$ test by replacing $\mathcal{I}(\bm{\theta_0})$ with 
$\mathcal{J}(\bm{\theta_0})$, which simplifies the calculation of the score statistic mainly under conditions in which the computation of $\mathcal{I}(\bm{\theta})$ is very complicated.

A prominent test with properties similar to those previously presented was developed by \cite{te:02}, the gradient test. It is based on a modification of the $S_R$ statistic. The advantage of the gradient statistic, denoted by $S_T$, is its {easy} of calculation, since just like the $S_{LR}$ statistic, there is no need to obtain the Fisher information matrix. The statistic $S_T$ is given by
\begin{equation*}\label{eq:gradient}
S_T = \dot{\bm \ell}(\bm{\theta_0})^{\top} (\widehat{\bm \theta} -\bm \theta_0).
\end{equation*}


The statistics $S_W$, $S_{LR}$, $S_R$ and $S_T$ are asymptotically equivalent, that is, these statistics, under $H_0$ and as $n\rightarrow\infty$, converge in distribution to $\chi^2_r$, where $r$ is the number of parameters set in $H_0$. We reject $H_0$ at nominal level $\alpha$ if the test statistic is larger than $C_{r,1-\alpha}$, the upper $\alpha$ quantile of the $\chi^2_r$ distribution.

\subsection{Information criteria}
In this work, the Akaike (AIC), Bayesian (BIC) and corrected Akaike (AICc) information criteria will be used to select the best models.
The AIC \citep{akaike1974} starts from the assumption of the existence of an unknown ``true model'', in such a way that, for a group of candidate models, the best model is the one that presents the smaller divergence in relation to the ``true model''. By taking the log-likelihood function at $\bm{\widehat{\theta}}$, $\ell(\bm{\widehat{\theta}})$, as an estimate of the divergence, penalized by the number  parameters $p$ of $\bm{\widehat{\theta}}$, the AIC is given by
\begin{equation}\label{eq:aic}
\mathrm{AIC} = -2 \ell(\bm{\widehat{\theta}}) + 2p.
\end{equation}
\cite{sakamoto1986} and \cite{sugiura1978} suggest that AIC can perform poorly in the face of a large number of parameters, which, in turn, can lead to a wrong conclusion regarding the best model. With the aim to circumvent the bias caused by this issue, \cite{bozdogan1987} proposed a correction to the AIC, denoted by AICc, which is given by
\begin{equation}\label{eq:aicc}
\mathrm{AICc} = \mathrm{AIC} + \frac{2p(p + 1)}{n - p -1}, 
\end{equation}
where $n$ is the sample size.

The BIC was proposed by  \cite{schwarz1978} and assumes the existence of a ``true model''. This criterion is defined as the statistic that maximizes the probability of choosing the true model, among the postulated models. In addition to penalizing the number of parameters $p$, the BIC also takes into account the sample size $n$ and is given by
\begin{equation}\label{eq:bic}
\mathrm{BIC} = -2 \ell(\bm{\widehat{\theta}}) + p\log(n). \nonumber
\end{equation}
For the three information criteria presented, the best model, among the candidate models, is the one with the lowest value of the information criterion.

\subsection{Residuals}
Residual analysis is an important tool to assess goodness of fit and departures from the assumptions of the regression model. We consider two types of residuals, the generalized Cox-Snell (GCS) and randomized quantile (RQ) residuals. The GCS \citep{cs:68} residual is given by
\begin{equation}\label{eq:res-cs}
\widehat{r_i}^{\text{GCS}} = - \log (  1 - F_Y(y_i;\widehat{Q}_i,\widehat{\phi}_i) ), \quad i = 1,2,\ldots,n, 
\end{equation}
where $F_Y(y_i;\widehat{Q}_i,\widehat{\phi}_i)$ is given in \eqref{eq:quant:cd}. {If the model is correctly specified, the GCS residual is asymptotically standard exponential distributed, EXP(1) in short; see \cite{b:10}}. The RQ residual is given by
\begin{equation}\label{eq:res-qr}
\widehat{r_i}^{\text{{RQ}}} = \Phi^{-1} ( F_Y(y_i;\widehat{Q}_i,\widehat{\phi}_i) ), \quad i = 1,2,\ldots,n, 
\end{equation}
where $\Phi(\cdot)^{-1}$ is the inverse function of the standard normal CDF; {see \cite{ds:96}}. {The RQ residual follows approximately a standard normal distribution, that is, N(0,1)}, if the model is correctly specified whatever the specification of the model is. For both residuals, the distribution assumption can be assessed through graphical techniques and descriptive statistics.

\section{Monte Carlo simulation studies}\label{sec:4}

In this section we will present the results of two Monte Carlo simulation studies for the log-symmetric quantile regression models. In the first simulation, the performance of the maximum likelihood estimators, information criteria and residuals are evaluated. In the second simulation, it is evaluated the performance of the Wald ($S_W$), score ($S_R$), likelihood ratio ($S_{LR}$) and gradient ($ S_{T}$) tests. In both simulations, the scenario considers the following setting: sample sizes $ n\in \{50, 100, 200\}$, quantiles $q\in\{0.25, 0.5, 0.75\}$, with 5,000 Monte Carlo replications. The following log-symmetric distributions are considered: 
log-normal (log-NO),
log-Student-$t$ (log-$t$),
log-power-exponential (log-PE),
log-hyperbolic (log-HP),
log-slash (log-SL),
log-contaminated-normal (log-CN),
extended Birnbaum-Saunders (EBS), and 
extended Birnbaum-Saunders-$t$ (EBS-$t$). The values of the extra parameters for the distributions are
$\vartheta=3$ (log-$t$), 
$\vartheta=0.3$ (log-PE),
$\vartheta=2$ (log-HP), 
$\vartheta=4$ (log-SL),
$\vartheta_1=0.1$,$\vartheta_2=0.2$ (log-CN),
$\vartheta=0.5$ (EBS), and 
$\vartheta_1=0.5$,$\vartheta_2=3$ (EBS-$t$).

\subsection{Study 1: performance of estimators, information criteria and residuals} 

The data generating model is given by
$$Y_i = (\beta_0 + \beta_1 x_i) \epsilon_i^{\sqrt{\tau_0 + \tau_1 w_i}},\quad i = 1, 2, \ldots, n,$$ 
with $\beta_0 = 1.5$, $\beta_1 = 0.5$, $\tau_0 = 1$, $\tau_1 = 0.5$, and 
$\epsilon_i\sim\text{QLS}(1,1,g)$. The covariate values for $x_i$ and $w_i$ were obtained as Bernoulli(0.5) and Uniform(0,1) random draws, respectively. The performance and recovery of the maximum likelihood estimators are assessed by the estimates of the bias and mean squared error (MSE), which are given, respectively, by
$$\textrm{Bias}(\widehat{\varphi}) = \frac{1}{M} \sum_{r = 1}^{M} (\widehat{\varphi}^{(r)} - \varphi), \quad \mathrm{MSE}(\widehat{\varphi}) = \frac{1}{M} \sum_{r = 1}^{M} (\widehat{\varphi}^{(r)} - \varphi)^2,$$
where $\varphi$ and $\widehat{\varphi}^{(r)}$ are the true value of the parameter and its respective $r$th estimate, while $M$ is the number of Monte Carlo replicas. Furthermore, we also compute the coverage probability (CP) of 95\% confidence interval (CI) based on the asymptotic normality \eqref{norm:asym}. The CP is computed as
$$\mathrm{CP}(\widehat{\varphi}) = \frac{1}{M} \sum_{r = 1}^{M} I(\varphi \in [\widehat{\varphi}_L^{(r)},\widehat{\varphi}_U^{(r)}]),$$
where $I(\cdot)$ is an indicator function such that $\varphi$ belongs to the $r$th interval $[\widehat{\varphi}_L^{(r)},\widehat{\varphi}_U^{(r)}]$, with 
$\widehat{\varphi}_L^{(r)}$ and $\widehat{\varphi}_U^{(r)}$ being the $r$th upper and lower limit estimates of the 95\% CI.

In addition to analyzing the performance of the maximum likelihood estimators, the AIC, BIC and AICc as well as the GCS and RQ residuals were also evaluated for the proposed quantile regression models. The idea behind the evaluation of the information criteria is in the ability of the AIC, BIC and AICc to choose the correct log-symmetric distribution. For example, if we generate samples from the log-normal quantile regression model, we should compute, for each Monte Carlo replica, the AIC, BIC and AICc based on the log-normal distribution and all other distributions. Then, the success rate with which these criteria choose the correct log-symmetric distribution is computed. This simulation exercise provides a way to assess the discriminatory capacity of the information criteria. The analysis regarding the performance of the GCS and RQ residuals is based on the empirical distribution of residuals. In this case, for each Monte Carlo replica, the GCS and RQ residuals associated with a log-symmetric quantile regression model are computed as well as their corresponding descriptive statistics, such as mean, median, standard deviation (SD), coefficient of skewness (CS) and coefficient of kurtosis (CK). Then, the averages of each descriptive statistics are computed. Such averages are expected to be close to theoretical values. The steps of the first Monte Carlo simulation study are described in Algorithm \ref{alg:simulation}.

\begin{algorithm}
	\floatname{algorithm}{Algorithm}
	\caption{Steps for Study 1.}\label{alg:simulation}
	\begin{algorithmic}[1]
		\State Choose a log-symmetric quantile regression model based on some density generator $g(\cdot)$ from Table~\ref{tab:gfuncions} and set the values of the model parameters.
		\State Generate 5,000 samples based on the chosen model.
		\State {Estimate the model parameters using the maximum likelihood method for each sample.}
		\State Compute the bias, MSE and CP.   
		\State Compute, for each replica, the GCS and RQ residuals and the respective descriptive statistics: 
		mean, median, SD, CS and CK. Compute the means of the descriptive statistics obtained from the 5,000 replications.
		\State Compute, for each replica, the AIC, BIC and AICc values for the chosen model and for different models in Table~\ref{tab:gfuncions}. Compute the percentage that the AIC, BIC and AICc choose the correct model, that is, the model based on the data generating process.
	\end{algorithmic}
\end{algorithm}

Tables \ref{tab:01}-\ref{tab:08} contain the Monte Carlo simulation results based on the log-NO, log-$t$, log-PE, log-HP, log-SL, log-CN, EBS e EBS-$t$ distributions. The following sample statistics for the maximum likelihood estimates are reported: bias, MSE and CP. 
A look at the results in \ref{tab:01}-\ref{tab:08} allows us to conclude that, as expected, in general, as the sample size increases, the bias and MSE decrease, for all the distributions considered. Moreover, we observe that the CP approaches the 95\% nominal level as the sample size increases. For example, for samples of size $n = 50$, the CP is approximately 93\%, and for $n=200$, the CP is approximately 95 \%. With the results presented in these tables, we can also observe that the MSE of the estimator of $\tau_1$ is always larger for the regression model with $q=0.5$ than with $q = 0.25$ and $q = 0.75$, for all distributions analyzed. Finally, we note that the MSEs of the estimators of $\beta_1$ and $\tau_1$ are always larger than the estimators associated with the other coefficients.

\begin{table}[H]
\footnotesize
	\centering
	\caption{\small Bias, MSE and CP from simulated data in the log-NO quantile regression model.}
		\adjustbox{max height=\dimexpr\textheight-3.5cm\relax,
		max width=\textwidth}{
	\begin{tabular}{cccccccccccccc}
	\toprule
	\multirow{2}{*}{$q$}& & \multicolumn{3}{c}{$n = 50$} & &\multicolumn{3}{c}{$n = 100$} & &\multicolumn{3}{c}{$n = 200$}\\
 	\cline{3-5} \cline{7-9} \cline{11-13}
	& & Bias & MSE & CP & & Bias & MSE & CP & & Bias & MSE & CP\\
	\hline
	\multirow{4}{*}{0.25}& $\beta_0$ & 0.04494 & 0.21187 & 0.9378 & & 0.02548 & 0.07570 & 0.9408 & & 0.00639 & 0.03417 & 0.9490\\
	& $\beta_1$ & -0.00517 & 0.31041 & 0.9366 & & -0.01253 & 0.14261 & 0.9400 & & 0.00313 & 0.07205 & 0.9446\\
	& $\tau_0$ & -0.08922 & 0.19829 & 0.9376 & & -0.04883 & 0.07609 & 0.9418 & & -0.02215 & 0.03902 & 0.9502\\
	& $\tau_1$ & 0.01857 & 0.47784 & 0.9430 & & 0.01845 & 0.20340 & 0.9422 & & 0.00481 & 0.10433 & 0.9516\\
	\hline
	\multirow{4}{*}{0.5}& $\beta_0$ & 0.00802 & 0.19609 & 0.9330 & &  0.00623 & 0.06855 & 0.9398 & & -0.00258 & 0.03038 & 0.9498\\
& $\beta_1$ & -0.00565 & 0.31264 & 0.9358 & & -0.01274 & 0.14346 & 0.9384 & & 0.00266 & 0.07201 & 0.9480\\
& $\tau_0$ & -0.09039 & 0.23368 & 0.9332 & & -0.04718 & 0.08827 & 0.9386 & & -0.02093 & 0.04420 & 0.9522\\
& $\tau_1$ & 0.02074 & 0.59772 & 0.9374 & & 0.01492 & 0.25472 & 0.9432 & & 0.00232 & 0.12644 & 0.9540\\
	\hline
	\multirow{4}{*}{0.75}& $\beta_0$ & -0.03438 & 0.13609 & 0.9366 && -0.01331 & 0.07714 & 0.9396 & & -0.01170 & 0.03487 & 0.9424\\
& $\beta_1$ & 0.00155 & 0.31954 & 0.9336 & & -0.01247 & 0.14316 & 0.9398 & & 0.00253 & 0.07186 & 0.9488\\
& $\tau_0$ & -0.10746 & 0.19714 & 0.9360 & & -0.04470 & 0.07654 & 0.9384 & & -0.01996 & 0.03862 & 0.9512\\
& $\tau_1$ & 0.05292 & 0.47134 & 0.9426 & & 0.00976 & 0.20408 & 0.9424 & & 0.00033 & 0.10397 & 0.9496\\
	\bottomrule
	\end{tabular}}
\label{tab:01}
\end{table}

\begin{table}[H]
\footnotesize
	\centering
	\caption{\small Bias, MSE and CP from simulated data in the log-$t$ quantile regression model ($\vartheta=3$).}
	\adjustbox{max height=\dimexpr\textheight-3.5cm\relax,
		max width=\textwidth}{
		\begin{tabular}{cccccccccccccc}
			\toprule
			\multirow{2}{*}{$q$}& & \multicolumn{3}{c}{$n = 50$} & &\multicolumn{3}{c}{$n = 100$} & &\multicolumn{3}{c}{$n = 200$}\\
			\cline{3-5} \cline{7-9} \cline{11-13}
			& & Bias & MSE & CP & & Bias & MSE & CP & & Bias & MSE & CP\\
			\hline
			\multirow{4}{*}{0.25}& $\beta_0$ & 0.05723 & 0.35426 & 0.9226 && 0.02137 & 0.13608 & 0.9390 & &  0.01011 & 0.06332 & 0.9438\\
			& $\beta_1$ & -0.01033 & 0.48476 & 0.9340 & & 0.00568 & 0.21711 & 0.9390 & & -0.00399 & 0.10860 & 0.9496\\
			& $\tau_0$ & -0.09946 & 0.36506 & 0.9332 & & -0.05036 & 0.14223 & 0.9354 & & -0.01374 & 0.07637 & 0.9452\\
			& $\tau_1$ & 0.03662 & 0.84097 & 0.9414 & & 0.01999 & 0.36923 & 0.9360 & & -0.00971 & 0.19696 & 0.9446\\
			\hline
			\multirow{4}{*}{0.5} & $\beta_0$ & -0.00009 & 0.33458 & 0.9278 & & -0.00367 & 0.11682 & 0.9366 & & -0.00144 & 0.05375 & 0.9476\\
			& $\beta_1$ & 0.00682 & 0.51008 & 0.9288 & & 0.00620 & 0.21809 & 0.9392 & & -0.00339 & 0.10864 & 0.9478\\
			& $\tau_0$ & -0.08664 & 0.43086 & 0.9302 & & -0.05151 & 0.17611 & 0.9370 & & -0.01108 & 0.09616 & 0.9408\\
			& $\tau_1$ & 0.02129 & 1.16259 & 0.9358 & & 0.02734 & 0.51431 & 0.9414 & & -0.01276 & 0.27128 & 0.9450\\
			\hline
			\multirow{4}{*}{0.75} & $\beta_0$ & -0.04706 & 0.37799 & 0.9250 & & -0.02775 & 0.13861 & 0.9364 & & -0.01245 & 0.06492 & 0.9430\\
			& $\beta_1$ & 0.00673 & 0.50690 & 0.9298 & & 0.00665 & 0.21735 & 0.9396 & & -0.00313 & 0.10819 & 0.9498\\
			& $\tau_0$ & -0.09495 & 0.34629 & 0.9320 & & -0.05336 & 0.13944 & 0.9388 & & -0.01596 & 0.07606 & 0.9476\\
			& $\tau_1$ & 0.02797 & 0.83592 & 0.9356 & & 0.02692 & 0.35996 & 0.9440 & & -0.00512 & 0.19132 & 0.9482\\
			\bottomrule
	\end{tabular}}
\label{tab:02}
\end{table}

\begin{table}[H]
\footnotesize
	\centering
	\caption{\small Bias, MSE and CP from simulated data in the log-PE quantile regression model ($\vartheta=0.3$).}
	\adjustbox{max height=\dimexpr\textheight-3.5cm\relax,
		max width=\textwidth}{
		\begin{tabular}{cccccccccccccc}
			\toprule
			\multirow{2}{*}{$q$}& & \multicolumn{3}{c}{$n = 50$} & &\multicolumn{3}{c}{$n = 100$} & &\multicolumn{3}{c}{$n = 200$}\\
			\cline{3-5} \cline{7-9} \cline{11-13}
			& & Bias & MSE & CP & & Bias & MSE & CP & & Bias & MSE & CP\\
			\hline
			\multirow{4}{*}{0.25}& $\beta_0$ & 0.03967 & 0.37253 & 0.9258 && 0.02905 & 0.13351 & 0.9380 && 0.01162 & 0.07233 & 0.9388\\
			& $\beta_1$ & 0.00865 & 0.53300 & 0.9272 && 0.00046 & 0.22886 & 0.9434 && 0.00072 & 0.11441 & 0.9438\\
			& $\tau_0$ & -0.09984 & 0.23753 & 0.9344 && -0.04151 & 0.10144 & 0.9440 && -0.02117 & 0.04064 & 0.9500\\
			& $\tau_1$ & 0.03592 & 0.58700 & 0.9376 && 0.00376 & 0.27762 & 0.9430 && 0.00477 & 0.11494 & 0.9472\\
			\hline
			\multirow{4}{*}{0.5} & $\beta_0$ & -0.00349 & 0.20361 & 0.9296 && 0.00462 & 0.11724 & 0.9416 && 0.00002 & 0.06459 & 0.9392\\
			& $\beta_1$ & -0.00374 & 0.45932 & 0.9378 && 0.00065 & 0.22971 & 0.9406 && 0.00077 & 0.11494 & 0.9430\\
			& $\tau_0$ & -0.08881 & 0.29855 & 0.9330 && -0.04069 & 0.11968 & 0.9464 && -0.02006 & 0.04712 & 0.9472\\
			& $\tau_1$ & 0.01916 & 0.96609 & 0.9314 &&  0.00210 & 0.35393 & 0.9464 && 0.00228 & 0.14724 & 0.9422\\
			\hline
			\multirow{4}{*}{0.75} & $\beta_0$ & -0.05048 & 0.23064 & 0.9310 && -0.01718 & 0.12940 & 0.9330 && -0.00948 & 0.06530 & 0.9428\\
			& $\beta_1$ & -0.00121 & 0.45677 & 0.9412 && -0.00203 & 0.22307 & 0.9438 && -0.00268 & 0.11625 & 0.9378\\
			& $\tau_0$ & -0.07954 & 0.19601 & 0.9342 && -0.03317 & 0.09536 & 0.9364 && -0.01860 & 0.04493 & 0.9480\\
			& $\tau_1$ & 0.00446 & 0.54319 & 0.9334 && -0.01392 & 0.30320 & 0.9406 && -0.00088 & 0.13713 & 0.9478\\
			\bottomrule
	\end{tabular}}
\label{tab:03}
\end{table}

\begin{table}[H]
\footnotesize
	\centering
	\caption{\small Bias, MSE and CP from simulated data in the log-HP quantile regression model ($\vartheta=2$).}
	\adjustbox{max height=\dimexpr\textheight-3.5cm\relax,
		max width=\textwidth}{
		\begin{tabular}{cccccccccccccc}
			\toprule
			\multirow{2}{*}{$q$}& & \multicolumn{3}{c}{$n = 50$} & &\multicolumn{3}{c}{$n = 100$} & &\multicolumn{3}{c}{$n = 200$}\\
			\cline{3-5} \cline{7-9} \cline{11-13}
			& & Bias & MSE & CP & & Bias & MSE & CP & & Bias & MSE & CP\\
			\hline
			\multirow{4}{*}{0.25}& $\beta_0$ & 0.04341 & 0.12272 & 0.9244 && 0.01370 & 0.06625 & 0.9460 &&  0.00484 & 0.01944 & 0.9510\\
			& $\beta_1$ & -0.01192 & 0.25651 & 0.9228 &&  0.00355 & 0.11393 & 0.9446 && 0.00100 & 0.03268 & 0.9546\\
			& $\tau_0$ & -0.09152 & 0.26537 & 0.9386 && -0.04850 & 0.09725 & 0.9414 && -0.01914 & 0.04941 & 0.9456\\
			& $\tau_1$ & 0.02235 & 0.78818 & 0.9368 && 0.01685 & 0.27689 & 0.9466 && 0.00534 & 0.14225 & 0.9464\\
			\hline
			\multirow{4}{*}{0.5} & $\beta_0$ &0.00800 & 0.11653 & 0.9250 && -0.00464 & 0.03579 & 0.9358 && 0.00132 & 0.02724 & 0.9534\\
			& $\beta_1$ & -0.00928 & 0.24365 & 0.9308 && 0.00624 & 0.07025 & 0.9374 && -0.00139 & 0.05645 & 0.9488\\
			& $\tau_0$ & -0.09105 & 0.23209 & 0.9356 && -0.04459 & 0.11275 & 0.9422 && -0.02084 & 0.05383 & 0.9480\\
			& $\tau_1$ & 0.02449 & 0.60907 & 0.9388 && 0.01181 & 0.31179 & 0.9470 && 0.00564 & 0.16596 & 0.9520\\
			\hline
			\multirow{4}{*}{0.75} & $\beta_0$ & -0.02912 & 0.14252 & 0.9330 && -0.01710 & 0.04095 & 0.9362 && -0.00667 & 0.03108 & 0.9514\\
			& $\beta_1$ &-0.00218 & 0.22892 & 0.9400 && 0.00600 & 0.07011 & 0.9384 && -0.00112 & 0.05636 & 0.9492\\
			& $\tau_0$ & -0.08729 & 0.20085 & 0.9366 && -0.04657 & 0.09803 & 0.9408 && -0.01984 & 0.04618 & 0.9464\\
			& $\tau_1$ & 0.01938 & 0.56486 & 0.9388 && 0.01505 & 0.24663 & 0.9444 && 0.00288 & 0.13103 & 0.9528\\
			\bottomrule
	\end{tabular}}
\label{tab:04}
\end{table}

\begin{table}[H]
\footnotesize
	\centering
	\caption{\small Bias, MSE and CP from simulated data in the log-SL quantile regression model ($\vartheta=4$).}
	\adjustbox{max height=\dimexpr\textheight-3.5cm\relax,
		max width=\textwidth}{
		\begin{tabular}{cccccccccccccc}
			\toprule
			\multirow{2}{*}{$q$}& & \multicolumn{3}{c}{$n = 50$} & &\multicolumn{3}{c}{$n = 100$} & &\multicolumn{3}{c}{$n = 200$}\\
			\cline{3-5} \cline{7-9} \cline{11-13}
			& & Bias & MSE & CP & & Bias & MSE & CP & & Bias & MSE & CP\\
			\hline
			\multirow{4}{*}{0.25}& $\beta_0$ & 0.04513 & 0.18733 & 0.9360 && 0.02660 & 0.09976 & 0.9406 && 0.00591 & 0.05687 & 0.9448\\
			& $\beta_1$ & -0.01673 & 0.38210 & 0.9398 && -0.01562 & 0.18943 & 0.9394 && 0.00751 & 0.10493 & 0.9444\\
			& $\tau_0$ & -0.08512 & 0.17147 & 0.9306 && -0.04243 & 0.08096 & 0.9410 && -0.01660 & 0.04662 & 0.9408\\
			& $\tau_1$ & 0.01686 & 0.46471 & 0.9350 && 0.01067 & 0.23064 & 0.9442 && 0.00140 & 0.12478 & 0.9434\\
			\hline
			\multirow{4}{*}{0.5} & $\beta_0$ & 0.00024 & 0.15159 & 0.9354 && -0.00169 & 0.11140 & 0.9456 && -0.00082 & 0.04713 & 0.9510\\
			& $\beta_1$ & -0.00578 & 0.39991 & 0.9350 && 0.00150 & 0.20521 & 0.9466 && 0.00134 & 0.09196 & 0.9462\\
			& $\tau_0$ & -0.09047 & 0.19744 & 0.9228 && -0.04150 & 0.11226 & 0.9404 && -0.01601 & 0.04652 & 0.9490\\
			& $\tau_1$ & 0.02633 & 0.57069 & 0.9280 &&  0.01002 & 0.32505 & 0.9428 && -0.00529 & 0.15561 & 0.9528\\
			\hline
			\multirow{4}{*}{0.75} & $\beta_0$ & -0.04315 & 0.18198 & 0.9320 && -0.01470 & 0.11162 & 0.9434 && -0.00857 & 0.05156 & 0.9452\\
			& $\beta_1$ & 0.00021 & 0.40340 & 0.9322 && -0.01007 & 0.18708 & 0.9460 && -0.00455 & 0.09277 & 0.9464\\
			& $\tau_0$ & -0.09574 & 0.18798 & 0.9344 && -0.03966 & 0.08969 & 0.9430 && -0.02439 & 0.04361 & 0.9416\\
			& $\tau_1$ & 0.03542 & 0.46394 & 0.9414 && 0.00483 & 0.25233 & 0.9408 && 0.01075 & 0.11783 & 0.9452\\
			\bottomrule
	\end{tabular}}
\label{tab:05}
\end{table}

\begin{table}[H]
\footnotesize
	\centering
	\caption{\small Bias, MSE and CP from simulated data in the log-CN quantile regression model ($\vartheta_1=0.1,\vartheta_1=0.2$).}
	\adjustbox{max height=\dimexpr\textheight-3.5cm\relax,
		max width=\textwidth}{
		\begin{tabular}{cccccccccccccc}
			\toprule
			\multirow{2}{*}{$q$}& & \multicolumn{3}{c}{$n = 50$} & &\multicolumn{3}{c}{$n = 100$} & &\multicolumn{3}{c}{$n = 200$}\\
			\cline{3-5} \cline{7-9} \cline{11-13}
			& & Bias & MSE & CP & & Bias & MSE & CP & & Bias & MSE & CP\\
			\hline
			\multirow{4}{*}{0.25}& $\beta_0$ & 0.03428 & 0.29328 & 0.9298 && 0.01734 & 0.10322 & 0.9394 && 0.00958 & 0.04582 & 0.9476\\
			& $\beta_1$ & -0.00047 & 0.40343 & 0.9352 && 0.00098 & 0.17773 & 0.9370 && -0.00548 & 0.08839 & 0.9462\\
			& $\tau_0$ & -0.08920 & 0.19994 & 0.9344 && -0.04330 & 0.09862 & 0.9404 && -0.01747 & 0.05315 & 0.9486\\
			& $\tau_1$ & 0.03943 & 0.52512 & 0.9388 && 0.01604 & 0.26394 & 0.9446 && 0.00265 & 0.14103 & 0.9508\\
			\hline
			\multirow{4}{*}{0.5} & $\beta_0$ & 0.00152 & 0.14127 & 0.9348 && -0.00149 & 0.08951 & 0.9416 && 0.00158 & 0.03971 & 0.9514\\
			& $\beta_1$ & -0.00866 & 0.37150 & 0.9310 &&  0.00227 & 0.17737 & 0.9384 && -0.00565 & 0.08850 & 0.9454\\
			& $\tau_0$ & -0.07597 & 0.26109 & 0.9344 && -0.04236 & 0.11553 & 0.9362 && -0.01608 & 0.06270 & 0.9518\\
			& $\tau_1$ & 0.01602 & 0.81035 & 0.9380 && 0.01602 & 0.34129 & 0.9388 && 0.00080 & 0.18008 & 0.9462\\
			\hline
			\multirow{4}{*}{0.75} & $\beta_0$ & -0.03343 & 0.17552 & 0.9368 && -0.02022 & 0.10162 & 0.9374 && -0.00631 & 0.04599 & 0.9492\\
			& $\beta_1$ & -0.00868 & 0.36991 & 0.9374 &&  0.00286 & 0.17763 & 0.9372 && -0.00573 & 0.08824 & 0.9458\\
			& $\tau_0$ & -0.07384 & 0.24572 & 0.9282 && -0.04242 & 0.09721 & 0.9408 && -0.01600 & 0.05257 & 0.9516\\
			& $\tau_1$ & 0.00432 & 0.66144 & 0.9344 && 0.01407 & 0.26288 & 0.9384 && -0.00014 & 0.13954 & 0.9476\\
			\bottomrule
	\end{tabular}}
\label{tab:06}
\end{table}

\begin{table}[H]
\footnotesize
	\centering
	\caption{\small Bias, MSE and CP from simulated data in the EBS quantile regression model ($\vartheta=0.5$).}
	\adjustbox{max height=\dimexpr\textheight-3.5cm\relax,
		max width=\textwidth}{
		\begin{tabular}{cccccccccccccc}
			\toprule
			\multirow{2}{*}{$q$}& & \multicolumn{3}{c}{$n = 50$} & &\multicolumn{3}{c}{$n = 100$} & &\multicolumn{3}{c}{$n = 200$}\\
			\cline{3-5} \cline{7-9} \cline{11-13}
			& & Bias & MSE & CP & & Bias & MSE & CP & & Bias & MSE & CP\\
			\hline
			\multirow{4}{*}{0.25}& $\beta_0$ & 0.01121 & 0.00881 & 0.9350 && 0.00491 & 0.00430 & 0.9414 && 0.00248 & 0.00226 & 0.9418\\
			& $\beta_1$ & -0.00343 & 0.01558 & 0.9380 && -0.00037 & 0.00810 & 0.9488 && 0.00062 & 0.00413 & 0.9422\\
			& $\tau_0$ & -0.09553 & 0.14441 & 0.9276 && -0.05220 & 0.07498 & 0.9378 && -0.02107 & 0.03415 & 0.9466\\
			& $\tau_1$ & 0.01635 & 0.45149 & 0.9368 && 0.02227 & 0.20758 & 0.9460 && 0.00437 & 0.10093 & 0.9470\\
			\hline
			\multirow{4}{*}{0.5} & $\beta_0$ & -0.00116 & 0.00698 & 0.9342 && -0.00022 & 0.00383 & 0.9404 && 0.00015 & 0.00210 & 0.9444\\
			& $\beta_1$ & 0.00448 & 0.01744 & 0.9310 && 0.00009 & 0.00805 & 0.9452 && 0.00098 & 0.00414 & 0.9446\\
			& $\tau_0$ & -0.09376 & 0.19738 & 0.9332 && -0.05998 & 0.08609 & 0.9440 && -0.02061 & 0.03796 & 0.9482\\
			& $\tau_1$ & 0.01882 & 0.73957 & 0.9316 && 0.03481 & 0.24082 & 0.9472 && 0.00267 & 0.12472 & 0.9474\\
			\hline
			\multirow{4}{*}{0.75} & $\beta_0$ & -0.01071 & 0.01067 & 0.9334 && -0.00507 & 0.00457 & 0.9374 && -0.00294 & 0.00236 & 0.9418\\
			& $\beta_1$ & 0.00300 & 0.01839 & 0.9264 && 0.00085 & 0.00849 & 0.9400 && 0.00081 & 0.00420 & 0.9428\\
			& $\tau_0$ & -0.09893 & 0.23746 & 0.9292 && -0.04283 & 0.08525 & 0.9330 && -0.02446 & 0.02988 & 0.9416\\
			& $\tau_1$ & 0.02173 & 0.60652 & 0.9346 && 0.01077 & 0.22118 & 0.9382 && 0.00422 & 0.08813 & 0.9404\\
			\bottomrule
	\end{tabular}}
\label{tab:07}
\end{table}

\begin{table}[H]
\footnotesize
	\centering
	\caption{\small Bias, MSE and CP from simulated data in the EBS-$t$ quantile regression model ($\vartheta_1=0.5,\vartheta_1=3$).}
	\adjustbox{max height=\dimexpr\textheight-3.5cm\relax,
		max width=\textwidth}{
		\begin{tabular}{cccccccccccccc}
			\toprule
			\multirow{2}{*}{$q$}& & \multicolumn{3}{c}{$n = 50$} & &\multicolumn{3}{c}{$n = 100$} & &\multicolumn{3}{c}{$n = 200$}\\
			\cline{3-5} \cline{7-9} \cline{11-13}
			& & Bias & MSE & CP & & Bias & MSE & CP & & Bias & MSE & CP\\
			\hline
			\multirow{4}{*}{0.25}& $\beta_0$ & 0.01005 & 0.01464 & 0.9318 && 0.00464 & 0.00826 & 0.9382 && 0.00212 & 0.00375 & 0.9474\\
			& $\beta_1$ & 0.00289 & 0.02661 & 0.9328 && 0.00164 & 0.01367 & 0.9390 && 0.00096 & 0.00662 & 0.9470\\
			& $\tau_0$ & -0.09564 & 0.23770 & 0.9374 && -0.04201 & 0.13874 & 0.9462 && -0.01412 & 0.06065 & 0.9464\\
			& $\tau_1$ & 0.03524 & 0.72409 & 0.9428 && 0.00279 & 0.41689 & 0.9496 && -0.00590 & 0.16126 & 0.9524\\
			\hline
			\multirow{4}{*}{0.5} & $\beta_0$ & -0.00034 & 0.01362 & 0.9308 && 0.00017 & 0.00773 & 0.9420 && 0.00100 & 0.00334 & 0.9476\\
			& $\beta_1$ & -0.00056 & 0.02825 & 0.9296 && -0.00122 & 0.01328 & 0.9426 && -0.00125 & 0.00654 & 0.9446\\
			& $\tau_0$ & -0.08863 & 0.37166 & 0.9340 && -0.03873 & 0.15450 & 0.9360 && -0.02320 & 0.07412 & 0.9516\\
			& $\tau_1$ & 0.01687 & 1.01367 & 0.9300 && 0.00324 & 0.44647 & 0.9404 && 0.00502 & 0.21397 & 0.9514\\
			\hline
			\multirow{4}{*}{0.75} & $\beta_0$ & -0.01242 & 0.01326 & 0.9346 && -0.00528 & 0.00779 & 0.9412 && -0.00219 & 0.00388 & 0.9408\\
			& $\beta_1$ & 0.00672 & 0.02939 & 0.9320 && 0.00133 & 0.01348 & 0.9406 && -0.00126 & 0.00669 & 0.9422\\
			& $\tau_0$ & -0.09278 & 0.32558 & 0.9374 && -0.03276 & 0.13614 & 0.9410 && -0.01992 & 0.05917 & 0.9468\\
			& $\tau_1$ & 0.03374 & 0.91462 & 0.9380 && -0.00390 & 0.33330 & 0.9426 && -0.00047 & 0.16003 & 0.9482\\
			\bottomrule
	\end{tabular}}
\label{tab:08}
\end{table}


Table \ref{tab:11} presents the simulation results for the AIC, BIC and AICc. This table reports the success rates according to these criteria, that is, the percentage of times that the criteria correctly chose the correct model (correct distribution). From this table, we observe that the success rates tend to increase with the increase in the sample size $n$, as expected, and that these rates are larger for the log-NO and EBS distributions.

	\begin{table}[H]
	\footnotesize
		\centering
		\caption{\small {Success rates from simulated data in the log-symmetric quantile regression}.}
		\adjustbox{max height=\dimexpr\textheight-3.5cm\relax,
			max width=\textwidth}{
			\begin{tabular}{llcccccccccccc}
				\toprule
				& & \multicolumn{3}{c}{$n = 50$} & &\multicolumn{3}{c}{$n = 100$} & &\multicolumn{3}{c}{$n = 200$}\\
				\cline{3-5} \cline{7-9} \cline{11-13}
				$q$& Model & AIC & BIC & AICc & & AIC & BIC & AICc & & AIC & BIC & AICc\\
				\hline
				\multirow{8}{*}{0.25}& Log-NO & 0.7058 & 0.7058 & 0.7058 && 0.7108 & 0.7108 & 0.7108 && 0.732 & 0.732 & 0.732 \\
				& Log-$t(3)$ & 0.3464 & 0.3464 & 0.3464 && 0.5226 & 0.5226 & 0.5226 && 0.6484 & 0.6484 & 0.6484 \\ 
				& Log-PE(0.3) & 0.2526 & 0.2526 & 0.2526 && 0.3402 & 0.3402 & 0.3402 && 0.455 & 0.455 & 0.455 \\
				& Log-HP(2) & 0.0630 & 0.0630 & 0.0630 && 0.1200 & 0.1200 & 0.1200 && 0.2272 & 0.2272 & 0.2272 \\
				& Log-SL(4) & 0.0348 & 0.0348 & 0.0348 && 0.0934 & 0.0934 & 0.0934 && 0.2024 & 0.2024 & 0.2024 \\
				& Log-CN(0.1,0.2) & 0.2300 & 0.2300& 0.2300 && 0.3542 & 0.3542 & 0.3542 && 0.5100 & 0.5100 & 0.5100 \\
				& EBS(0.5) & 0.7386 & 0.7386 & 0.7386 && 0.7206 & 0.7206 & 0.7206 && 0.7230 & 0.7230 & 0.7230 \\
				& EBS-$t(0.5,3)$ & 0.2068 & 0.2068 &	0.2068 && 0.3242 & 0.3242 & 0.3242 && 0.4790 & 0.4790 & 0.4790 \\
				\hline
				\multirow{8}{*}{0.5}& Log-NO & 0.7208 & 0.7208 & 	0.7208 && 0.7242 & 0.7242 &	0.7242 && 0.7366 & 0.7366 & 0.7366 \\
				& Log-$t(3)$ & 0.4008 & 0.4008 & 0.4008 && 0.5112 & 0.5112 & 0.5112 && 0.6400 & 0.6400 & 0.6400 \\
				& Log-PE(0.3) & 0.2474 & 0.2474 & 0.2474 && 0.3514 & 0.3514 & 0.3514 && 0.4430 & 0.4430 & 0.4430 \\
				& Log-HP(2) & 0.0618 & 0.0618 & 0.0618 && 0.1206 & 0.1206 & 0.1206 && 0.2210 & 0.2210 & 0.2210 \\
				& Log-SL(4) & 0.0308 & 0.0308 & 0.0308 && 0.0928 & 0.0928 & 0.0928 && 0.2038 & 0.2038 & 0.2038 \\
				& Log-CN(0.1,0.2) & 0.2280 & 0.2280 &	0.2280 && 0.3562 & 0.3562 & 0.3562 && 0.5084 &	0.5084 & 0.5084 \\
				& EBS(0.5) & 0.7462 & 0.7462 & 0.7462 && 0.7298 & 0.7298 & 0.7298 && 0.7280 & 0.7280 & 0.7280 \\
				& EBS-$t(0.5,3)$ & 0.1990 & 0.1990 & 0.1990 && 0.3248 & 0.3248 & 0.3248 && 0.4774 & 0.4774 & 0.4774 \\
				\hline
				\multirow{8}{*}{0.75} & Log-NO & 0.7096 & 0.7096 & 0.7096 && 0.7122 & 0.7122 &	0.7122 && 0.7306 & 0.7306 & 0.7306 \\
				& Log-$t(3)$ & 0.4176 & 0.4176 & 0.4176 && 0.5212 & 0.5212 & 0.5212 && 0.6452 & 0.6452 & 0.6452 \\
				& Log-PE(0.3) & 0.2426 & 0.2426 & 0.2426 && 0.3538 & 0.3538 & 0.3538 && 0.4486 & 0.4486 & 0.4486 \\
				& Log-HP(2) & 0.0648 &	0.0648 & 0.0648 && 0.1252 & 0.1252 & 0.1252 && 0.2256 &	0.2256 & 0.2256 \\
				& Log-SL(4) & 0.0356 & 0.0356 & 0.0356 && 0.0938 & 0.0938 & 0.0938 && 0.2044 &	0.2044 & 0.2044 \\
				& Log-CN(0.1,0.2) & 0.2310 &	0.2310 &	0.2310 && 0.3574 & 0.3574 & 0.3574 && 0.5050 & 0.5050 & 0.5050 \\
				& EBS(0.5) & 0.7376 & 0.7376 & 0.7376 && 0.7294 & 0.7294 & 0.7294 && 0.7266 &	0.7266 & 0.7266 \\
				& EBS-$t(0.5,3)$ & 0.2050 & 0.2050 & 0.2050 && 0.3318 & 0.3318 & 0.3318 && 0.4750 & 0.4750 & 0.4750 \\
				\bottomrule
		\end{tabular}
	}
		\label{tab:11}
	\end{table}


Tables \ref{tab:09}-\ref{tab:10} present the simulation results for the GCS and RQ residuals. The objective here is to verify whether the GCS and RQ residuals behave according to their reference distributions. In this sense, Tables \ref{tab:09}-\ref{tab:10} show the mean, median, SD, CS and CK, whose values are expected to be 1, 0.69, 1, 2 and 6, respectively, for the GCS residual, and 0, 0, 1, 0 and 0, respectively, for the RQ residual. From \ref{tab:09}, we observe that the means, medians and SDs are close to 1, 0.69, 1, respectively, that is, the values of the reference EXP(1) distribution. Moreover, the values of the CS and CK approach, in general, the values of the reference EXP(1) distribution, as the sample size increases. From Table \ref{tab:10}, we note that the mean, median and SD values are very clore to 0, 0, 1, respectively, that is, the reference values of the N(0,1) distribution. In addition, as the sample size increases, the values of the CS and CK approach the values of the reference N(0,1) distribution.

\begin{table}[!ht]
\footnotesize
	\centering
	\caption{\small Summary statistics of the GCS residuals.}
	\adjustbox{max height=\dimexpr\textheight-3.5cm\relax,
		max width=\textwidth}{	
		\begin{tabular}{clcccccccc}
			\toprule
			& & \multicolumn{8}{c}{$n = 50$}\\
			\cline{3-10}
			$q$ & Statistic & Log-NO & Log-$t$ & Log-PE & Log-HP & Log-SL & Log-CN & EBS & EBS-$t$\\
			\hline
			\multirow{5}{*}{0.25}& Mean & 0.99936 & 0.99404 & 0.99644 & 0.99611 & 0.99758 & 1.07521 & 1.00123 & 0.99435\\
			& Median & 0.70226 & 0.69642 & 0.69922 & 0.69828 & 0.70149 & 0.63628 & 0.70325 & 0.69687\\
			& SD & 0.9818 & 0.98300 & 0.98114 & 0.98371 & 0.97898 & 1.26828 & 0.97875 & 0.98342\\
			& CS & 1.52807 & 1.56819 & 1.53842 & 1.55628 & 1.51692 & 1.40280 & 1.48478 & 1.56940\\
			& CK & 2.55102 & 2.75474 & 2.58906 & 2.68969 & 2.48426 & 1.70423 & 2.31197 & 2.76390\\ 
			\hline
			\multirow{5}{*}{0.5}& Mean & 1.00067 & 0.99685 & 0.99964 & 0.99897 & 0.99900 & 1.05825 & 1.00188 & 0.99961 \\
			& Median & 0.69720 & 0.69518 & 0.69591 & 0.69761 & 0.69643 & 0.59712 &  0.69827 &  0.69473\\
			& SD & 0.99083 & 0.98565 & 0.99004 & 0.98957 & 0.98663 & 1.30431 & 0.99038 & 0.99158\\
			& CS & 1.55564 & 1.55944 & 1.55711 & 1.56082 & 1.54231 & 1.83512 &  1.53950 & 1.57030 \\
			& CK & 2.65167 & 2.69073 & 2.66652 & 2.68922 & 2.59905 & 3.94936 & 2.57457 & 2.73967\\ 
			\hline
			\multirow{5}{*}{0.75}& Mean & 1.00208 & 1.00216 & 1.00253 & 1.00273 & 1.00069 & 0.67645 & 1.00189 & 1.00243\\
			& Median & 0.69355 & 0.69318 & 0.69286 & 0.69338 & 0.69325 & 0.24333 & 0.69429 & 0.69370 \\
			& SD & 1.00088 & 0.99531 & 0.99963 & 1.00020 & 0.99712 & 1.03971 & 0.99888 & 0.99773 \\
			& CS & 1.59439 & 1.57219 & 1.58905 & 1.59299 & 1.58222 & 2.08323 & 1.58600 & 1.58787 \\
			& CK & 2.82522 & 2.71662 & 2.79731 & 2.84244 & 2.77479 & 4.70738 & 2.82211 & 2.81158\\
			\bottomrule	
			& & \multicolumn{8}{c}{$n = 100$}\\
			\hline
			\multirow{5}{*}{0.25}& Mean & 0.99987 & 0.99697 & 0.99876 & 0.99849 & 0.99910 & 1.28750 & 1.00067 & 0.99611\\
			& Median & 0.69714 & 0.69533 & 0.69572 & 0.69557 & 0.69651 & 0.78760 & 0.69647 & 0.69531\\
			& SD & 0.99084 & 0.99236 & 0.99127 & 0.99113 & 0.99066 & 1.45686 & 0.98949 & 0.98831\\
			& CS & 1.72120 & 1.75991 & 1.73983 & 1.73707 & 1.73105 & 1.57111 & 1.69862 & 1.74351\\
			& CK & 3.73464 & 4.00959 & 3.85506 & 3.84957 & 3.83319 & 2.76641 & 3.60417 & 3.90985\\ 
			\hline
			\multirow{5}{*}{0.5}& Mean & 1.00059 & 0.99882 & 1.00014 & 0.99915 & 0.99954 & 1.03951 & 1.00105 & 0.99884\\
			& Median & 0.69532 & 0.69495 &  0.69383 & 0.69476 & 0.69447 & 0.52992 & 0.69412 &  0.69401\\
			& SD & 0.99504 & 0.99385 & 0.99500 & 0.99332 & 0.99527 & 1.32493 & 0.99442 & 0.99269\\
			& CS & 1.73446 & 1.75317 & 1.74640 & 1.74562 & 1.74992 & 1.84787 &  1.72423 & 1.74386\\
			& CK & 3.79463 & 3.96013 & 3.88064 & 3.91642 & 3.92758 & 3.90925 & 3.7512 & 3.87944\\
			\hline
			\multirow{5}{*}{0.75}& Mean & 1.00128 & 1.00145 & 1.00136 & 1.00056 & 1.00074 & 0.81053 & 1.00109 & 1.00069 \\
			& Median & 0.69339 & 0.69407 &  0.69192 & 0.69296 & 0.69286 & 0.32290 & 0.69245 & 0.69363\\
			& SD & 1.0008 & 0.99851 & 1.00049 & 0.99889 & 0.99978 & 1.18421 & 1.00061 &  0.99672\\
			& CS & 1.76260 & 1.75754 & 1.77238 & 1.76733 & 1.75966 & 2.24027 & 1.76164 & 1.74963\\
			& CK & 3.95368 & 3.95836 & 4.03496 & 4.03327 & 3.9553 & 6.07540 & 3.96037 & 3.90162\\
			\bottomrule	
			& & \multicolumn{8}{c}{$n = 200$}\\
			\hline
			\multirow{5}{*}{0.25}& Mean & 1.00018 & 0.99908 & 0.99960 & 0.99931 & 0.9992 & 1.28156 & 1.00047 & 0.99891 \\
			& Median & 0.69487 & 0.69412 & 0.69428 & 0.69466 & 0.69592 & 0.79174 & 0.69639 & 0.69376\\
			& SD & 0.99608 & 0.99675 & 0.99596 & 0.99618 & 0.99468 & 1.42813 & 0.99300 & 0.99496\\
			& CS & 1.84419 & 1.86845 & 1.85341 & 1.86476 & 1.85288 & 1.57580 & 1.82405 & 1.85192\\
			& CK & 4.61888 & 4.83472 & 4.69717 & 4.83238 & 4.73117 & 2.73989 & 4.49670 & 4.71112\\ 
			\hline
			\multirow{5}{*}{0.5}& Mean & 1.00044 & 0.99994 & 1.00031 & 0.99955 & 1.00009 & 1.02630 & 1.00051 & 0.99887\\
			& Median & 0.69378 & 0.69361 &  0.69327 & 0.69510 & 0.69403 & 0.52399 & 0.69500 & 0.69414\\
			& SD & 0.99827 & 0.99760 & 0.99773 & 0.99684 & 0.99801 & 1.29699 & 0.99688 & 0.99518\\
			& CS & 1.85139 & 1.86526 & 1.85437 & 1.86244 & 1.85947 & 1.87475 & 1.84482 & 1.85365\\
			& CK & 4.65578 & 4.80970 & 4.68798 & 4.78175 & 4.75484 & 4.06773 & 4.61053 & 4.71592\\
			\hline
			\multirow{5}{*}{0.75}& Media & 1.00060 & 1.00113 & 1.00091 & 1.00036 & 1.00060 & 0.79130 & 1.00037 & 1.00022\\
			& Median & 0.69238 & 0.69302 & 0.69282 & 0.69426 & 0.69341 & 0.30988 & 0.69382 & 0.69362\\
			& SD & 1.00109 & 0.99994 & 1.00088 & 0.99948 & 1.00057 & 1.15141 & 0.99890 & 0.99836\\
			& CS & 1.86760 & 1.86810 & 1.87219 & 1.87235 & 1.87140 & 2.27928 & 1.85969 & 1.86558\\
			& CK & 4.76094 & 4.81148 & 4.81338 & 4.83858 & 4.83023 & 6.37006 & 4.72970 & 4.79141\\
			\bottomrule	
		\end{tabular} 
	}
	\label{tab:09}
\end{table}

\begin{table}[!ht]
\footnotesize
	\centering
	\caption{\small Summary statistics of the RQ residuals.}
	\adjustbox{max height=\dimexpr\textheight-3.5cm\relax,
		max width=\textwidth}{
		\begin{tabular}{clcccccccc}
			\toprule
			& & \multicolumn{8}{c}{$n = 50$}\\
			\cline{3-10}
			$q$ & Estatística & Log-NO & Log-$t$ & Log-PE & Log-HP & Log-SL & Log-CN & EBS & EBS-$t$\\
			\hline
			\multirow{5}{*}{0.25}& Mean & -0.00038 & -0.00448 & -0.00285 & -0.00268 & -0.00067 & 0.18892 & 0.00080 & -0.00444\\
			& Median & 0.00653 & 0.00143 & 0.00479 & 0.00311 & 0.00579 & 0.21746 &  0.00711 & 0.00186\\
			& SD & 1.01013 & 1.00519 & 1.00885 & 1.00740 & 1.00780 & 1.16427 & 1.01208 & 1.00652\\
			& CS & -0.02458 & -0.01694 & -0.02598 & -0.01574 & -0.02771 & -0.08877 & -0.03768 & -0.01799\\
			& CK & -0.31632 & -0.29297 & -0.29107 & -0.29093 & -0.30447 & -0.58239 & -0.33039 & -0.28323\\ 
			\hline
			\multirow{5}{*}{0.5}& Mean & -0.00014 & -0.00099 & -0.00031 & 0.00008 & -0.00011 & -0.00200 & 0.00020 & 0.00108\\
			& Median & 0.00019 & -0.00042 &  0.00067 & 0.00219 & -0.00048 & -0.00075 &  0.00098 & 0.00102\\
			& SD & 1.00998 & 1.00367 & 1.00861 & 1.00662 & 1.00756 & 1.23045 & 1.01191 & 1.00508\\
			& CS & 0.00680 & -0.00422 & 0.00182 & 0.00465 & -0.00058 & 0.00357 & -0.00134 & 0.00789\\
			& CK & -0.33186 & -0.30871 & -0.30544 & -0.30779 & -0.32195 & -0.20582 & -0.32464 & -0.30257\\
			\hline
			\multirow{5}{*}{0.75} & Mean & 0.00035 & 0.00339 & 0.00197 & 0.00309 & 0.00064 & -0.17968 & -0.00062 & 0.00345\\
			& Median & -0.00439 & -0.00268 &  -0.00306 & -0.00301 & -0.00457 & -0.20295 & -0.00423 & -0.00217\\
			& SD & 1.01012 & 1.00532 & 1.00874 &  1.00713 & 1.00779 & 1.17251 & 1.01181 & 1.00606\\
			& CS & 0.03333 & 0.00835 & 0.02420 & 0.02689 & 0.02616 & 0.08275 & 0.02336 & 0.01545\\
			& CK & -0.31318 & -0.29367 & -0.29421 & -0.29202 & -0.30382 & -0.50849 & -0.30819 & -0.28281\\		
			\bottomrule	
			& & \multicolumn{8}{c}{$n = 100$}\\
			\hline
			\multirow{5}{*}{0.25}& Mean & -0.00005 & -0.00250 & -0.00087 & -0.00064 & -0.00012 & 0.23444 & 0.00024 & -0.00291\\
			& Median & 0.00262 & 0.00147 &  0.00193 & 0.00145 & 0.00210 & 0.26566 &  0.00159 & 0.00141\\
			& SD & 1.00501 & 1.00265 & 1.00447 & 1.00323 & 1.00392 & 1.25715 & 1.00610 & 1.00284\\
			& CS & -0.01025 & -0.00522 & -0.01251 & -0.00663 & -0.01102 & -0.10634 & -0.01571 & -0.01513\\
			& CK & -0.16536 & -0.15415 & -0.14267 & -0.15708 & -0.15116 & -0.39488 & -0.17582 & -0.14859\\
			\hline
			\multirow{5}{*}{0.5}& Mean & 0.00016 & -0.00017 & 0.00021 & -0.00024 & 0.00003 & -0.00104 & 0.00012 & -0.00021\\
			& Median & 0.00031 & 0.00092 & -0.00045 & 0.00027 & -0.00054 & -0.00022 &  -0.00142 & -0.00030\\
			& SD & 1.00498 & 1.0017 & 1.00435 & 1.00344 & 1.00357 & 1.26945 & 1.00602 & 1.00241\\
			& CS & 0.00452 & 0.00124 & 0.00027 & -0.00037 & 0.00317 & -0.00300 & 0.00026 & -0.00107\\
			& CK & -0.17381 & -0.16176 & -0.15378 & -0.15974 & -0.14998 & -0.42956 & -0.1760 & -0.16302\\
			\hline
			\multirow{5}{*}{0.75}& Mean & 0.00031 & 0.00199 & 0.00106 & 0.00076 & 0.00056 & -0.23798 & -0.00049 & 0.00134\\
			& Median & -0.00207 & -0.00015 & -0.00283 & -0.00195 & -0.00257 & -0.26895 &  -0.00353 & -0.00074\\
			& SD & 1.0051 & 1.00256 & 1.00443 &  1.00362 & 1.00398 & 1.26510 & 1.00597 & 1.00284\\
			& CS & 0.01961 & 0.00747 & 0.01308 & 0.01236 & 0.01712 & 0.10668 & 0.01872 & 0.00490\\
			& CK & -0.16354 & -0.15548 & -0.14104 & -0.14852 & -0.1537 & -0.36732 & -0.16617 & -0.15229\\
			\bottomrule	
			& & \multicolumn{8}{c}{$n = 200$}\\
			\hline
			\multirow{5}{*}{0.25}& Mean & 0.00016 & -0.00058 & -0.00024 & -0.00035 & -0.0003 & 0.23705 & 0.00052 & -0.00103\\
			& Median & 0.00098 & 0.00060 & 0.00079 & 0.00107 & 0.00243 & 0.26940 &  0.00278 & 0.00103\\
			& SD & 1.00243 & 1.00123 & 1.00221 & 1.00190 & 1.00176 & 1.24485 & 1.00297 & 1.00141\\
			& CS & -0.00194 & -0.00108 & -0.00487 & -0.00403 & -0.00663 & -0.11239 & -0.01314 & -0.00346\\
			& CK & -0.08664 & -0.07952 & -0.07249 & -0.07082 & -0.07804 & -0.40279 & -0.08985 & -0.08044\\ 
			\hline
			\multirow{5}{*}{0.5}& Mean & 0.00009 & 0.00045 & 0.00032 & -0.00013 & 0.00016 & -0.00058 & 0.00006 & -0.00057 \\
			& Median & -0.00038 & -0.00006 & -0.00049 & 0.00168 & 0.00003 & 0.00009 & 0.00104 & 0.00008\\
			& SD & 1.00249 & 1.00087 & 1.00216 & 1.00159 & 1.00183 & 1.25481 & 1.00301 &  1.00127\\
			& CS & 0.00605 & 0.00235 & 0.00152 & 0.00091 & 0.00558 & -0.00009 & -0.00049 & -0.00437\\
			& CK & -0.09175 & -0.08312 & -0.07903 & -0.08080 & -0.08750 & -0.43546 & -0.08946 & -0.08011\\ 
			\hline
			\multirow{5}{*}{0.75}& Mean & -0.00001 & 0.00143 & 0.00070 & 0.00048 & 0.00044 & -0.23952 & -0.00027 & 0.00051\\
			& Median & -0.00214 & -0.00077 & -0.00102 & 0.00063 & -0.00074 & -0.27161 &  -0.00046 & -0.00004\\
			& SD & 1.00249 & 1.00121 & 1.00221 & 1.00173 & 1.00194 & 1.24974 & 1.00294 & 1.00142\\
			& CS & 0.01398 & 0.00592 & 0.00897 & 0.00642 & 0.01168 & 0.11432 & 0.00587 & 0.00316\\
			& CK & -0.08627 & -0.07922 & -0.07190 & -0.07527 & -0.08225 & -0.38922 & -0.08485 & -0.07795\\
			\bottomrule	
		\end{tabular} 
	}
	\label{tab:10}
\end{table}

\subsection{Study 2: size and power of the tests}
We now present a Monte Carlo simulation study to evaluate the performance of the $S_W,S_{LR},S_R$ and $S_T$ tests. Two measures are considered: size (null rejection rate) and power (nonnull rejection rate). The tests nominal levels used are $\alpha = 0.01, 0.05, 0.1$, and the quantiles and sample sizes are the same as in the previous simulation study. 
The data generating model is given by
$$Y_i = (\beta_0 + \beta_1 x_{1i} + \beta_2 x_{2i} + \beta_3 x_{3i})\epsilon_{i}^{\sqrt{\phi}},\quad 1, \ldots, n,$$
where $\epsilon_i\sim\text{QLS}(1,1,g)$, the covariates values are obtained as Bernoulli(0.5) random draws, and the coefficients not fixed under $H_0$ are $\beta_l = 1, \forall \; \beta_l \neq 0$ and $\phi = 3$, that is, the $\beta_l$ coefficients not fixed under $ H_0 $ are equal to 1. Algorithm \ref{alg:htest} describes briefly the Monte Carlo estimation of the size and power of the tests.

\begin{algorithm}
	\floatname{algorithm}{Algorithm}
	\caption{Steps for Monte Carlo estimation of the size and power of the tests.}\label{alg:htest}
	\begin{algorithmic}[1]
		\State Set the number of parameters $r$ to be tested under $H_0$ and the nominal level $\alpha$.
		\State Set the values of the parameters and generate 5,000 samples of size $n$ based on the chosen function $g(\cdot)$ in Table~\ref{tab:gfuncions}, according to the model postulated in $H_0$ (size), or according to the model postulated in $H_1$ (power).
		\State Estimate the model parameters using the maximum likelihood method for each sample.
		\State Compute the test statistics $S_W$, $S_{LR}$, $S_R$ and $S_T$ for each sample.
		\State Compute the critical value $C_{r,1-\alpha}$ for each test and sample.
		\State Obtain the Monte Carlo estimation of the size/power of the test by calculating the proportion of replicates of the test statistic ($S_W$, $S_{LR}$, $S_R$ and $S_T$) larger than the critical value $C_{r,1-\alpha}$.
	\end{algorithmic}
\end{algorithm}

Tables \ref{tab:12}-\ref{tab:19} contain the tests null rejection rates for different log-symmetric quantile regression models. 
The interest here is to test $H_0: \beta_1 = \cdots = \beta_r = 0$, for $r = 1,3$, that is, the proportion of times that each statistic ($S_W$, $S_{LR}$, $S_R$ and $S_T$) is larger than the upper $\alpha$ quantile of the $\chi^2_{r}$ distribution. We consider $r = 1,3$ in order to verify whether there are differences in test performance under these scenarios. In fact, from Tables \ref{tab:12}-\ref{tab:19}, we can observe a difference in the tests null rejection rates of the tests. For example, for $r = 1$, the tests null rejection rates are closer to the nominal levels than when $r= 3$. We also observe that the tests null rejection rates associated with the all the tests approach the corresponding nominal levels as the sample size increases, as expected. In general, the statistics $S_W$ and $S_{LR}$ present null rejection rates closer to the corresponding nominal levels.

Figures \ref{fig:01}-\ref{fig:08} show the power curves of the four tests. We computed the tests nonnull rejection rates considering the hypothesis $H_0: \beta_1 = \cdots = \beta_t$ against $H_1: \beta_1 = \cdots = \beta_r = \delta,$ $|\delta| = 0,1, \ldots, 4$, where the power function is denoted by $\pi(\delta)$. The power results are shown only for $n = 100$ and $\alpha = 0.01$, as similar results were observed for other settings. With the exception of the results based on the EBS $(\vartheta=0.5)$ and EBS-$t$ $(\vartheta_1=0.5,\vartheta_2=3)$ distributions, the power  with $ r = 3 $ is greater than the power with $ r = 1 $, which indicates that the tests studied are more efficient when we test multiple parameters. In general, we observe that the power to test $r = 3$ parameters simultaneously is greater for the $S_R$ test, followed by the $S_T$ test. When $r =1$, the powers of the tests using the $S_R$ and $S_T$ statistics practically coincides. For the quantile regression model with the EBS $(\vartheta=0.5)$ and EBS-$t$ $(\vartheta_1=0.5,\vartheta_2=3)$ distributions, the power curves of the four tests overlapped, except for the power curve of the Wald test (slightly above the others) with 
$r=1$ for the EBS-$t$ $(\vartheta_1=0.5,\vartheta_2=3)$. In general, the power tends to 1 as $|\delta|$ increases for all cases considered in this simulation, as expected.

\begin{table}[!ht]
\footnotesize
	\centering
	\caption{\small Null rejection rates for $H_0: \beta_1 = \cdots = \beta_r = 0$ in the log-NO quantile regression model.}
	\adjustbox{max height=\dimexpr\textheight-3.5cm\relax,
		max width=\textwidth}{
		\begin{tabular}{ccccccccccccccc}
			\toprule
			\multirow{2}{*}{$t$}&\multirow{2}{*}{$q$}& & \multicolumn{3}{c}{$n = 50$} & &\multicolumn{3}{c}{$n = 100$} & &\multicolumn{3}{c}{$n = 200$}\\
			\cline{4-6} \cline{8-10} \cline{12-14}
			&&  & 1\% & 5\% & 10\% & & 1\% & 5\% & 10\% & & 1\% & 5\% & 10\%\\
			\hline
			\multirow{12}{*}{3}&\multirow{4}{*}{0.25}& $S_W$ & 0.0276 & 0.0864 & 0.1272 && 0.0136 & 0.0586 & 0.1078 && 0.015 & 0.0588 & 0.1056\\
			&& $S_{LR}$ & 0.0132 & 0.0644 & 0.1162 && 0.0118 & 0.0538 & 0.1134 && 0.0068 & 0.053 & 0.1046 \\
			&& $S_R$ &   0.0238 &  0.0716 &  0.1172 &&  0.0244 &  0.0698 &  0.1158 &&  0.0240 &  0.0722 &  0.1204\\
			&& $S_T$ & 0.0186 & 0.0704 & 0.1562 && 0.0210 & 0.0674 & 0.1224 && 0.0174 & 0.0712 & 0.1186\\
			\cline{2-14}
			&\multirow{4}{*}{0.5}& $S_W$ & 0.0294 & 0.0848 & 0.1438  && 0.0146 & 0.0596 & 0.1100 && 0.0154 & 0.0594 & 0.1090\\
			&& $S_{LR}$ & 0.017 & 0.0732 & 0.1124 && 0.013 & 0.0538 & 0.1114 && 0.011 & 0.0504 & 0.1064 \\
			&& $S_R$ &  0.0248 &  0.0716 &  0.1172 &&  0.0244 &  0.0698 &  0.1158 &&  0.024 &  0.0722 &  0.1204\\
			&& $S_T$ & 0.0186 & 0.068 & 0.1276 && 0.021 & 0.0674 & 0.1224 && 0.0174 & 0.0712 & 0.1186 \\
			\cline{2-14}
			&\multirow{4}{*}{0.75} & $S_W$ & 0.0296 & 0.0836 & 0.1412 && 0.0140 & 0.0590 & 0.1098 && 0.0158 & 0.0584 & 0.1074\\
			&& $S_{LR}$ & 0.0146 & 0.0668 & 0.1184 && 0.013 & 0.0548 & 0.1148 && 0.0116 & 0.0544 & 0.104 \\
			&& $S_R$ &  0.0248 &  0.0716 &  0.1172 &&  0.0244 &  0.0698 &  0.1158 &&  0.024 &  0.0722 &  0.1162\\
			&& $S_T$ & 0.0186 & 0.0704 & 0.1276 && 0.0212 & 0.0674 & 0.1224 && 0.0174 & 0.0694 & 0.1186 \\
			\bottomrule
			\multirow{12}{*}{1}&\multirow{4}{*}{0.25}& $S_W$ &  0.0172 & 0.0818 & 0.139 && 0.0136 & 0.0588 & 0.1154 && 0.0104 & 0.0554 & 0.1026 \\
			&& $S_{LR}$ &0.042 & 0.1296 & 0.1908 && 0.0254 & 0.0906 & 0.1494 && 0.016 & 0.0482 & 0.0788 \\
			&& $S_R$ & 0.0102 & 0.047 & 0.1024 && 0.0102 & 0.0506 & 0.0954 && 0.0104 & 0.0486 & 0.098 \\
			&& $S_T$ & 0.014 & 0.0722 & 0.1216 && 0.0196 & 0.065 & 0.12 && 0.0178 & 0.067 & 0.128 \\
			\cline{2-14}
			&\multirow{4}{*}{0.5}& $S_W$ & 0.0226 & 0.0824 & 0.1382 && 0.0138 & 0.0606 & 0.1158 && 0.0098 & 0.0552 & 0.1048 \\
			&& $S_{LR}$ & 0.0406 & 0.129 & 0.1984 && 0.0264 & 0.1086 & 0.1344 && 0.014 & 0.0516 & 0.0944 \\
			&& $S_R$ & 0.008 & 0.047 & 0.1024 && 0.0102 & 0.0506 & 0.0954 && 0.0104 & 0.0486 & 0.098 \\
			&& $S_T$ & 0.019 & 0.0722 & 0.1216 && 0.0196 & 0.065 & 0.12 && 0.017 & 0.067 & 0.1282 \\
			\cline{2-14}
			&\multirow{4}{*}{0.75} & $S_W$ & 0.023 & 0.083 & 0.1398 && 0.0138 & 0.0606 & 0.1158 && 0.01 & 0.0556 & 0.1052 \\
			&& $S_{LR}$ & 0.0388 & 0.1308 & 0.1944 && 0.025 & 0.0894 & 0.1482 && 0.0156 & 0.0448 & 0.0766 \\
			&& $S_R$ & 0.008 & 0.047 & 0.1024 && 0.0102 & 0.0506 & 0.0954 && 0.0104 & 0.0486 & 0.098 \\
			&& $S_T$ & 0.0206 & 0.0722 & 0.1216 && 0.0196 & 0.065 & 0.12 && 0.017 & 0.067 & 0.1282 \\
			\bottomrule
	\end{tabular}}
\label{tab:12}
\end{table}

\begin{figure}[!ht]
	\centering
	\subfigure[$r=3$]{
	\includegraphics[scale = 0.5]{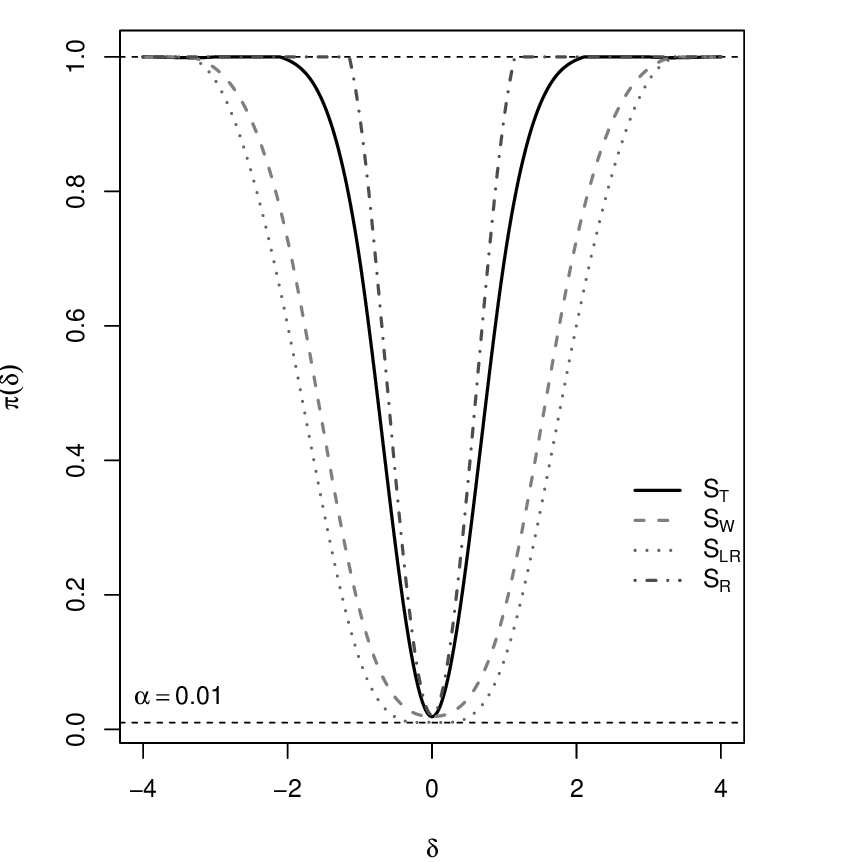}
	}
	\subfigure[$r=1$]{
	\includegraphics[scale = 0.5]{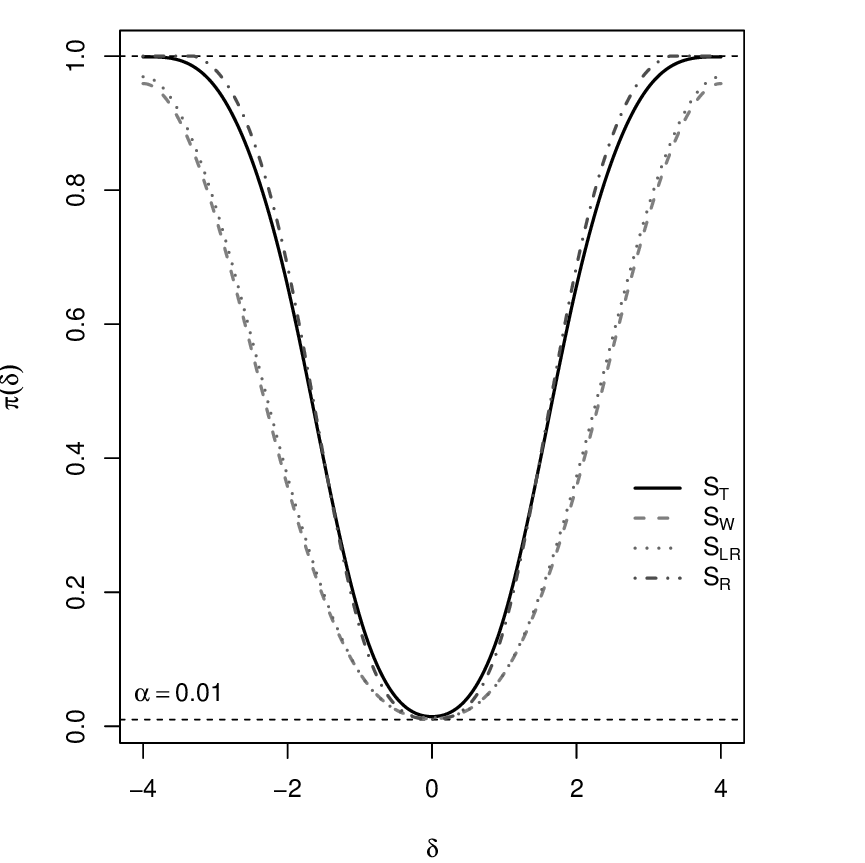}
    }
   \caption{\small Power curve of the tests in the log-NO quantile regression model.} 
\label{fig:01}
\end{figure}

\begin{table}[!ht]
\footnotesize
	\centering
	\caption{\small Null rejection rates for $H_0: \beta_1 = \cdots = \beta_r = 0$ in the log-$t$ quantile regression model ($\vartheta=3$).}
	\adjustbox{max height=\dimexpr\textheight-3.5cm\relax,
		max width=\textwidth}{
		\begin{tabular}{ccccccccccccccc}
			\toprule
			\multirow{2}{*}{$t$} & \multirow{2}{*}{$q$}& & \multicolumn{3}{c}{$n = 50$} & &\multicolumn{3}{c}{$n = 100$} & &\multicolumn{3}{c}{$n = 200$}\\
			\cline{4-6} \cline{8-10} \cline{12-14}
			&& & 1\% & 5\% & 10\% & & 1\% & 5\% & 10\% & & 1\% & 5\% & 10\%\\
			\hline
			\multirow{12}{*}{3}&\multirow{4}{*}{0.25}& $S_W$ & 0.0270 & 0.0920 & 0.1486 && 0.0160 & 0.0696 & 0.1186 && 0.0152 & 0.0562 & 0.1098 \\
			&& $S_{LR}$ & 0.0162 & 0.0644 & 0.128 && 0.0114 & 0.0544 & 0.1108 && 0.0122 & 0.0534 & 0.1082 \\
			&& $S_R$ &  0.0244 &  0.0714 &  0.1178 &&  0.0226 &  0.0654 &  0.1078 &&  0.0232 &  0.0664 &  0.1132 \\
			&& $S_T$ & 0.0204 & 0.0782 & 0.1422 && 0.0188 & 0.0718 & 0.1236 && 0.0192 & 0.0696 & 0.1242 \\
			\cline{2-14}
			&\multirow{4}{*}{0.5}& $S_W$ & 0.0274 & 0.0920 & 0.1486 && 0.0156 & 0.0696 & 0.1186 && 0.0150 & 0.0562 & 0.1098\\
			&& $S_{LR}$ & 0.0146 & 0.0728 & 0.1248 && 0.0134 & 0.0584 & 0.1132 && 0.0124 & 0.0506 & 0.1092 \\
			&& $S_R$ &  0.0244 & 0.0714 &  0.1178 &&  0.0232 
			&  0.0654 &  0.1078 &&  0.0232 &  0.0664 &  0.1132\\
			&& $S_T$ & 0.0214 & 0.0782 & 0.13 && 0.0188 & 0.0718 & 0.1236 && 0.0192 & 0.0696 & 0.1242\\
			\cline{2-14}
			&\multirow{4}{*}{0.75} & $S_W$ & 0.0290 & 0.0912 & 0.1494 && 0.0148 & 0.0676 & 0.1182 && 0.0152 & 0.0582 & 0.1094\\
			&& $S_{LR}$ & 0.0138 & 0.0754 & 0.1212 && 0.0122 & 0.0562 & 0.111 && 0.0096 & 0.057 & 0.1014 \\
			&& $S_R$ &  0.0244 &  0.0714 &  0.1178 &&  0.0226 &  0.0654 & 0.1078 &&  0.0232 &  0.0664 &  0.1132\\
			&& $S_T$ & 0.0214 & 0.0782 & 0.1424 && 0.0188 & 0.0718 & 0.1236 && 0.0192 & 0.0696 & 0.1242 \\
			\bottomrule
			\multirow{12}{*}{1}&\multirow{4}{*}{0.25}& $S_W$ & 0.0176 & 0.0658 & 0.123 && 0.0158 & 0.0702 & 0.1216 && 0.0154 & 0.0652 & 0.113 \\
			&& $S_{LR}$ & 0.0456 & 0.1486 & 0.2258 && 0.0318 & 0.1172 & 0.1826 && 0.0242 & 0.0776 & 0.1232 \\
			&& $S_R$ & 0.0088 & 0.0496 & 0.1006 && 0.0098 & 0.0452 & 0.0928 && 0.0098 & 0.0468 & 0.0968 \\
			&& $S_T$ & 0.0232 & 0.0684 & 0.1208 && 0.0178 & 0.069 & 0.1208 && 0.0184 & 0.0642 & 0.1192 \\
			\cline{2-14}
			&\multirow{4}{*}{0.5}& $S_W$ & 0.017 & 0.07 & 0.1224 && 0.0142 & 0.0646 & 0.1216 && 0.0152 & 0.0692 & 0.1152 \\
			&& $S_{LR}$ & 0.0432 & 0.1486 & 0.2258 && 0.0362 & 0.1172 & 0.1826 && 0.0214 & 0.0776 & 0.123 \\
			&& $S_R$ & 0.0088 & 0.0496 & 0.1006 && 0.0098 & 0.0452 & 0.0928 && 0.0098 & 0.0468 & 0.0968\\
			&& $S_T$ & 0.0224 & 0.0752 & 0.1332 && 0.0178 & 0.069 & 0.1208 && 0.0184 & 0.0642 & 0.1192 \\
			\cline{2-14}
			&\multirow{4}{*}{0.75} & $S_W$ & 0.0178 & 0.0702 & 0.1242 && 0.0152 & 0.0638 & 0.125 && 0.0154 & 0.0684 & 0.1158 \\
			&& $S_{LR}$ & 0.0566 & 0.1486 & 0.2258 && 0.0348 & 0.1172 & 0.1826 && 0.0232 & 0.0776 & 0.123 \\
			&& $S_R$ & 0.0088 & 0.0496 & 0.1006 && 0.0098 & 0.0452 & 0.0928 && 0.0098 & 0.0468 & 0.0968\\
			&& $S_T$ & 0.0224 & 0.0752 & 0.1332 && 0.0178 & 0.069 & 0.1208 && 0.0184 & 0.064 & 0.1192 \\
			\bottomrule
	\end{tabular}}
\label{tab:13}
\end{table}

\begin{figure}[!ht]
	\centering
	\subfigure[$r=3$]{
		\includegraphics[scale = 0.5]{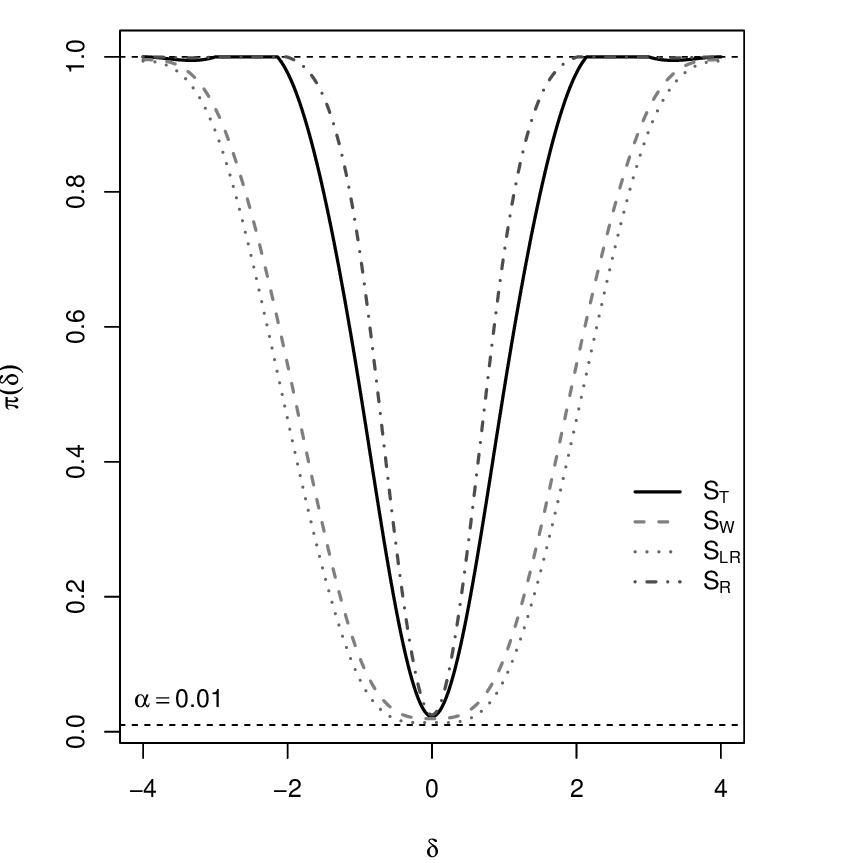}
	}
	\subfigure[$r=1$]{
		\includegraphics[scale = 0.5]{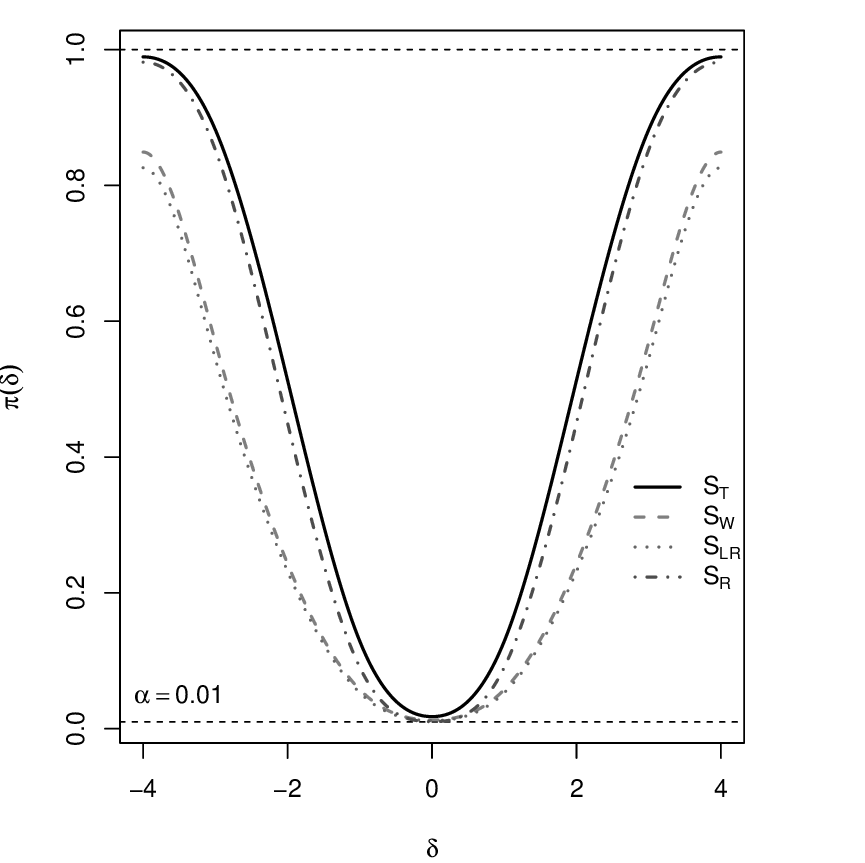}
	}
	\caption{\small Power curve of the tests in the log-$t$ quantile regression model ($\vartheta=3$).}
\label{fig:02}
\end{figure}

\begin{table}[!ht]
\footnotesize
	\centering
	\caption{\small Null rejection rates for $H_0: \beta_1 = \cdots = \beta_r = 0$ in the log-PE quantile regression model ($\vartheta=0.3$).}
	\adjustbox{max height=\dimexpr\textheight-3.5cm\relax,
		max width=\textwidth}{
		\begin{tabular}{ccccccccccccccc}
			\toprule
			\multirow{2}{*}{$t$}&\multirow{2}{*}{$q$}& & \multicolumn{3}{c}{$n = 50$} & &\multicolumn{3}{c}{$n = 100$} & &\multicolumn{3}{c}{$n = 200$}\\
			\cline{4-6} \cline{8-10} \cline{12-14}
			&& & 1\% & 5\% & 10\% & & 1\% & 5\% & 10\% & & 1\% & 5\% & 10\%\\
			\hline
			\multirow{12}{*}{3}& \multirow{4}{*}{0.25}& $S_W$ & 0.0284 & 0.0906 & 0.1508 && 0.0180 & 0.0684 & 0.1180 && 0.0148 & 0.0532 & 0.1066\\
			&& $S_{LR}$ & 0.0114 & 0.063 & 0.1154 && 0.0142 & 0.061 & 0.1158 && 0.011 & 0.0544 &0.1006 \\
			&& $S_R$ & 0.0218 &  0.0644 &  0.1078 && 0.0276 & 0.0714 &  0.1134 &&  0.0222 & 0.0646 & 0.1122\\
			&& $S_T$ & 0.0218 & 0.0772 & 0.127 && 0.0218 & 0.0706 & 0.1246 && 0.017 & 0.0628 & 0.1174 \\
			\cline{2-14}
			&\multirow{4}{*}{0.5}& $S_W$ & 0.0276 & 0.0894 & 0.1548 && 0.0204 & 0.0688 & 0.1196 && 0.0158 & 0.0540 & 0.1060\\
			&& $S_{LR}$ & 0.013 & 0.061 & 0.1184 && 0.013 & 0.0552 & 0.1138 && 0.0104 & 0.0564 & 0.1052 \\
			&& $S_R$ & 0.0228 & 0.0644 & 0.1078 &&  0.0276 & 0.0714 & 0.1134 && 0.0222 & 0.0646 & 0.1122\\
			&& $S_T$ & 0.0192 & 0.0704 & 0.1264 && 0.0218 & 0.0706 & 0.1244 && 0.017 & 0.063 & 0.1172\\
			\cline{2-14}
			&\multirow{4}{*}{0.75} & $S_W$ & 0.0262 & 0.0900 & 0.1534 && 0.0198 & 0.0694 & 0.1286 && 0.0150 & 0.0538 & 0.1176\\
			&& $S_{LR}$ & 0.0144 & 0.0652 & 0.121 && 0.0108 & 0.0542 & 0.117 && 0.0108 & 0.0552 & 0.1048 \\
			&& $S_R$ & 0.0228 & 0.0644 & 0.1078 && 0.0276 & 0.0714 & 0.1134 && 0.0222 & 0.0646 & 0.1122\\
			&& $S_T$ & 0.0192 & 0.0736 & 0.1264 && 0.0218 & 0.0706 & 0.1244 && 0.017 & 0.0628 & 0.1168\\
			\bottomrule
			\multirow{12}{*}{1}& \multirow{4}{*}{0.25}& $S_W$ & 0.0166 & 0.0658 & 0.1196 && 0.0154 & 0.0536 & 0.105 && 0.0108 & 0.0528 & 0.1042 \\
			&& $S_{LR}$ & 0.0436 & 0.158 & 0.2318 && 0.0338 & 0.1062 & 0.1716 && 0.0198 & 0.072 & 0.1124 \\
			&& $S_R$ & 0.0092 & 0.0508 & 0.0984 && 0.0092 & 0.0522 & 0.1014 && 0.0094 & 0.0468 & 0.0936 \\
			&& $S_T$ & 0.0206 & 0.0784 & 0.1336 && 0.0154 & 0.069 & 0.1262 && 0.016 & 0.0584 & 0.1122 \\
			\cline{2-14}
			&\multirow{4}{*}{0.5}& $S_W$ & 0.0204 & 0.068 & 0.1236 && 0.016 & 0.0564 & 0.1032 && 0.01 & 0.0548 & 0.1054 \\
			&& $S_{LR}$ & 0.046 & 0.158 & 0.2318 && 0.0338 & 0.1062 & 0.1716 && 0.0198 & 0.072 & 0.1124 \\
			&& $S_R$ & 0.0068 & 0.0508 & 0.0984 && 0.0092 & 0.0522 & 0.1014 && 0.0094 & 0.0468 & 0.0936 \\
			&& $S_T$ & 0.0194 & 0.0784 & 0.1332 && 0.0184 & 0.069 & 0.1262 && 0.016 & 0.0584 & 0.1122 \\
			\cline{2-14}
			&\multirow{4}{*}{0.75} & $S_W$ & 0.0204 & 0.0674 & 0.1252 && 0.0162 & 0.0564 & 0.1056 && 0.0112 & 0.0556 & 0.1072 \\
			&& $S_{LR}$ & 0.0506 & 0.158 & 0.2318 && 0.0338 & 0.1062 & 0.1716 && 0.0198 & 0.072 & 0.1124 \\
			&& $S_R$ & 0.0068 & 0.0508 & 0.0984 && 0.0092 & 0.0522 & 0.1014 && 0.0094 & 0.0468 & 0.0936 \\
			&& $S_T$ & 0.0206 & 0.0784 & 0.1332 && 0.0184 & 0.069 & 0.1262 && 0.016 & 0.0584 & 0.1122 \\
			\bottomrule
	\end{tabular}}
\label{tab:14}
\end{table}

\begin{figure}[!ht]
	\centering
	\subfigure[$r=3$]{
		\includegraphics[scale = 0.5]{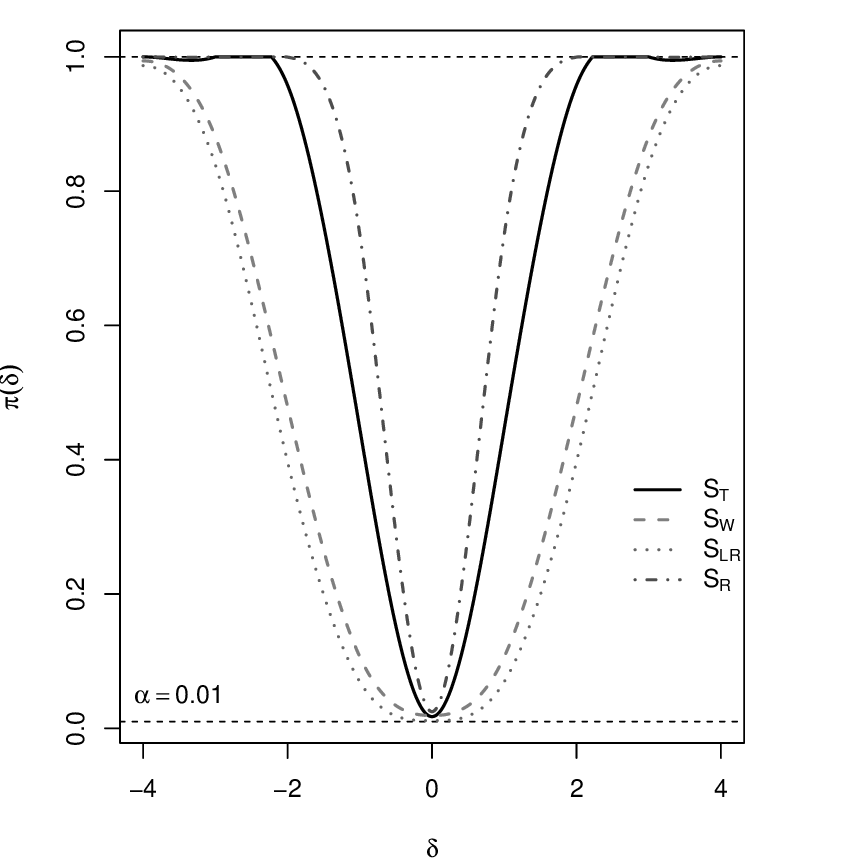}
	}
	\subfigure[$r=1$]{
		\includegraphics[scale = 0.5]{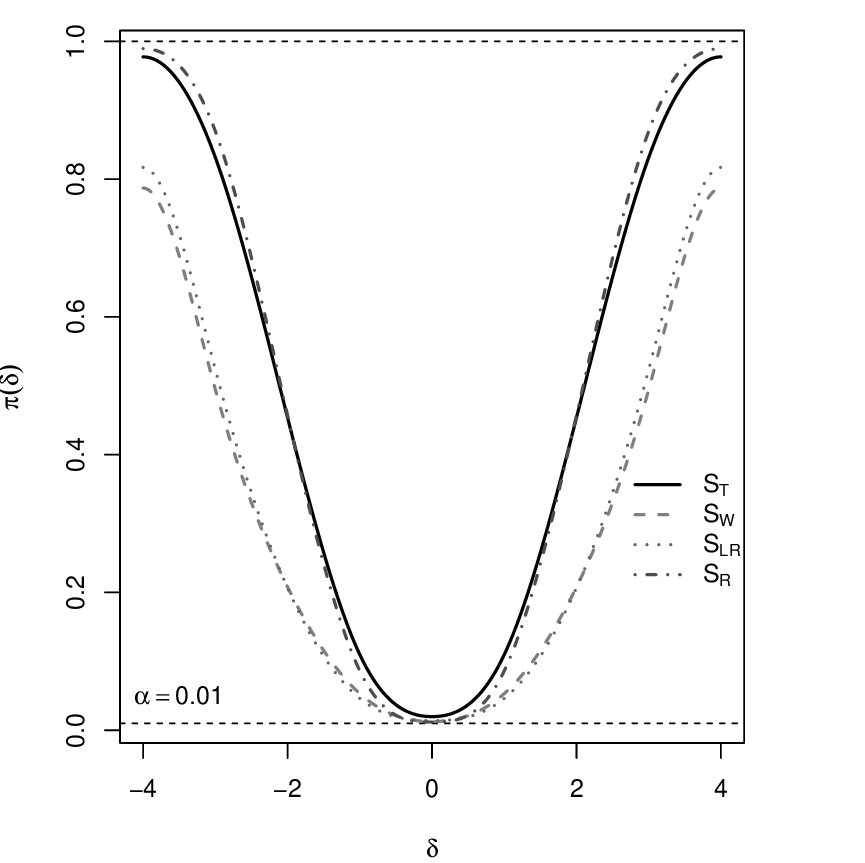}
	}
	\caption{\small Power curve of the tests in the log-PE quantile regression model ($\vartheta=0.3$).}
\label{fig:03}
\end{figure}

\begin{table}[!ht]
\footnotesize
	\centering
	\caption{\small Null rejection rates for $H_0: \beta_1 = \cdots = \beta_r = 0$ in the log-HP quantile regression model ($\vartheta=2$).}
	\adjustbox{max height=\dimexpr\textheight-3.5cm\relax,
		max width=\textwidth}{
		\begin{tabular}{ccccccccccccccc}
			\toprule
			\multirow{2}{*}{$t$}&\multirow{2}{*}{$q$}& & \multicolumn{3}{c}{$n = 50$} & &\multicolumn{3}{c}{$n = 100$} & &\multicolumn{3}{c}{$n = 200$}\\
			\cline{4-6} \cline{8-10} \cline{12-14}
			&& & 1\% & 5\% & 10\% & & 1\% & 5\% & 10\% & & 1\% & 5\% & 10\%\\
			\hline
			\multirow{12}{*}{3} & \multirow{4}{*}{0.25}& $S_W$ & 0.0228 & 0.0784 & 0.1322 && 0.0138 & 0.0676 & 0.1218 && 0.0174 & 0.0602 & 0.1084\\
			&& $S_{LR}$ & 0.013 & 0.063 & 0.1122 && 0.0144 & 0.0596 & 0.113 && 0.0106 & 0.0542 & 0.1028 \\
			&& $S_R$ &  0.028 & 0.0712 & 0.1176 && 0.024 & 0.0698 & 0.111 && 0.0252 & 0.0722 & 0.1124\\
			&& $S_T$ & 0.0196 & 0.0896 & 0.1432 && 0.0196 & 0.0688 & 0.12 && 0.0208 & 0.0698 & 0.1294\\
			\cline{2-14}
			&\multirow{4}{*}{0.5}& $S_W$ & 0.0264 & 0.0782 & 0.1318 && 0.0138 & 0.0676 & 0.1212 && 0.0168 & 0.0594 & 0.1100\\
			&& $S_{LR}$ & 0.015 & 0.063 & 0.1122 && 0.0144 & 0.0596 & 0.113 && 0.0106 & 0.0542 & 0.1028 \\
			&& $S_R$ &  0.024 & 0.0698 & 0.111 && 0.0252 & 0.0722 & 0.1124 && 0.0252 & 0.0722 & 0.1124\\
			&& $S_T$ & 0.0266 & 0.0896 & 0.1432 && 0.0196 & 0.0688 & 0.12 && 0.0208 & 0.0698 & 0.1238 \\
			\cline{2-14}
			&\multirow{4}{*}{0.75} & $S_W$ & 0.0208 & 0.0828 & 0.1320 && 0.0190 & 0.0688 & 0.1254 && 0.0136 & 0.0586 & 0.1002\\
			&& $S_{LR}$ & 0.015 & 0.063 & 0.1122 && 0.0144 & 0.0596 & 0.113 && 0.0106 & 0.0542 & 0.1028 \\
			&& $S_R$ &  0.0214 & 0.0712 & 0.1176 && 0.024 & 0.0698 & 0.111 && 0.0252 & 0.0722 & 0.1124\\
			&& $S_T$ & 0.0196 & 0.0896 & 0.1432 && 0.0196 & 0.0688 & 0.1202 && 0.0194 & 0.078 & 0.1238\\
			\bottomrule
			\multirow{12}{*}{1} & \multirow{4}{*}{0.25}& $S_W$ &  0.0166 & 0.0616 & 0.1172 && 0.0142 & 0.0564 & 0.1096 && 0.011 & 0.049 & 0.099\\
			&& $S_{LR}$ & 0.0372 & 0.1196 & 0.2064 && 0.022 & 0.0782 & 0.1368 && 0.0104 & 0.0378 & 0.0664 \\
			&& $S_R$ & 0.0114 & 0.0492 & 0.0998 && 0.01 & 0.0536 & 0.1008 && 0.0108 & 0.0476 & 0.097 \\
			&& $S_T$ & 0.0156 & 0.0862 & 0.1424 && 0.0184 & 0.0674 & 0.1212 && 0.0176 & 0.0676 & 0.132 \\
			\cline{2-14}
			&\multirow{4}{*}{0.5}& $S_W$ & 0.017 & 0.0628 & 0.1178 && 0.0148 & 0.0576 & 0.1102 && 0.0104 & 0.0486 & 0.099 \\
			&& $S_{LR}$ & 0.0356 & 0.1208 & 0.1902 && 0.0226 & 0.0848 & 0.13 && 0.014 & 0.0428 & 0.0704 \\
			&& $S_R$ & 0.0106 & 0.0492 & 0.0998 && 0.01 & 0.0536 & 0.1008 && 0.0108 & 0.0476 & 0.097 \\
			&& $S_T$ & 0.0238 & 0.086 & 0.1424 && 0.0184 & 0.0674 & 0.1212 && 0.0176 & 0.0676 & 0.132 \\
			\cline{2-14}
			&\multirow{4}{*}{0.75} & $S_W$ & 0.0168 & 0.0634 & 0.1176 && 0.0146 & 0.058 & 0.1116 && 0.0112 & 0.0492 & 0.0994 \\
			&& $S_{LR}$ & 0.0394 & 0.1018 & 0.1806 && 0.0276 & 0.0748 & 0.126 && 0.0126 & 0.0372 & 0.0658 \\
			&& $S_R$ & 0.0106 & 0.0492 & 0.0998 && 0.01 & 0.0536 & 0.1008 && 0.0108 & 0.0476 & 0.097 \\
			&& $S_T$ & 0.0238 & 0.0862 & 0.1424 && 0.0184 & 0.0674 & 0.1212 && 0.0176 & 0.0676 & 0.132 \\
			\bottomrule
	\end{tabular}}
\label{tab:15}
\end{table}

\begin{figure}[!ht]
	\centering
	\subfigure[$r=3$]{
		\includegraphics[scale = 0.5]{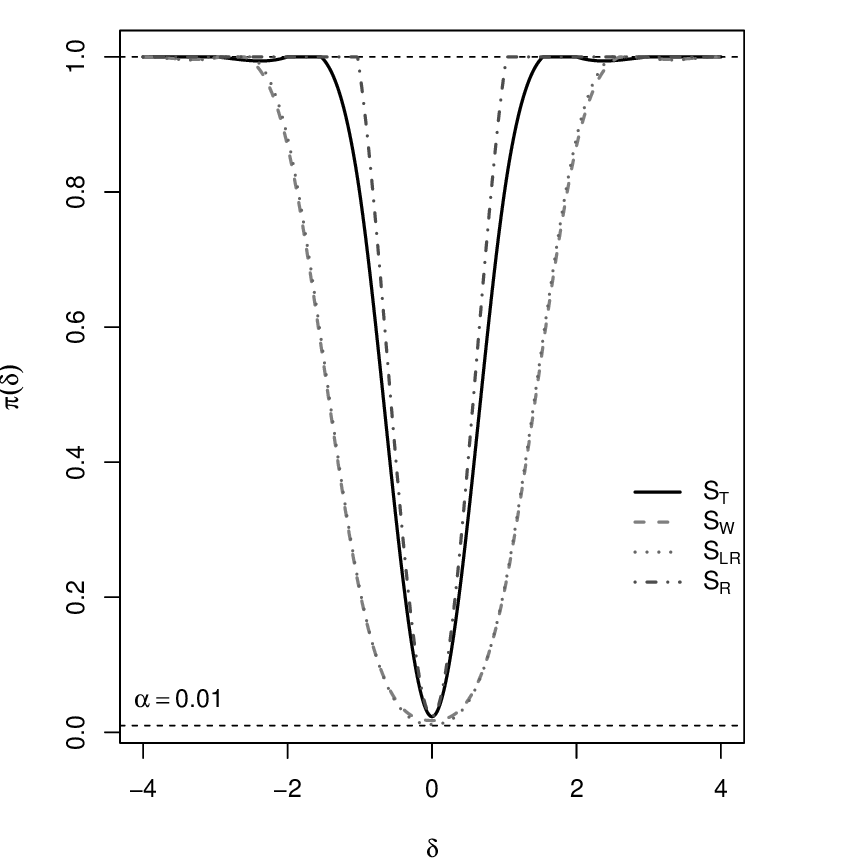}
	}
	\subfigure[$r=1$]{
		\includegraphics[scale = 0.5]{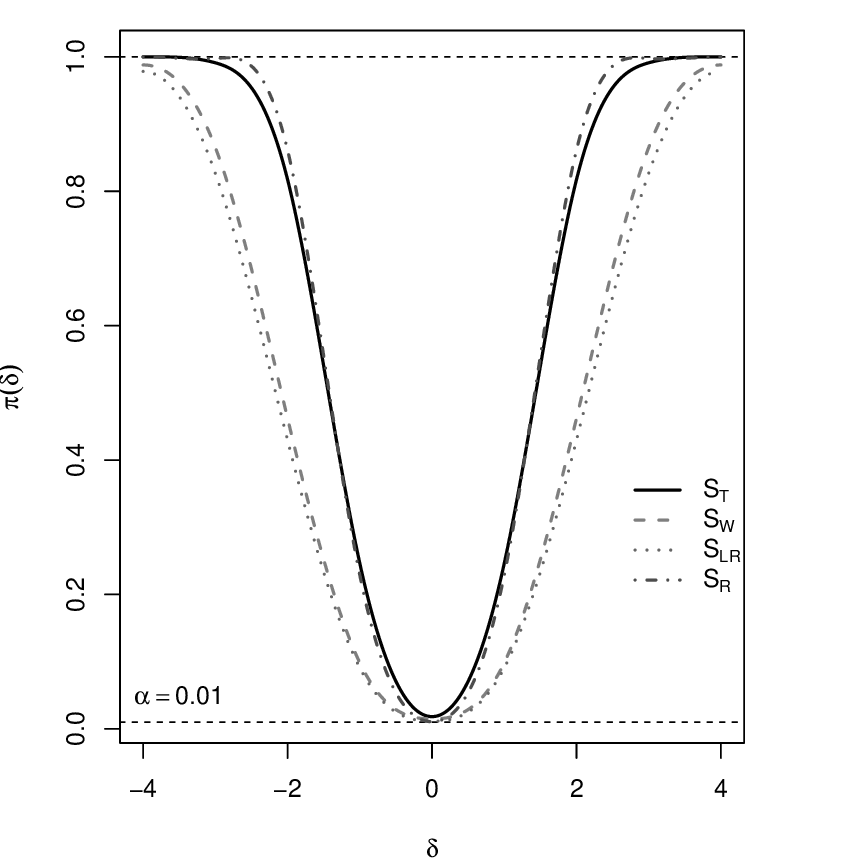}
	}
	\caption{\small Power curve of the tests in the log-HP quantile regression model ($\vartheta=2$).}
	\label{fig:04}
\end{figure}

\begin{table}[!ht]
\footnotesize
	\centering
	\caption{\small Null rejection rates for $H_0: \beta_1 = \cdots = \beta_r = 0$ in the log-SL quantile regression model ($\vartheta=4$).}
	\adjustbox{max height=\dimexpr\textheight-3.5cm\relax,
		max width=\textwidth}{
		\begin{tabular}{ccccccccccccccc}
			\toprule
			\multirow{2}{*}{$t$} & \multirow{2}{*}{$q$}& & \multicolumn{3}{c}{$n = 50$} & &\multicolumn{3}{c}{$n = 100$} & &\multicolumn{3}{c}{$n = 200$}\\
			\cline{4-6} \cline{8-10} \cline{12-14}
			&& & 1\% & 5\% & 10\% & & 1\% & 5\% & 10\% & & 1\% & 5\% & 10\%\\
			\hline
			\multirow{12}{*}{3} & \multirow{4}{*}{0.25}& $S_W$ & 0.0274 & 0.0764 & 0.1482 && 0.0166 & 0.0644 & 0.1224 && 0.0134 & 0.0554 & 0.1026\\
			&& $S_{LR}$ & 0.0136 & 0.062 & 0.1276 && 0.0112 & 0.059 & 0.1138 && 0.0122 & 0.0534 & 0.1094 \\
			&& $S_R$ &  0.0196 & 0.0666 & 0.1134 && 0.0262 & 0.0702 & 0.1186 && 0.024 & 0.0672 & 0.1142 \\
			&& $S_T$ & 0.0226 & 0.0814 & 0.1396 && 0.0258 & 0.075 & 0.133 && 0.0218 & 0.0702 & 0.1282 \\
			\cline{2-14}
			&\multirow{4}{*}{0.5}& $S_W$ & 0.0260 & 0.0766 & 0.1414 && 0.0168 & 0.0638 & 0.1188 && 0.0130 &  0.063 & 0.1146\\
			&& $S_{LR}$ & 0.0134 & 0.0622 & 0.1218 && 0.0096 & 0.0534 & 0.1082 && 0.0102 & 0.0496 & 0.0992 \\
			&& $S_R$ &  0.0196 & 0.0666 & 0.1134 && 0.0262 & 0.0702 & 0.1186 && 0.024 & 0.0672 & 0.1142\\
			&& $S_T$ & 0.0264 & 0.0814 & 0.1396 && 0.0258 & 0.075 & 0.133 && 0.0178 & 0.0702 & 0.1282\\
			\cline{2-14}
			& \multirow{4}{*}{0.75} & $S_W$ & 0.0264 & 0.0822 & 0.1376 && 0.0174 & 0.064 & 0.118 && 0.0126 & 0.052 & 0.113\\
			&& $S_{LR}$ & 0.0154 & 0.065 & 0.1192 && 0.0134 & 0.055 & 0.1036 && 0.009 & 0.0536 & 0.1056 \\
			&& $S_R$ &   0.0196 & 0.0712 & 0.1134 && 0.0262 & 0.0702 
			& 0.1186 && 0.024 & 0.0672 & 0.1142 \\
			&& $S_T$ & 0.0264 & 0.0814 & 0.1396 && 0.0258 & 0.075 & 0.133 && 0.0218 & 0.0742 & 0.1282 \\
			\bottomrule
			\multirow{12}{*}{1} & \multirow{4}{*}{0.25}& $S_W$ &  0.0158 & 0.066 & 0.1188 && 0.0108 & 0.063 & 0.1122 && 0.0124 & 0.0502 & 0.0998 \\
			&& $S_{LR}$ & 0.0478 &0.1416 & 0.215 && 0.0336 & 0.1198 & 0.1548 && 0.0198 & 0.0606 & 0.1112 \\
			&& $S_R$ & 0.0096 & 0.053 & 0.1026 && 0.0112 & 0.0526 & 0.103 && 0.009 & 0.0446 & 0.0966 \\
			&& $S_T$ & 0.024 & 0.0686 & 0.1288 && 0.0188 & 0.0722 & 0.1248 && 0.0156 & 0.0686 & 0.1216 \\
			\cline{2-14}
			&\multirow{4}{*}{0.5}& $S_W$ & 0.0162 & 0.066 & 0.1218 && 0.0106 & 0.0608 & 0.1124 && 0.0116 & 0.0498 & 0.1008\\
			&& $S_{LR}$ & 0.0474 & 0.1376 & 0.198 && 0.0338 & 0.1058 & 0.1644 && 0.0238 & 0.0658 & 0.1038 \\
			&& $S_R$ & 0.009 & 0.053 & 0.1026 && 0.0112 & 0.0526 & 0.103 && 0.009 & 0.0446 & 0.0966 \\
			&& $S_T$ & 0.0196 & 0.0716 & 0.1244 && 0.0188 & 0.072 & 0.1248 && 0.0156 & 0.0686 & 0.1218 \\
			\cline{2-14}
			& \multirow{4}{*}{0.75} & $S_W$ & 0.015 & 0.0682 & 0.123 && 0.011 & 0.0616 & 0.1132 && 0.0124 & 0.05 & 0.0996 \\
			&& $S_{LR}$ & 0.0396 & 0.1378 & 0.2114 && 0.0366 & 0.0904 & 0.1586 && 0.0186 & 0.0716 & 0.1 \\
			&& $S_R$ & 0.009 & 0.053 & 0.1026 && 0.0112 & 0.0526 & 0.103 && 0.009 & 0.0446 & 0.0966 \\
			&& $S_T$ & 0.0196 & 0.0718 & 0.1242 && 0.0188 & 0.072 & 0.1248 && 0.0156 & 0.0684 & 0.1218 \\
			\bottomrule
	\end{tabular}}
\label{tab:16}
\end{table}

\begin{figure}[!ht]
	\centering
	\subfigure[$r=3$]{
		\includegraphics[scale = 0.5]{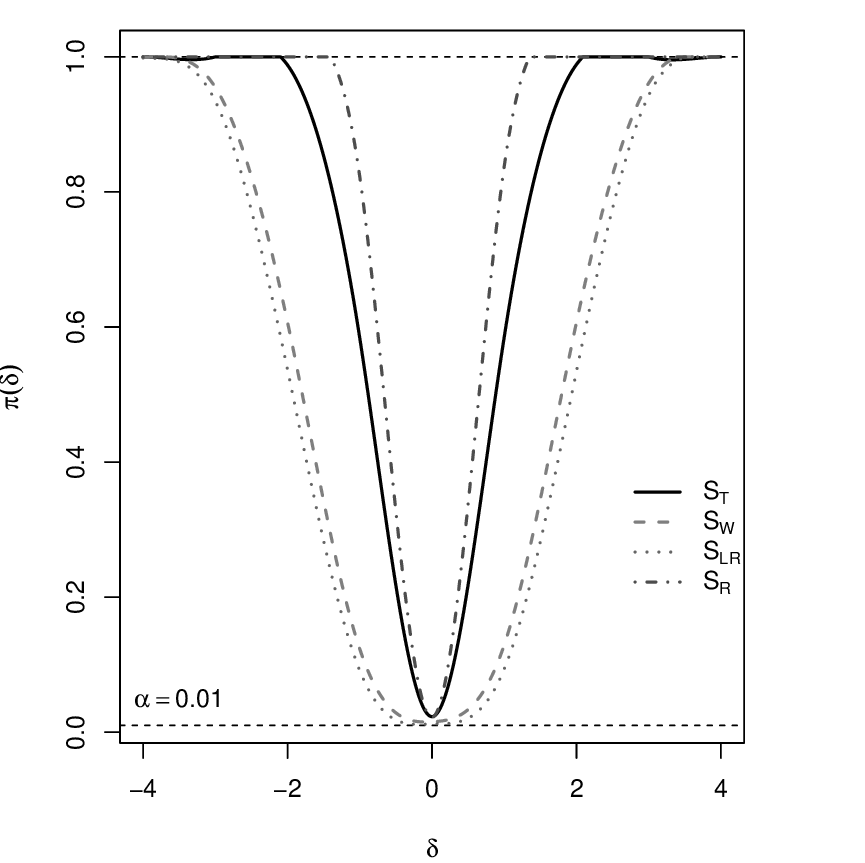}
	}
	\subfigure[$r=1$]{
		\includegraphics[scale = 0.5]{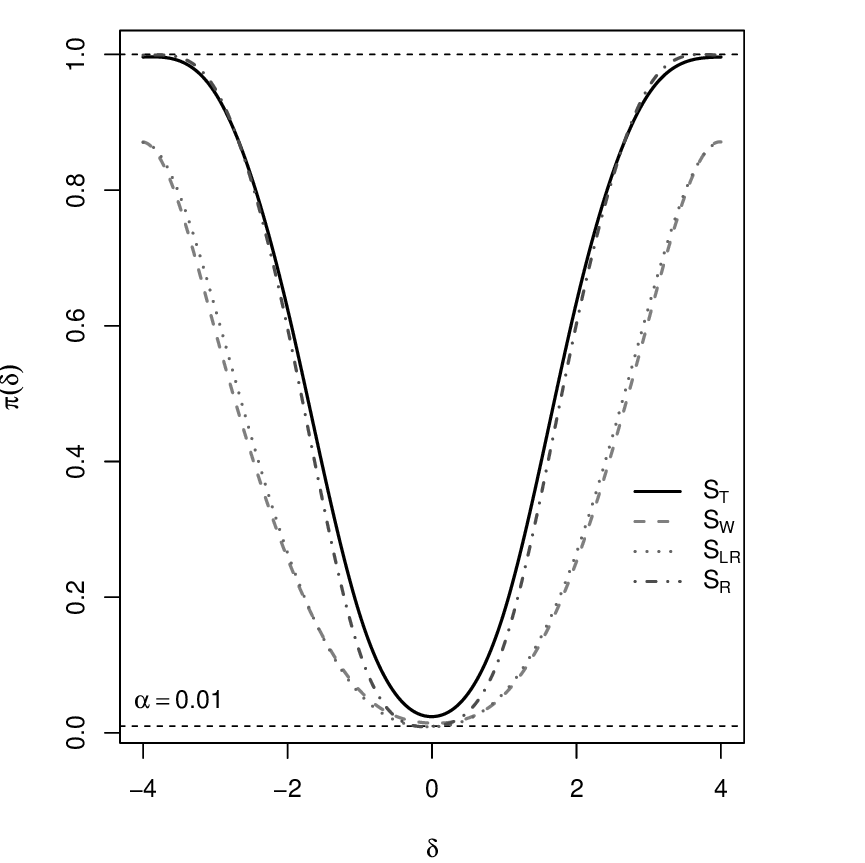}
	}
	\caption{\small Power curve of the tests in the log-SL quantile regression model ($\vartheta=4$).}
	\label{fig:05}
\end{figure}

\begin{table}[!ht]
\footnotesize
	\centering
	\caption{\small Null rejection rates for $H_0: \beta_1 = \cdots = \beta_r = 0$ in the log-CN quantile regression model ($\vartheta_1=0.1,\vartheta_2=0.2$).}
	\adjustbox{max height=\dimexpr\textheight-3.5cm\relax,
		max width=\textwidth}{
		\begin{tabular}{ccccccccccccccc}
			\toprule
			\multirow{2}{*}{$t$}&\multirow{2}{*}{$q$}& & \multicolumn{3}{c}{$n = 50$} & &\multicolumn{3}{c}{$n = 100$} & &\multicolumn{3}{c}{$n = 200$}\\
			\cline{4-6} \cline{8-10} \cline{12-14}
			&& & 1\% & 5\% & 10\% & & 1\% & 5\% & 10\% & & 1\% & 5\% & 10\%\\
			\hline
			\multirow{12}{*}{3} & \multirow{4}{*}{0.25}& $S_W$ & 0.0312 & 0.0906 & 0.1432 && 0.0164 & 0.0674 & 0.1206 && 0.0136 & 0.0628 & 0.1120\\
			&& $S_{LR}$ & 0.0144 & 0.0706 & 0.117 && 0.0122 & 0.0492 & 0.1032 && 0.0102 & 0.0488 & 0.1008 \\
			&& $S_R$ &  0.0274 & 0.0724 & 0.1132 && 0.0218 & 0.0672 & 0.1154 && 0.0288 & 0.0802 & 0.1298 \\
			&& $S_T$ & 0.0118 & 0.083 & 0.138 && 0.0162 & 0.0618 & 0.1178 && 0.026 & 0.0766 & 0.1372\\
			\cline{2-14}
			&\multirow{4}{*}{0.5}& $S_W$ & 0.0264 & 0.0894 & 0.1436 && 0.0158 & 0.0682 & 0.1208 && 0.0134 & 0.0606 & 0.1118\\
			&& $S_{LR}$ & 0.0146 & 0.0646 & 0.122 && 0.0112 & 0.0562 & 0.1134 && 0.0122 & 0.0546 & 0.0998 \\
			&& $S_R$ &  0.0274 & 0.0724 & 0.1166 && 0.0218 & 0.0672 & 0.1154 && 0.0288 & 0.0802 & 0.1298 \\
			&& $S_T$ & 0.0206 & 0.083 & 0.138 && 0.0162 & 0.0618 & 0.1178 && 0.0258 & 0.0766 & 0.1372\\
			\cline{2-14}
			&\multirow{4}{*}{0.75} & $S_W$ & 0.0308 & 0.0894 & 0.1458 && 0.0160 & 0.0666 & 0.1200 && 0.0132 & 0.0622 & 0.1114\\
			&& $S_{LR}$ & 0.0146 & 0.0644 & 0.1222 && 0.0122 & 0.055 & 0.1152 && 0.0128 & 0.0506 & 0.1076 \\
			&& $S_R$ &  0.0274 & 0.0616 & 0.1166 && 0.0218 & 0.0672 & 0.1154 && 0.0288 & 0.0802 & 0.1298 \\
			&& $S_T$ & 0.024 & 0.0832 & 0.1378 && 0.0162 & 0.0618 & 0.1178 && 0.0258 & 0.0766 & 0.1372 \\
			\bottomrule
			\multirow{12}{*}{1} & \multirow{4}{*}{0.25}& $S_W$ & 0.0146 & 0.0620 & 0.1166 && 0.0124 & 0.061 & 0.1102 && 0.0102 & 0.0546 & 0.1038 \\
			&& $S_{LR}$ & 0.0498 & 0.1298 & 0.1982 && 0.024 & 0.0946 & 0.148 && 0.021 & 0.0592 & 0.0988 \\
			&& $S_R$ &  0.01 & 0.0496 & 0.1028 && 0.009 & 0.0496 & 0.0988 && 0.0104 & 0.0572 & 0.1166 \\
			&& $S_T$ & 0.0224 & 0.0748 & 0.1372 && 0.0154 & 0.0682 & 0.1208 && 0.0186 & 0.0756 & 0.135 \\
			\cline{2-14}
			&\multirow{4}{*}{0.5}& $S_W$ & 0.0166 & 0.0614 & 0.1156 && 0.0128 & 0.0618 & 0.1108 && 0.0108 & 0.0546 & 0.1054\\
			&& $S_{LR}$ & 0.0366 & 0.1298 & 0.1982 && 0.024 & 0.0946 & 0.148 && 0.021 & 0.0592 & 0.0988 \\
			&& $S_R$ & 0.0104 & 0.0496 & 0.1028 && 0.009 & 0.0496 & 0.0988 && 0.0104 & 0.0572 & 0.1166 \\
			&& $S_T$ & 0.0194 & 0.0748 & 0.1372 && 0.0154 & 0.0682 & 0.121 && 0.0186 & 0.0756 & 0.1352 \\
			\cline{2-14}
			&\multirow{4}{*}{0.75} & $S_W$ & 0.017 & 0.0632 & 0.1172 && 0.0132 & 0.0634 & 0.1106 && 0.0104 & 0.055 & 0.1064 \\
			&& $S_{LR}$ & 0.0366 & 0.1298 & 0.1982 && 0.024 & 0.0946 & 0.148 && 0.021 & 0.0592 & 0.0988 \\
			&& $S_R$ & 0.0104 & 0.0496 & 0.1028 && 0.009 & 0.0496 & 0.0988 && 0.0104 & 0.0572 & 0.1166 \\
			&& $S_T$ & 0.019 & 0.0748 & 0.1372 && 0.0154 & 0.0682 & 0.121 && 0.0186 & 0.0754 & 0.1352 \\
			\bottomrule
	\end{tabular}}
\label{tab:17}
\end{table}

\begin{figure}[!ht]
	\centering
	\subfigure[$r=3$]{
		\includegraphics[scale = 0.5]{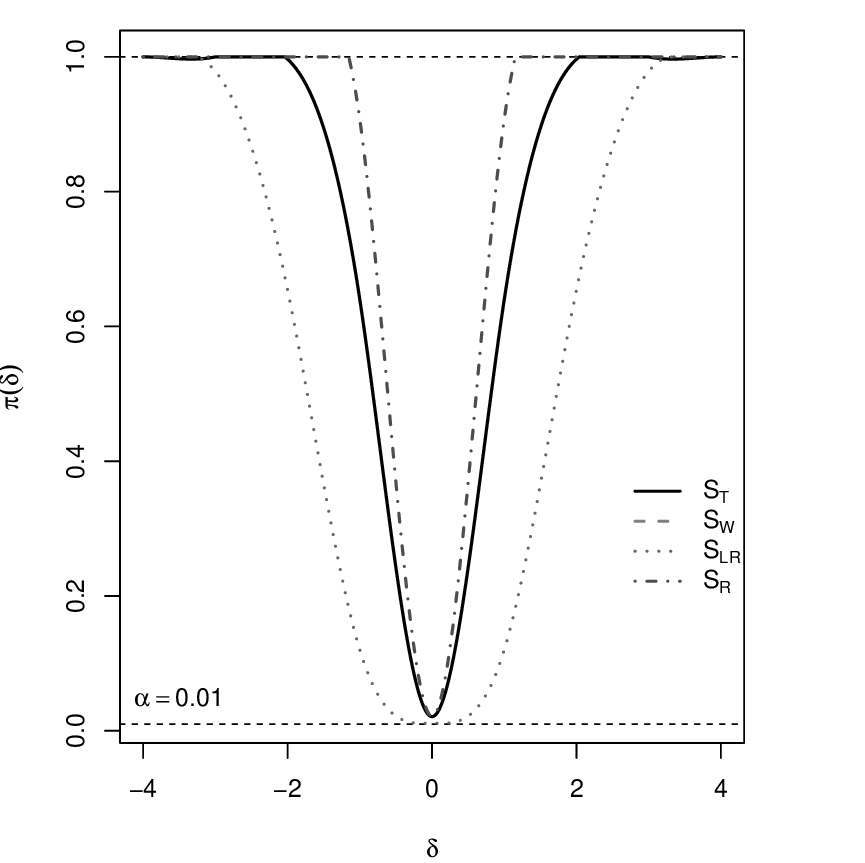}
	}
	\subfigure[$r=1$]{
		\includegraphics[scale = 0.5]{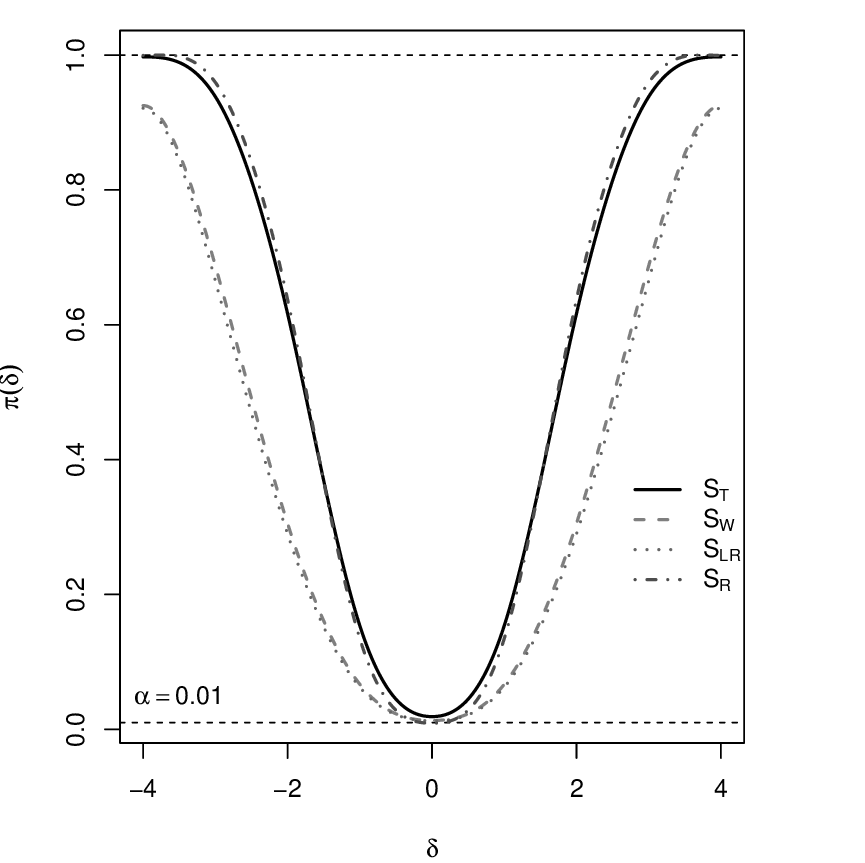}
	}
	\caption{\small Power curve of the tests in the log-CN quantile regression model ($\vartheta_1=0.1,\vartheta_2=0.2$).}
\label{fig:06}
\end{figure}

\begin{table}[!ht]
\footnotesize
	\centering
	\caption{\small Null rejection rates for $H_0: \beta_1 = \cdots = \beta_r = 0$ in the EBS quantile regression model ($\vartheta=0.5$).}
	\adjustbox{max height=\dimexpr\textheight-3.5cm\relax,
		max width=\textwidth}{
		\begin{tabular}{ccccccccccccccc}
			\toprule
			\multirow{2}{*}{$t$}&\multirow{2}{*}{$q$}& & \multicolumn{3}{c}{$n = 50$} & &\multicolumn{3}{c}{$n = 100$} & &\multicolumn{3}{c}{$n = 200$}\\
			\cline{4-6} \cline{8-10} \cline{12-14}
			&&  & 1\% & 5\% & 10\% & & 1\% & 5\% & 10\% & & 1\% & 5\% & 10\%\\
			\hline
			\multirow{12}{*}{3}&\multirow{4}{*}{0.25}& $S_W$ &  0.0222 & 0.0756 & 0.1384 && 0.015 & 0.063 & 0.118 && 0.0126 & 0.0576 & 0.106 \\
			&& $S_{LR}$ & 0.0118 & 0.0612 & 0.1194 && 0.0114 & 0.0576 & 0.111 && 0.011 & 0.0532 & 0.1012 \\
			&& $S_R$ & 0.0272 & 0.0728 & 0.1208 && 0.0266 & 0.072 & 0.114 && 0.0252 & 0.0694 & 0.1202 \\
			&& $S_T$ & 0.0192 & 0.0684 & 0.1246 && 0.0222 & 0.0696 & 0.1226 && 0.019 & 0.072 & 0.127 \\
			\cline{2-14}
			&\multirow{4}{*}{0.5}& $S_W$ & 0.0222 & 0.0756 & 0.1386 && 0.015 & 0.0632 & 0.1176 && 0.0126 & 0.0576 & 0.106\\
			&& $S_{LR}$ & 0.0162 & 0.0612 & 0.1194 && 0.0114 & 0.0576 & 0.111 && 0.011 & 0.0532 & 0.1012 \\
			&& $S_R$ & 0.0262 & 0.0728 & 0.1208 && 0.0266 & 0.072 & 0.114 && 0.0252 & 0.0694 & 0.1202 \\
			&& $S_T$ & 0.0206 & 0.0684 & 0.1246 && 0.0222 & 0.0696 & 0.1226 && 0.019 & 0.072 & 0.1272 \\
			\cline{2-14}
			&\multirow{4}{*}{0.75} & $S_W$ & 0.0222 & 0.0756 & 0.1386 && 0.015 & 0.0632 & 0.1178 && 0.0126 & 0.0576 & 0.1062 \\
			&& $S_{LR}$ & 0.0142 & 0.0612 & 0.1194 && 0.0114 & 0.0576 & 0.111 && 0.011 & 0.0532 & 0.1012 \\
			&& $S_R$ & 0.0262 & 0.0728 & 0.1208 && 0.0266 & 0.072 & 0.114 && 0.0252 & 0.0694 & 0.1202 \\
			&& $S_T$ & 0.0206 & 0.0684 & 0.1246 && 0.0222 & 0.0696 & 0.1226 && 0.019 & 0.072 & 0.1272 \\
			\bottomrule
			\multirow{12}{*}{1}&\multirow{4}{*}{0.25}& $S_W$ & 0.0242 & 0.0232 & 0.0538 && 0.0132 & 0.0404 & 0.0846 && 0.0068 & 0.0448 & 0.0926 \\
			&& $S_{LR}$ & 0.0114 & 0.0602 & 0.1096 && 0.0092 & 0.0526 & 0.097 && 0.008 & 0.052 & 0.0978 \\
			&& $S_R$ &  0.0106 & 0.0474 & 0.0948 && 0.0118 & 0.0536 & 0.1014 && 0.0118 & 0.0522 & 0.0988 \\
			&& $S_T$ & 0.0306 & 0.074 & 0.1216 && 0.0256 & 0.0818 & 0.1386 && 0.0162 & 0.0718 & 0.1292 \\
			\cline{2-14}
			&\multirow{4}{*}{0.5}& $S_W$ & 0.0046 & 0.0232 & 0.0536 && 0.0134 & 0.0644 & 0.0848 && 0.0036 &  0.045 & 0.0926\\
			&& $S_{LR}$ & 0.012 & 0.0604 & 0.1096 && 0.0112 & 0.0524 & 0.098 && 0.0082 & 0.0516 & 0.0988 \\
			&& $S_R$ & 0.0116 & 0.0474 & 0.0948 && 0.0118 & 0.0536 & 0.1014 && 0.0118 & 0.0522 & 0.0988 \\
			&& $S_T$ & 0.0206 & 0.074 & 0.1216 && 0.0256 & 0.0818 & 0.1386 && 0.0162 & 0.072 & 0.1292 \\
			\cline{2-14}
			&\multirow{4}{*}{0.75} & $S_W$ & 0.0046 & 0.0232 & 0.0536 && 0.0134 & 0.0644 & 0.1216 && 0.0036 & 0.0446 & 0.0928 \\
			&& $S_{LR}$ & 0.015 & 0.0608 & 0.1108 && 0.011 & 0.0532 & 0.0958 && 0.008 & 0.0516 & 0.1 \\
			&& $S_R$ & 0.0116 & 0.0474 & 0.0948 && 0.0118 & 0.0536 & 0.1014 && 0.0118 & 0.0522 & 0.0988 \\
			&& $S_T$ & 0.0206 & 0.074 & 0.1216 && 0.0256 & 0.0818 & 0.1386 && 0.0162 & 0.0722 & 0.1292 \\
			\bottomrule
	\end{tabular}}
\label{tab:18}
\end{table}

\begin{figure}[!ht]
	\centering
	\subfigure[$r=3$]{
		\includegraphics[scale = 0.5]{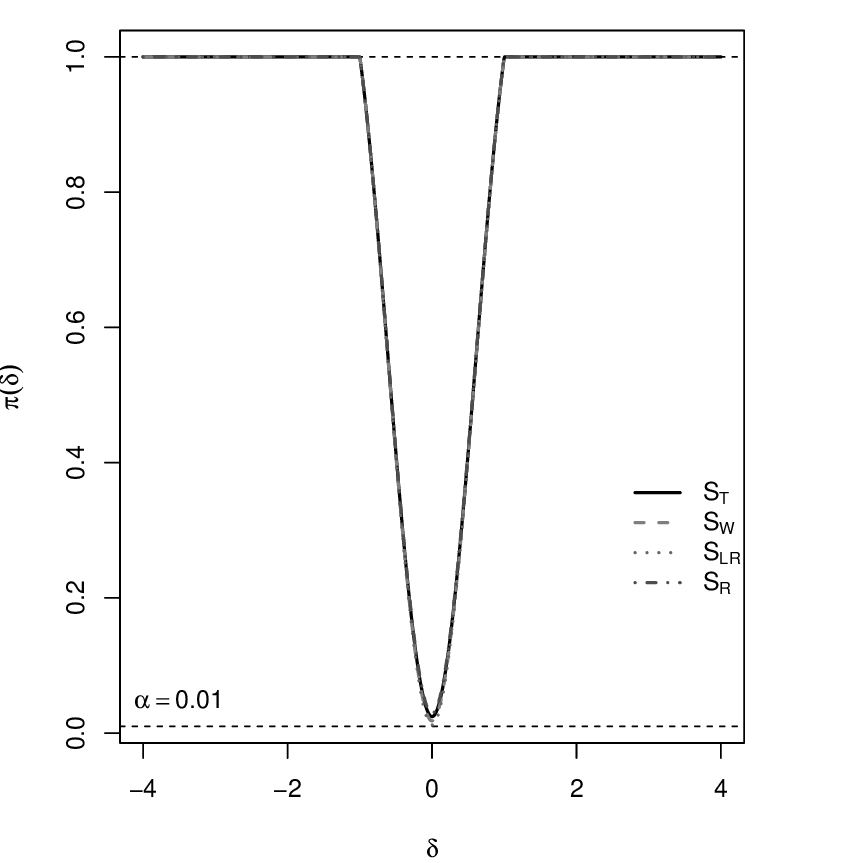}
	}
	\subfigure[$r=1$]{
		\includegraphics[scale = 0.5]{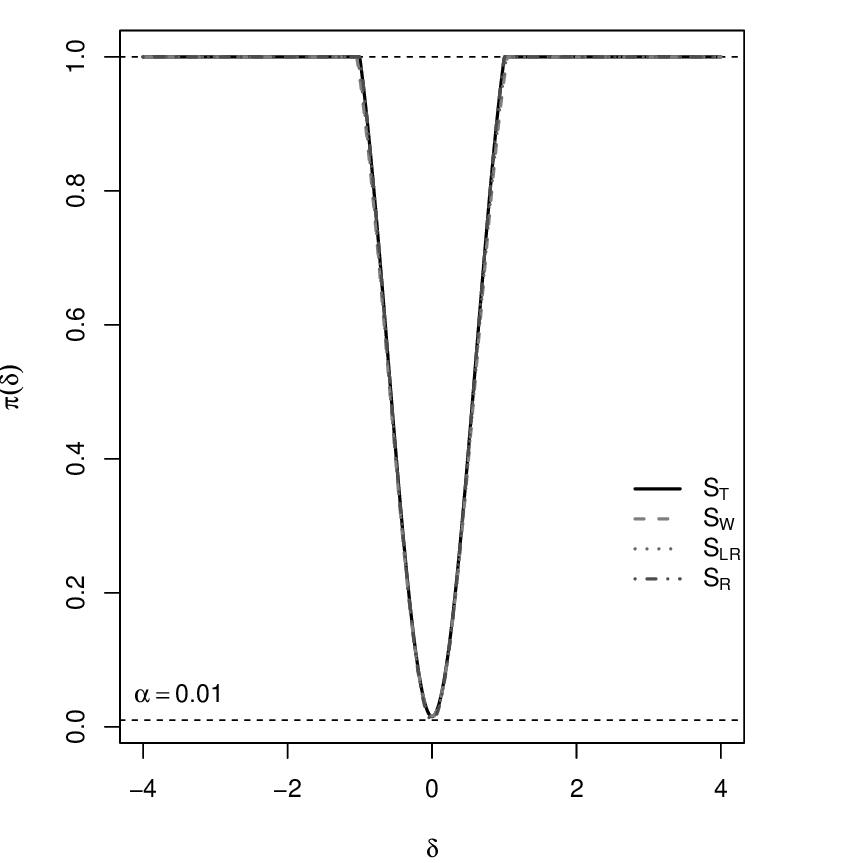}
	}
	\caption{\small Power curve of the tests in the EBS quantile regression model ($\vartheta=0.5$).}	
\label{fig:07}
\end{figure}

\begin{table}[!ht]
\footnotesize
	\centering
	\caption{\small Null rejection rates for $H_0: \beta_1 = \cdots = \beta_r = 0$ in the EBS-$t$ quantile regression model ($\vartheta_1=0.5,\vartheta_2=3$).}
	\adjustbox{max height=\dimexpr\textheight-3.5cm\relax,
		max width=\textwidth}{
		\begin{tabular}{ccccccccccccccc}
			\toprule
			\multirow{2}{*}{$t$} & \multirow{2}{*}{$q$}& & \multicolumn{3}{c}{$n = 50$} & &\multicolumn{3}{c}{$n = 100$} & &\multicolumn{3}{c}{$n = 200$}\\
			\cline{4-6} \cline{8-10} \cline{12-14}
			&&  & 1\% & 5\% & 10\% & & 1\% & 5\% & 10\% & & 1\% & 5\% & 10\%\\
			\hline
			\multirow{12}{*}{3} & \multirow{4}{*}{0.25}& $S_W$ & 0.0256 & 0.0806 & 0.1382 && 0.0226 & 0.0718 & 0.1292 && 0.0114 & 0.0578 & 0.1102 \\
			&& $S_{LR}$ & 0.0142 & 0.0712 & 0.1288 && 0.0122 & 0.0586 & 0.1108 && 0.01 & 0.0556 & 0.0992 \\
			&& $S_R$ & 0.0264 & 0.0708 & 0.1146 && 0.0262 & 0.0624 & 0.105 && 0.0296 & 0.0746 & 0.1178 \\
			&& $S_T$ & 0.0194 & 0.0642 & 0.12 && 0.0234 & 0.068 &  0.1262 && 0.0204 & 0.067 & 0.1192 \\
			\cline{2-14}
			&\multirow{4}{*}{0.5}& $S_W$ & 0.0234 & 0.0806 & 0.138 && 0.0228 & 0.0718 & 0.1288 && 0.011 & 0.0582 & 0.1106 \\
			&& $S_{LR}$ & 0.0138 & 0.0674 & 0.1144 && 0.0126 & 0.056 & 0.1064 && 0.0104 & 0.0534 & 0.0946 \\
			&& $S_R$ & 0.0264 & 0.0708 & 0.1146 && 0.0262 & 0.0624 & 0.105 && 0.0296 & 0.0746 & 0.1178 \\
			&& $S_T$ & 0.0142 & 0.0644 & 0.12 && 0.0234 & 0.068 &  0.1262 && 0.0202 & 0.067 & 0.1192 \\
			\cline{2-14}
			&\multirow{4}{*}{0.75} & $S_W$ & 0.0236 & 0.0806 & 0.1382 && 0.0228 & 0.0718 & 0.1292 && 0.0112 & 0.058 & 0.1108 \\
			&& $S_{LR}$ & 0.0144 & 0.0608 & 0.1144 && 0.0142 & 0.0556 & 0.11 && 0.0128 & 0.051 & 0.104 \\
			&& $S_R$ & 0.0264 & 0.0708 & 0.1146 && 0.0262 & 0.0624 & 0.105 && 0.0296 & 0.0746 & 0.1178 \\
			&& $S_T$ & 0.0142 & 0.0642 & 0.12 && 0.0234 & 0.068 &  0.1244 && 0.0202 & 0.067 & 0.1192 \\
			\bottomrule
			\multirow{12}{*}{1} & \multirow{4}{*}{0.25}& $S_W$ & 0.0124 & 0.053 & 0.106 && 0.0192 & 0.0434 & 0.0906 && 0.0074 & 0.038 & 0.0816 \\
			&& $S_{LR}$ & 0.0148 & 0.0624 & 0.117 && 0.0122 & 0.057 & 0.1088 && 0.0108 & 0.0554 & 0.101 \\
			&& $S_R$ &  0.012 & 0.0496 & 0.1026 && 0.011 & 0.047 & 0.0958 && 0.011 & 0.0498 & 0.098\\
			&& $S_T$ & 0.0244 & 0.0742 & 0.1274 && 0.027 & 0.0788 &  0.1294 && 0.0174 & 0.0698 & 0.1234\\
			\cline{2-14}
			&\multirow{4}{*}{0.5}& $S_W$ & 0.0146 & 0.0566 & 0.1062 && 0.0192 & 0.0434 & 0.0906 && 0.0072 & 0.038 & 0.0814 \\
			&& $S_{LR}$ & 0.013 & 0.0624 & 0.117 && 0.0122 & 0.0572 & 0.1088 && 0.0108 & 0.0554 & 0.101 \\
			&& $S_R$ & 0.012 & 0.0496 & 0.1026 && 0.0104 & 0.047 & 0.0958 && 0.011 & 0.0498 & 0.098 \\
			&& $S_T$ & 0.0186 & 0.0742 & 0.1274 && 0.0272 & 0.0788 
			&  0.1294 && 0.0174 & 0.067 & 0.1234 \\
			\cline{2-14}
			&\multirow{4}{*}{0.75} & $S_W$ & 0.015 & 0.0568 & 0.1066 && 0.0192 & 0.0432 & 0.091 && 0.0074 & 0.0378 & 0.0814 \\
			&& $S_{LR}$ & 0.013 & 0.0624 & 0.117 && 0.0122 & 0.057 & 0.1088 && 0.0108 & 0.0554 & 0.101 \\
			&& $S_R$ & 0.012 & 0.0496 & 0.1026 && 0.0104 & 0.047 & 0.0958 && 0.011 & 0.0498 & 0.098 \\
			&& $S_T$ & 0.0186 & 0.0742 & 0.1272 && 0.0272 & 0.0788 
			& 0.1294 && 0.0174 & 0.0698 & 0.1236 \\
			\bottomrule
	\end{tabular}}
\label{tab:19}
\end{table}

\begin{figure}[!ht]
	\centering
	\subfigure[$r=3$]{
		\includegraphics[scale = 0.5]{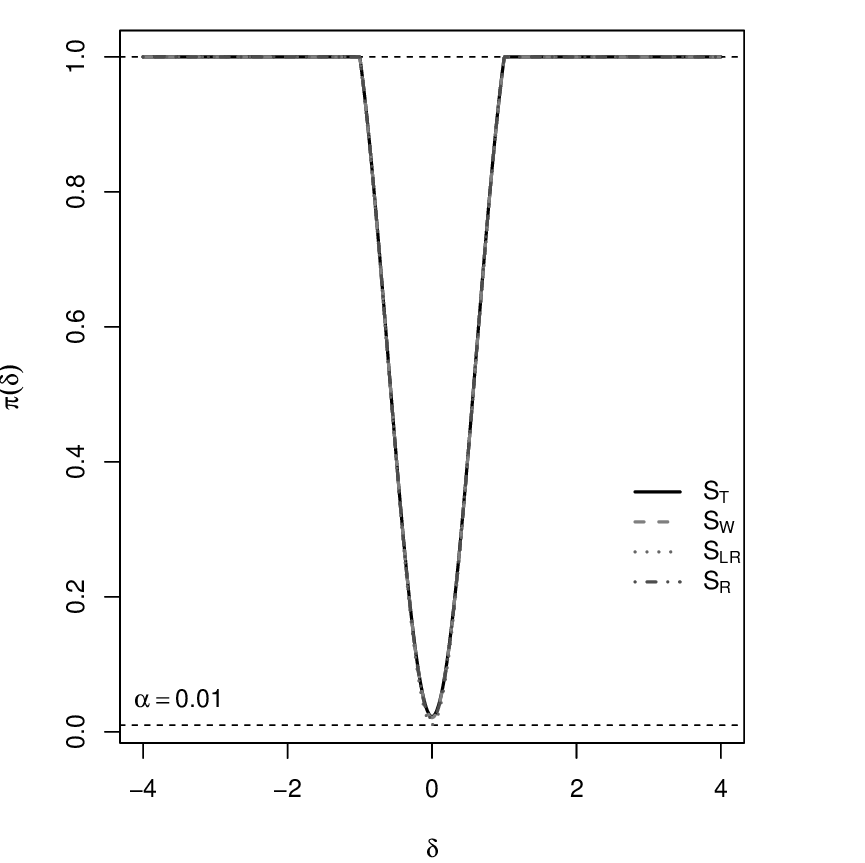}
	}
	\subfigure[$r=1$]{
		\includegraphics[scale = 0.5]{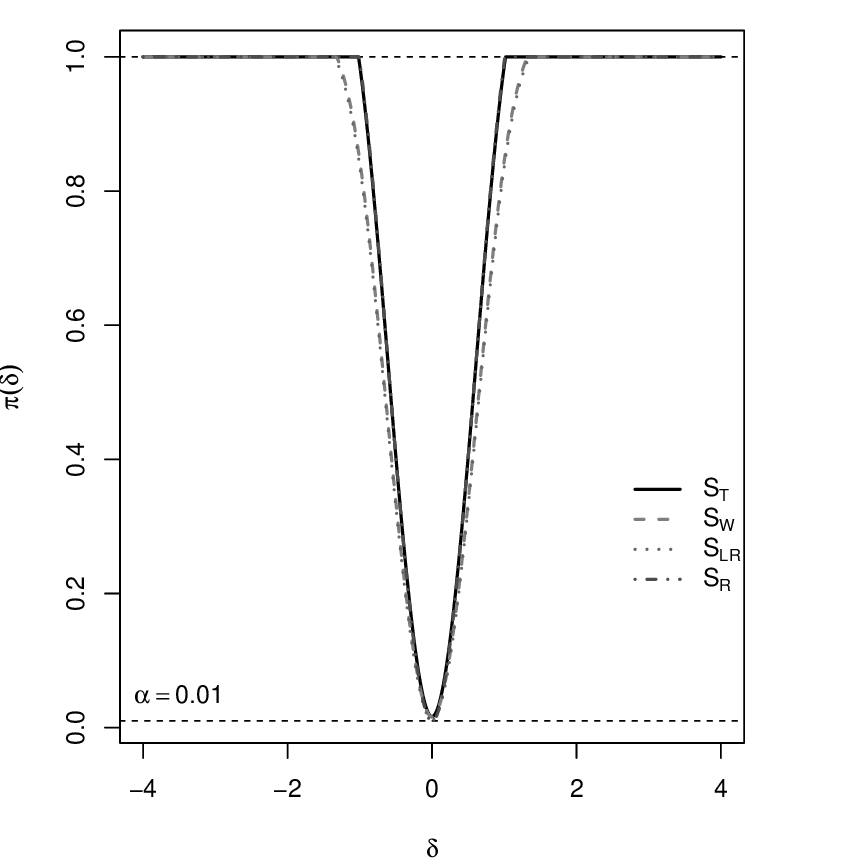}
	}
	\caption{\small Power curve of the tests in the EBS-$t$ quantile regression model ($\vartheta_1=0.5,\vartheta=3$).}	
\label{fig:08}
\end{figure}

\clearpage

\section{Application to real data}\label{sec:5}

The proposed log-symmetric quantile regression models are now used to analyze a web-scraped data set on movies from 1929 to 2016. This data set was constructed by \cite{vslm:18} from the internet movie database (IMDb) [\texttt{https://www.imdb.com}], and contains 250 cataloged movies, however, due to the fact that some old movies not produced in the USA do not have their respective financial information, the analysis of the data set proceeded with 155 movies. The variables considered in the study were: $Y$, box office, in millions of US\$; $x_1$, number of votes, in millions; $ x_2 $, expense, in thousands of US\$ ; and $ x_3 $, rating, according to users. Here, $Y$ is the response and $x_1$, $x_2$ and $x_3$ are the respective covariates.

Table \ref{tab:desc} reports descriptive statistics of the observed box offices (in millions of US\$), including the mean, median, minimum, maximum, SD, CS, CK and coefficient of variation (CV) values. From this table, we observe that the median and mean of the box offices are respectively 60.995 and 127.915, namely, the mean is greater than the median which indicates a positively skewed feature in the data. Moreover, CV is 126.29\%, which indicates a high level of dispersion around the mean. We also observe that the CS value confirms the skewed nature and the CK value indicates a high kurtosis feature in the data.

\begin{table}[H]
\footnotesize
	\caption{\small Summary statistics for the box office data, in millions of US\$.}
	\centering
	\begin{tabular}{ccccccccc}
	 \toprule
	  $n$ & Min. &  Median & Mean & Max. & SD & CS & CK & CV\\
	  \hline
	  155 & 0.1149 &  60.9947 & 127.3915 & 936.6274 & 160.8849 & 1.9370 & 4.4831 & 126.2916\\
     \bottomrule
	\end{tabular}
\label{tab:desc}
\end{table}

\cite{vslm:18} analyzed the box office data using log-symmetric regression models \citep{vp:16a}, which according to the authors is justified due to the asymmetric nature of the data. However, as explained earlier, quantile regression provides a richer characterization of the effects of covariates on the response. Thus, we consider the proposed log-symmetric quantile regression model as a more general alternative to {the log-symmetric regression models by \cite{vp:16a}}. We then analyze the box office data using the proposed model, expressed as 
\begin{equation}
Y_i = (\beta_0 + \beta_1 x_{1i} + \beta_2 x_{2i} + \beta_3 x_{3i})\varepsilon_{i}^{\sqrt{\tau_0 + \tau_1 x_{1i} + \tau_2 x_{2i} + \tau_3 x_{3i}}}, \quad i = 1,2, \ldots, 155, \nonumber
\end{equation}
where $\epsilon_i\sim\text{QLS}(1,1,g)$. 

Initially, we have to find the best model amongst the quantile regression models based the log-NO,
log-$t$,
log-PE,
log-HP,
log-SL,
log-CN,
EBS and 
EBS-$t$ distributions. In this sense, the information criteria (AIC, BIC and AICc) and the root of the mean square error (RMSE) of the prediction, are used. The idea is to fit the log-symmetric regression models for each $q = 0.01, 0.02, \ldots, 0.98, 0.99$, and then to compute the mean of the corresponding AIC, BIC, AICc and RMSE values obtained from the $q$s. Table \ref{tab:criterions} reports the results and we initially observe that the model with a log-CN distribution has the lowest values for the information criteria. However, for the log-$t$, log-PE, log-HP and log-SL distributions, the information criteria values are very close to the log-CN ones. Moreover, we observe that the models based on the log-NO, EBS and log-HP distributions have the lowest RMSE values, indicating that they are the most accurate models among those analyzed, in terms of prediction. Therefore, considering the results of the information criteria and the RMSE, the best log-symmetric quantile regression model is the one based on the log-HP distribution.

Figure \ref{fig:estimates} plots the estimated parameters in the log-HP quantile regression model across $q$. Results suggest that the box office data display asymmetric dynamics: the estimates of $\widehat{\beta}_0$, $\widehat{\beta}_1$ and $\widehat{\beta}_2$ ($\widehat{\beta}_3$) tend to decrease (increase) as $q$ increases, with a change of sign in $\widehat{\beta}_0$, $\widehat{\beta}_1$ and $\widehat{\beta}_3$. On the other hand, the estimates of $\widehat{\tau}_0$ and $\widehat{\tau}_1$ ($\widehat{\tau}_3$) tend initially to decrease (increase) as $q$ increases and from $q=0.9$ these estimates tend to increase (decrease). The {pattern} for the estimates of $\widehat{\tau}_2$ is similar to $\widehat{\tau}_0$ and $\widehat{\tau}_1$, but the change in the trend is close to $q=0.99$. Table \ref{tab:estimates} presents the maximum likelihood estimates and standard errors for the log-HP quantile regression model parameters with $q = 0.25, 0.5, 0.75$. The estimate of $\widehat{\vartheta}$ was equal to 1 for all values of $q$. {We can interpret the estimated coefficients in terms of their effect on the response $Y_i$ (box offices, in millions of US\$); see \citet[p.82]{weisberg:14} for similar interpretation. For example, an increase in the number of votes by 1 million, increases the $25^{\circ}$ percentile ($q=0.25$) of the box office by $(\exp(2.3891)-1)\times 100\%=990.37\%$, whereas the increase in the $75^{\circ}$ percentile ($q=0.75$) is of $(\exp(1.2510)-1)\times 100\%=249.39\%$. Therefore, the effect of the number of votes on the box office is smaller for movies with bigger box offices (larger quantiles). We also observe that the budget behaves similarly to the number of votes, that is, the effect on the box office is smaller for movies with bigger budgets. Finally, we note that the effect of the rating on the box office is negative, where movies with smaller box offices (smaller quantiles) have greater negative effects.}

\begin{table}[!ht]
\footnotesize
	\centering
	\caption{\small Means of AIC, BIC, AICc and RMSE values based on $q = 0.01,0.02, \ldots, 0.98, 0.99$ for different log-symmetric quantile regression models.}
	\adjustbox{max height=\dimexpr\textheight-3.5cm\relax,
		max width=\textwidth}{
		\begin{tabular}{lcccccccc}
			\toprule
			Criterion & Log-NO & Log-$t$ & Log-PE & Log-HP & Log-SL & Log-CN & EBS & EBS-$t$ \\
			\hline
			AIC & 514.0027 & 499.4266 & 500.5683 & 501.1267 & 499.3116 & 497.9390 & 514.1874 & 499.4132 \\
			BIC & 538.3501 & 523.7740 & 524.9157 & 525.4741 &  523.6590 & 522.2864 & 538.5348 & 523.7606 \\ 
			AICc & 514.9890 & 500.4129 & 501.5546 & 502.1130 &  500.2979 & 498.9253 & 515.1737 & 500.3995 \\
			RMSE & 191.7217 & 225.3356 & 212.2050 & 198.2404 &  239.7783 & 230.8775 & 192.0232 & 233.8590 \\
			\bottomrule
	\end{tabular}}
		\label{tab:criterions}
\end{table}

\begin{figure}[!ht]
	\centering
	\subfigure[$\widehat{\beta}_0$]{\includegraphics[scale = 0.26]{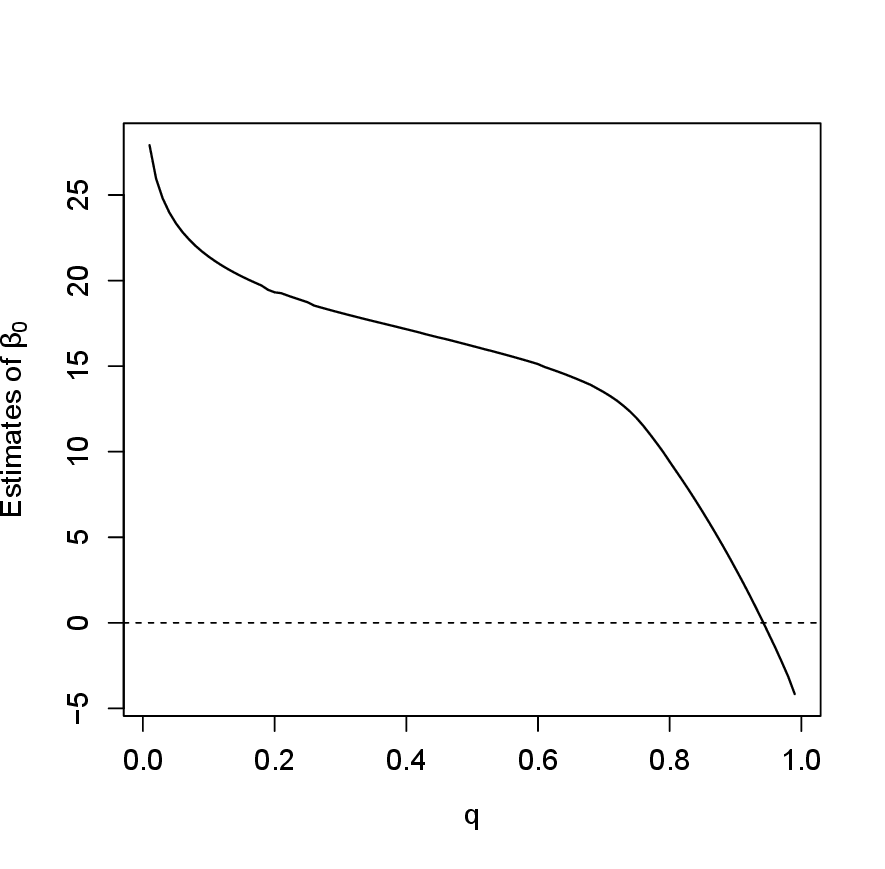}}
	\subfigure[$\widehat{\beta}_1$]{\includegraphics[scale = 0.26]{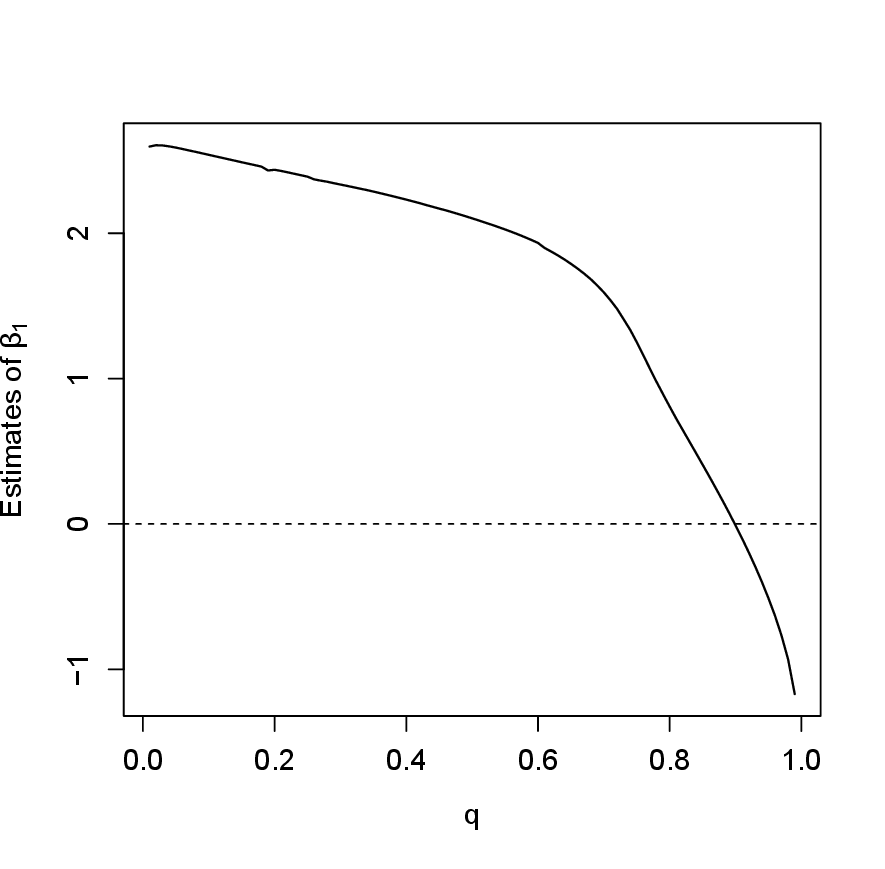}}
	\subfigure[$\widehat{\beta}_2$]{\includegraphics[scale = 0.26]{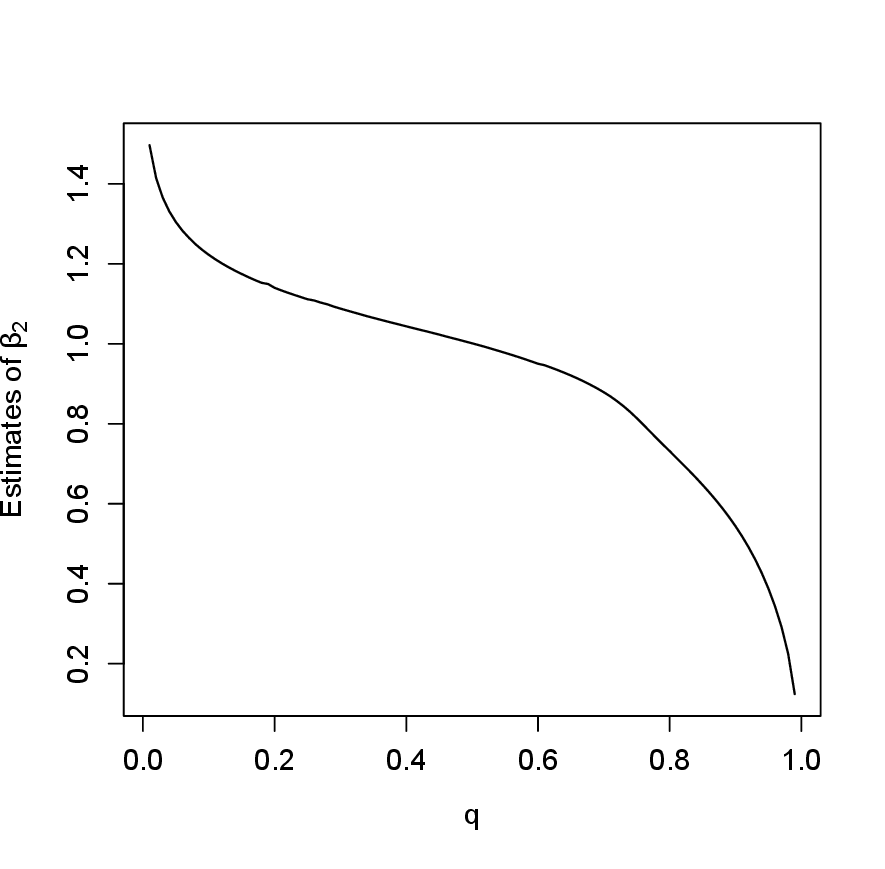}}
	\subfigure[$\widehat{\beta}_3$]{\includegraphics[scale = 0.26]{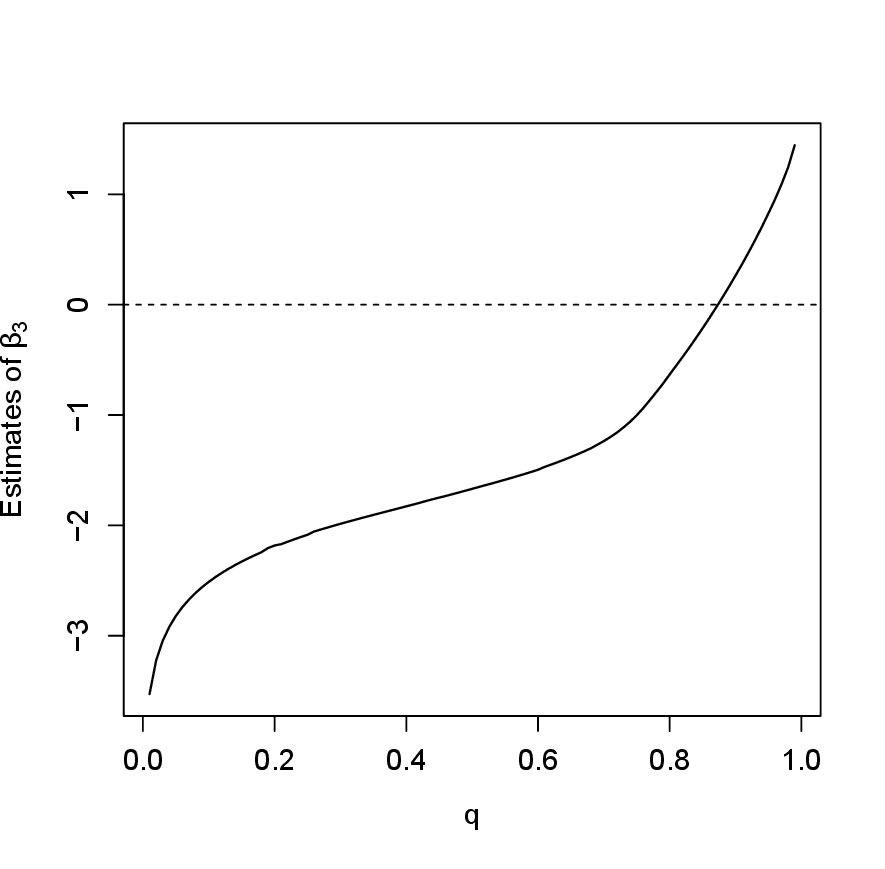}}\\
	\subfigure[$\widehat{\tau}_0$]{\includegraphics[scale = 0.26]{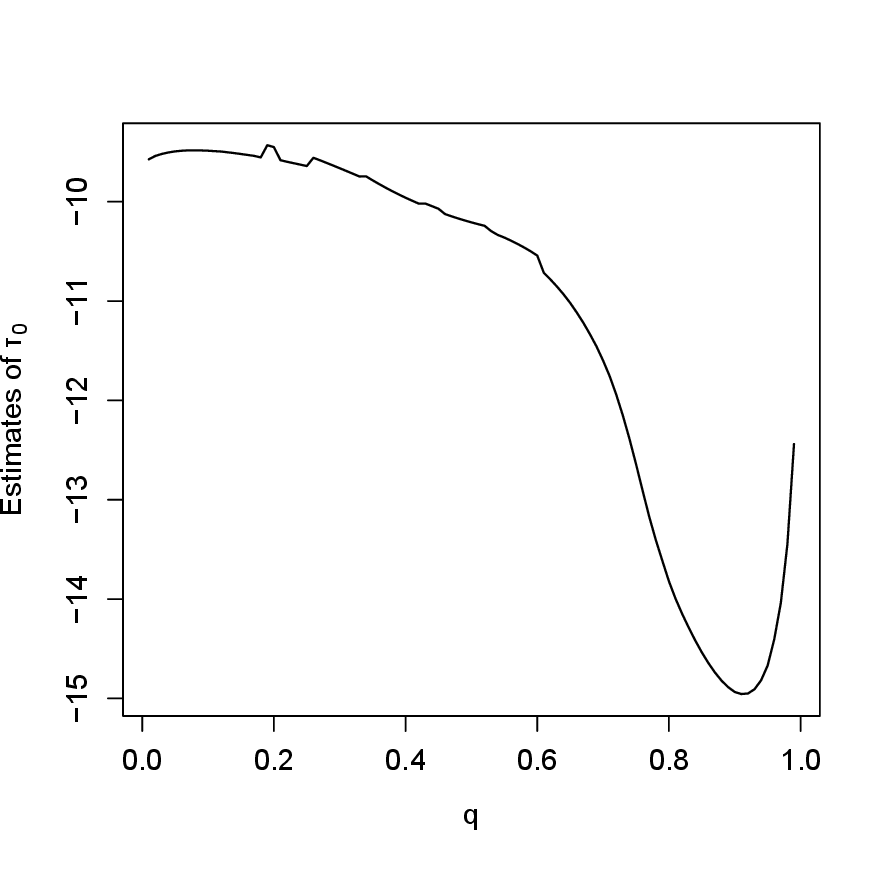}}
	\subfigure[$\widehat{\tau}_1$]{\includegraphics[scale = 0.26]{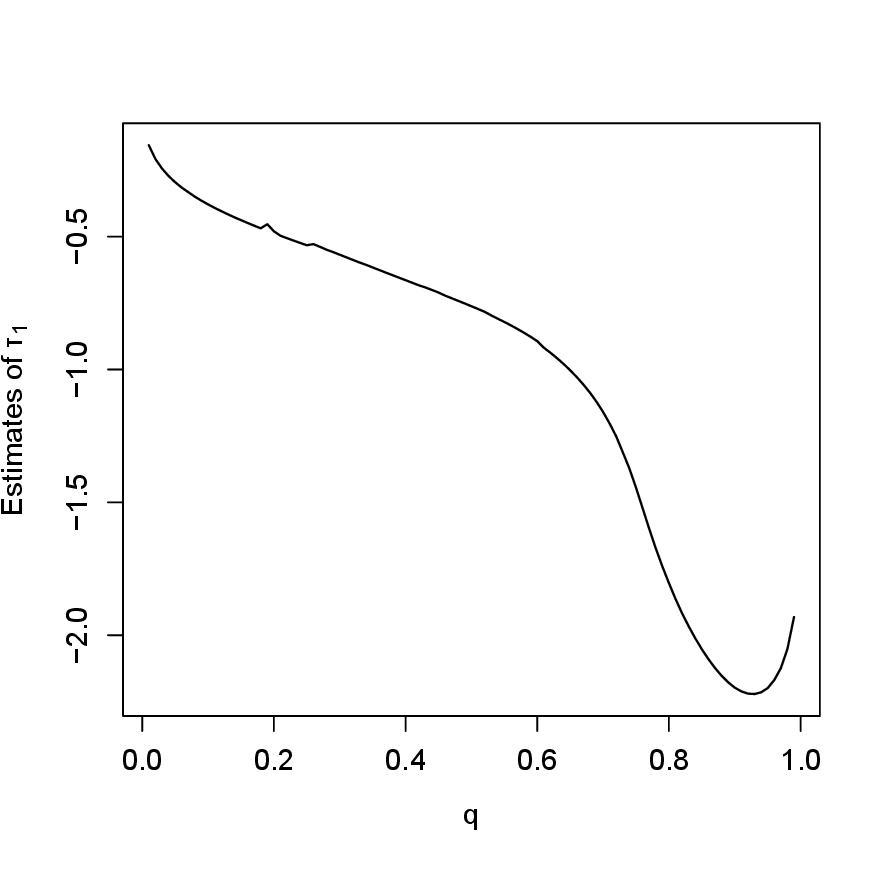}}
	\subfigure[$\widehat{\tau}_2$]{\includegraphics[scale = 0.26]{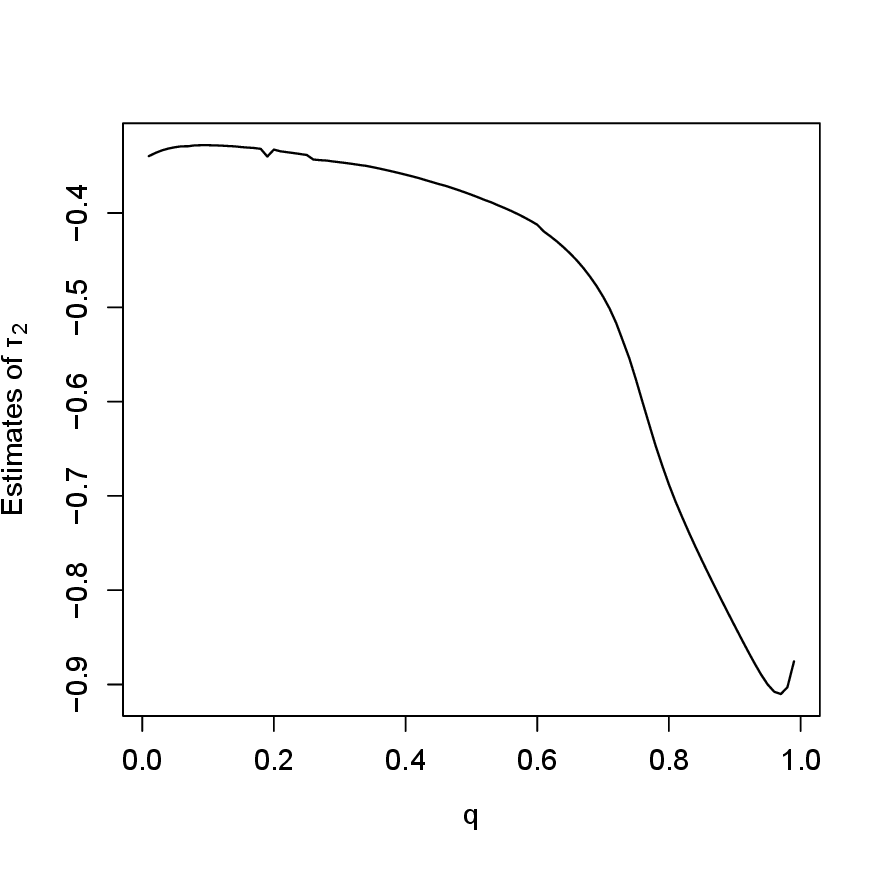}}
	\subfigure[$\widehat{\tau}_3$]{\includegraphics[scale = 0.26]{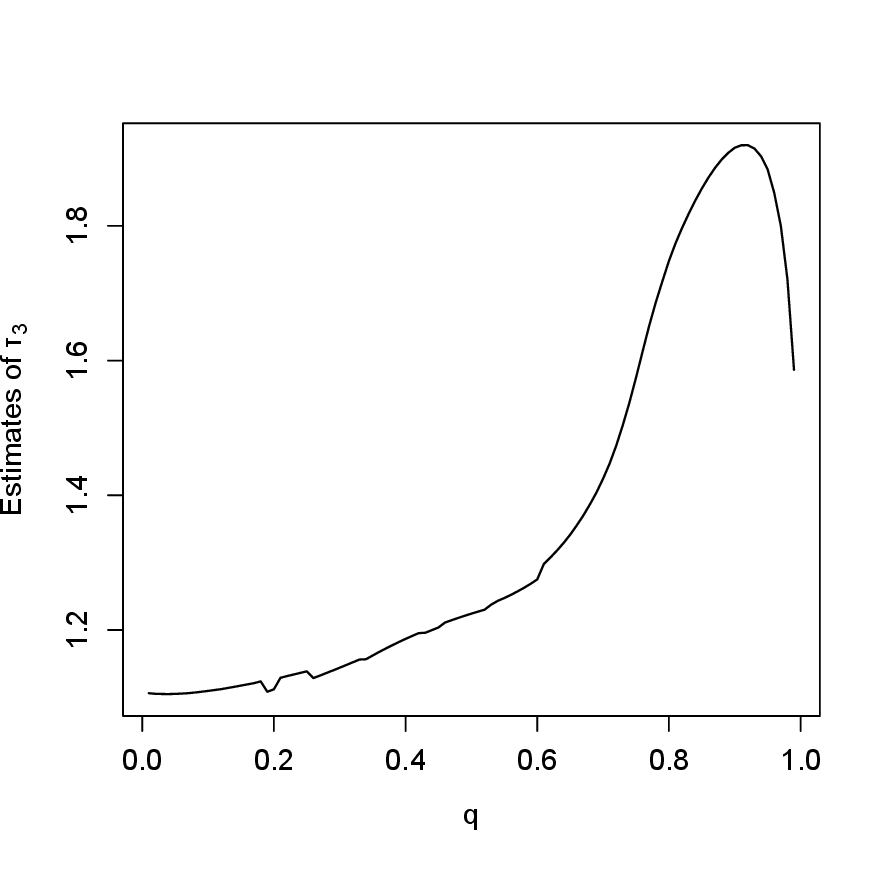}}
	\caption{\small Estimated parameters in the log-HP quantile regression model across $q$ for the box office data.}	
	\label{fig:estimates}
\end{figure}

	\begin{table}[!ht]
	\footnotesize
		\centering
		\caption{\small Maximum likelihood estimates (with standard errors in parentheses) for the log-HP quantile regression model ($\widehat{\vartheta} = 1$).}
		\adjustbox{max height=\dimexpr\textheight-3.5cm\relax,
		max width=\textwidth}{
		\begin{tabular}{cccccc}
			\toprule
			\multirow{2}{*}{$q$} & & \multicolumn{4}{c}{Estimate (Standard error)} \\
			\cline{3-6}
			& & Intercept ($\widehat{\beta}_0/\widehat{\tau}_0$) & Votes ($\widehat{\beta}_1/\widehat{\tau}_1$) & Budget ($\widehat{\beta}_2/\widehat{\tau}_2$) & Rating ($\widehat{\beta}_3/\widehat{\tau}_3$) \\
			\hline
			\multirow{2}{*}{0.25} & $\widehat{\beta}$ & 18.7352 (4.4784) & 2.3891 (0.4168) & 1.1115 (0.1716) & -2.0848 (0.5508) \\
			& $\widehat{\tau}$ & -9.6409 (8.5539) & -0.5325 (0.7200) & -0.3385 (0.2954) & 1.1387 (1.0594) \\
			\cline{2-6}
			\multirow{2}{*}{0.5} & $\widehat{\beta}$ & 16.1816 (3.9019) & 2.1025 (0.4026) & 1.0013 (0.1499) & -1.6689 (0.4814) \\
			& $\widehat{\tau}$ & -10.2079 (3.8779) & -0.7619 (0.5187) & -0.3806 (0.2923) & 1.2243 (0.4798) \\
			\cline{2-6}
			\multirow{2}{*}{0.75} & $\widehat{\beta}$ & 11.9536 (6.5034) & 1.2510 (0.8749) & 0.8142 (0.1644) & -1.0023 (0.8438) \\
			& $\widehat{\tau}$ & -12.6406 (6.3288) & -1.4446 (0.9279) & -0.5768 (0.3268) & 1.5746 (0.8140)\\
			\bottomrule
		\end{tabular}}
		\label{tab:estimates}	
	\end{table}

Next, we test the null hypotheses $H_0: \beta_1 = 0$, $H_0: \beta_2 = 0$, $H_0: \beta_3 = 0$ and $H_0: \tau_1 = 0$, $H_0: \tau_2 = 0$, $H_0: \tau_3 = 0$, using {the $S_W$, $S_{LR}$, $S_R$ and $S_T$ tests. The observed values of the different test statistics and the corresponding $p$-values, based on the the log-HP quantile regression model, are given in Table \ref{tab:tests}. Considering a 5\% significance level, we do not reject the hypothesis $H_0: \beta_1 = 0$ for the statistics $S_W$ and $S_R$, with $q = 0.75$. For $H_0: \beta_2 = 0$, we do not reject the null hypothesis for any of the tests. Testing $H_0: \beta_3 = 0$, at a 5\% significance level, we do not reject the null hypothesis for $S_R$, with $q = 0.25, 0.5, 0.75$ and also for the $S_W$ and $S_T$ tests, with $q= 0.75 $. Furthermore, for the coefficients associated with the skewness (or relative dispersion), that is, $\tau_1, \tau_2, \tau_3$, we do not reject $H_0: \tau_1 = 0$ only for the $S_T$ test and $H_0: \tau_3 = 0$ for the $S_T$ and $ S_{LR}$ tests, with $q = 0.75$,  at a 5\% significance level. A relevant aspect that can be noted is that the observed values of the statistics, used to test the parameters $\beta_1$, $\beta_2$ and $\beta_3$, decrease as $q$ increases.}



\begin{table}[!ht]
\footnotesize
	\centering
	\caption{\small Observed values of the $S_W$, $S_{LR}$, $S_R$ and $S_T$ test statistics and the corresponding $p$-values for the log-HP quantile regression model.}
	\adjustbox{max height=\dimexpr\textheight-3.5cm\relax,
		max width=\textwidth}{
		\begin{tabular}{cccccccccc}
			\toprule
			\multirow{2}{*}{$q$} && \multicolumn{2}{c}{$H_0: \beta_1 = 0$} && \multicolumn{2}{c}{$H_0: \beta_2 = 0$} && \multicolumn{2}{c}{$H_0: \beta_3 = 0$}\\
			\cline{3-4} \cline{6-7} \cline{9-10} 
			&& Statistics & $p$-value && Statistics & $p$-value && Statistics & $p$-value\\
			\hline
			\multirow{4}{*}{0.25} & $S_W$ & 32.855 & $<0.0001$ &&  41.958 & $<0.0001$ && 14.327 & 0.0002\\ 
			& $S_{LR}$ & 34.839 & $<0.0001$ && 26.123 & $<0.0001$ && 14.921 & 0.0001 \\
			& $S_R$ & 4.0034 & 0.0454 && 7.8686 & 0.0050 && 0.0061 & 0.9378\\
			& $S_T$ & 31.674 & $<0.0001$ && 19.982 & $<0.0001$ && 13.798 & 0.0002 \\
			\hline
			\multirow{4}{*}{0.5} & $S_W$ & 27.272 & $<0.0001$ && 44.641 & $<0.0001$ && 12.019 & 0.0005\\ 
			& $S_{LR}$ & 26.62 & $<0.0001$ && 30.3810 & $<0.0001$ && 13.223 & 0.0003 \\
			& $S_R$ & 1.6362 & 0.2008 && 10.6870 & 0.0011 && 0.0080 & 0.9289 \\
			& $S_T$ & 27.796 & $<0.0001$ && 23.819 & $<0.0001$ && 12.889 & 0.0003\\
			\hline
			\multirow{4}{*}{0.75} & $S_W$ & 2.0448 & 0.1527 && 24.532 & $<0.0001$ && 1.4109 & 0.2349 \\ 
			& $S_{LR}$ & 5.6346 & 0.0176 && 25.224 & $<0.0001$ && 3.0571 & 0.0804 \\
			& $S_R$ & 0.1559 & 0.6929 && 4.8372 & 0.0279 && 0.0009 & 0.9757 \\
			& $S_T$ & 6.5448 & 0.0105 && 21.455 & $<0.0001$ && 3.0775 & 0.0794 \\
			\bottomrule
			\multirow{2}{*}{$q$} && \multicolumn{2}{c}{$H_0: \tau_1 = 0$} && \multicolumn{2}{c}{$H_0: \tau_2 = 0$} && \multicolumn{2}{c}{$H_0: \tau_3 = 0$}\\
			\cline{3-4} \cline{6-7} \cline{9-10}
			&& Statistics & $p$-value && Statistics & $p$-value && Statistics & $p$-value\\
			\hline
			\multirow{4}{*}{0.25} & $S_W$ & 0.5470 & 0.4595 && 1.3128 & 0.2519 && 1.1553 & 0.2824 \\ 
			& $S_{LR}$ & 0.8809 & 0.3480 && 1.3403 & 0.2470 && 2.6462 & 0.1038 \\
			& $S_R$ & 0.1296 & 0.7189 && 0.5341 & 0.4649 && 0.0014 & 0.9700\\
			& $S_T$ & 0.8795 & 0.3483 && 1.3290 & 0.2490 && 2.6321 & 0.1047 \\
			\hline
			\multirow{4}{*}{0.5} & $S_W$ & 2.1577 & 0.1419 && 1.6959 & 0.1928 && 6.5109 & 0.0107 \\ 
			& $S_{LR}$ & 1.5583 & 0.2119 && 1.6458 & 0.1995 && 2.9014 & 0.0885 \\
			& $S_R$ & 0.2368 & 0.6265 && 0.7194 & 0.3963 && 0.0018 & 0.9663 \\
			& $S_T$ & 1.5671 & 0.2106 && 1.6174 & 0.2034 && 2.9183 & 0.0876 \\
			\hline
			\multirow{4}{*}{0.75} & $S_W$ & 2.4239 & 0.1195 && 3.1166 & 0.0775 && 3.7418 & 0.0531\\ 
			& $S_{LR}$ & 3.5582 & 0.0593 && 3.0260 & 0.0819 && 4.1351 & 0.0420 \\
			& $S_R$ & 0.2980 & 0.5851 && 0.7495 & 0.3866 && 0.0017 & 0.9673 \\
			& $S_T$ & 3.8473 & 0.0498 && 2.9802 & 0.0843 && 4.2370 & 0.0396 \\
			\bottomrule
	\end{tabular}}
	\label{tab:tests}
\end{table}

Figure \ref{fig:resid} displays the quantile versus quantile (QQ) plots with simulated envelope of the GCS and RQ residuals for the reduced log-HP quantile regression model with $q = 0.25, 0.5, 0.75$. This reduced model considers only significant predictors according to Table \ref{tab:tests}, such that $\widehat{Q}_i = \exp(\widehat{\beta}_0 + \widehat{\beta_1} x_{1i} + \widehat{\beta_2} x_{2i} + \widehat{\beta}_3 x_{3i})$ and $\widehat{\phi}=\exp(\widehat{\tau}_0)$, $i = 1,2, \ldots, 155$. From Figure \ref{fig:resid}, we observe that most of the points are within the bands, indicating that the GCS and RQ residuals in the reduced log-HP quantile regression model show good agreements with the expected EXP(1) and N(0,1) distributions, respectively.

\begin{figure}[H]
	\centering
	\subfigure[$q = 0.25$]{\includegraphics[scale = 0.2]{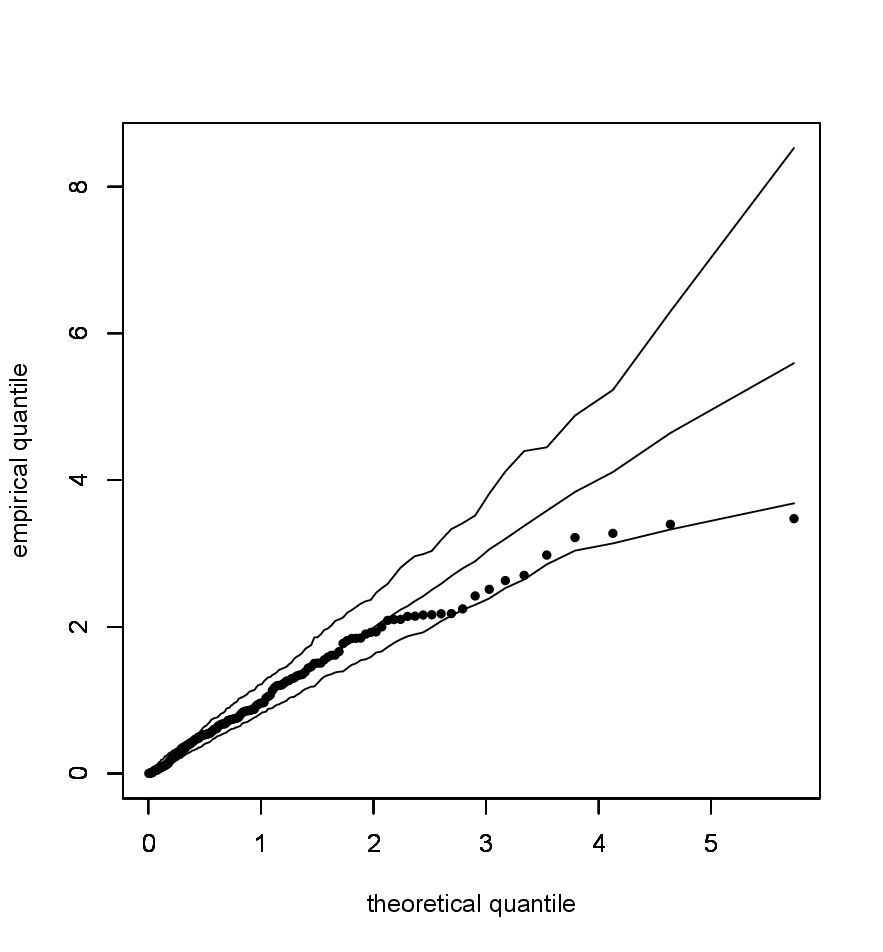}}
	\subfigure[$q = 0.5$]{\includegraphics[scale = 0.2]{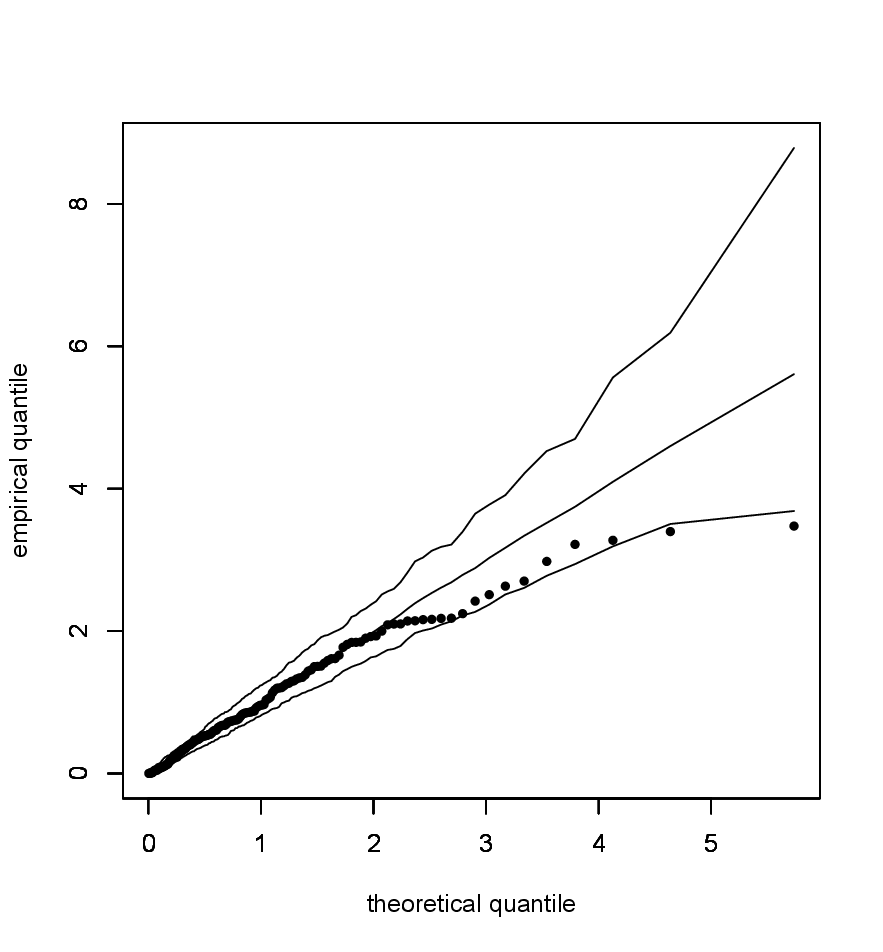}}
	\subfigure[$q = 0.75$]{\includegraphics[scale = 0.2]{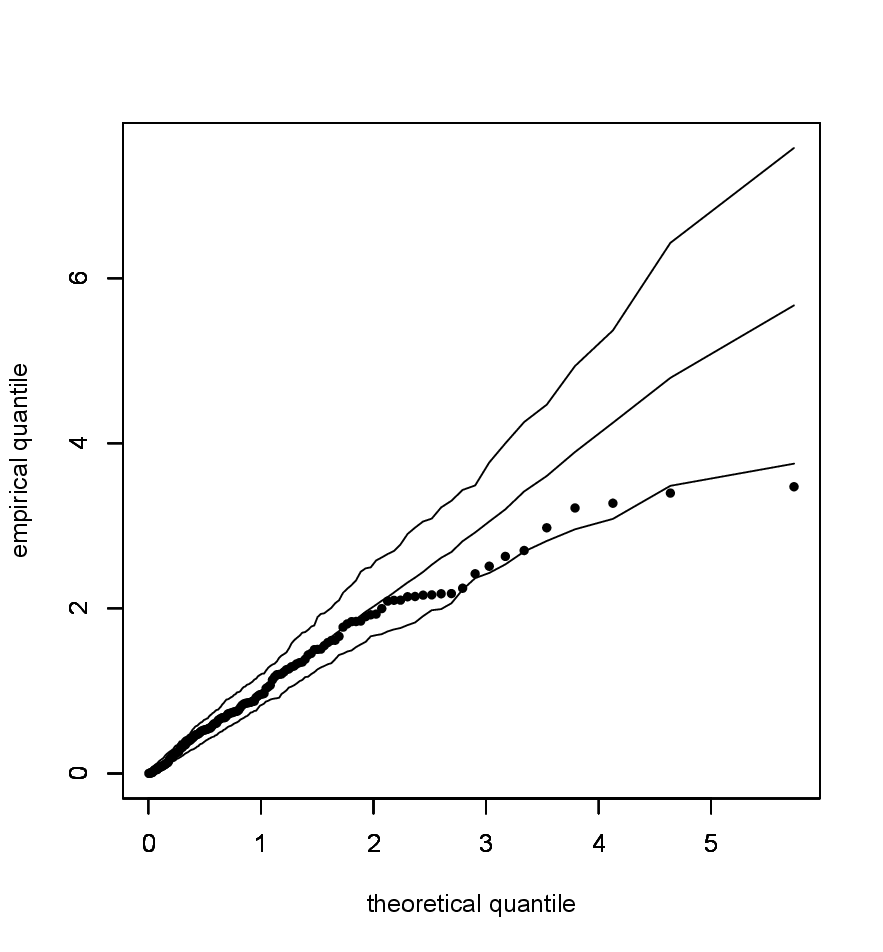}}\\
	\subfigure[$q = 0.25$]{\includegraphics[scale = 0.2]{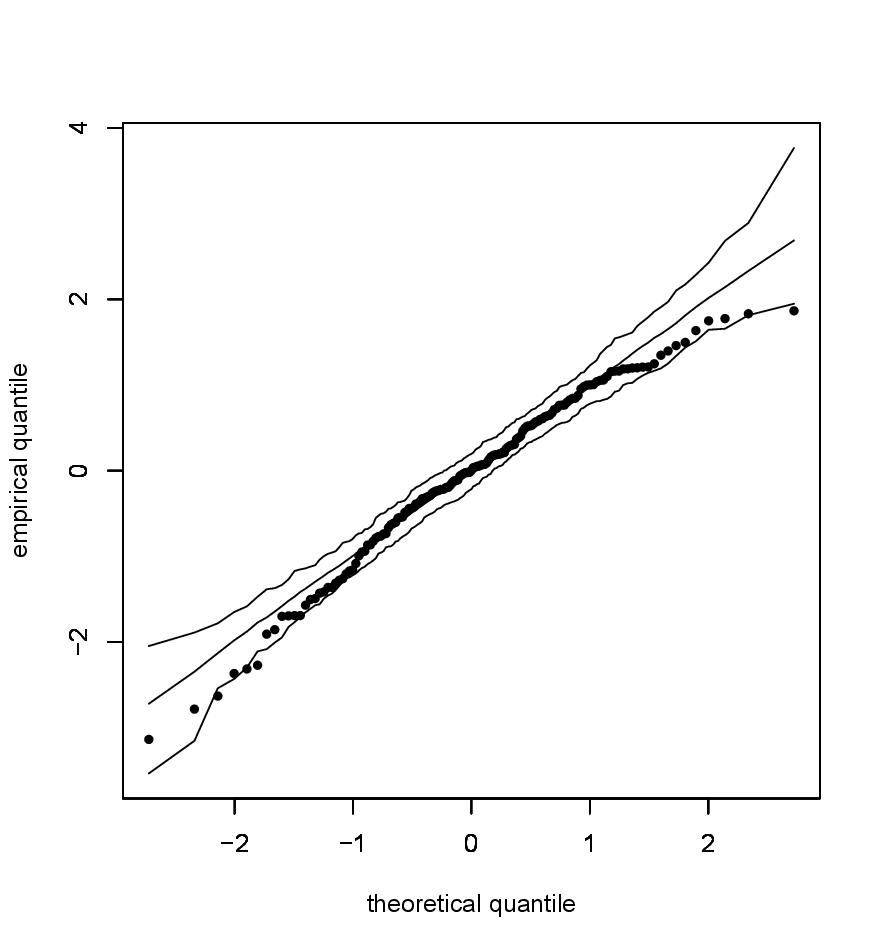}}
	\subfigure[$q = 0.5$]{\includegraphics[scale = 0.2]{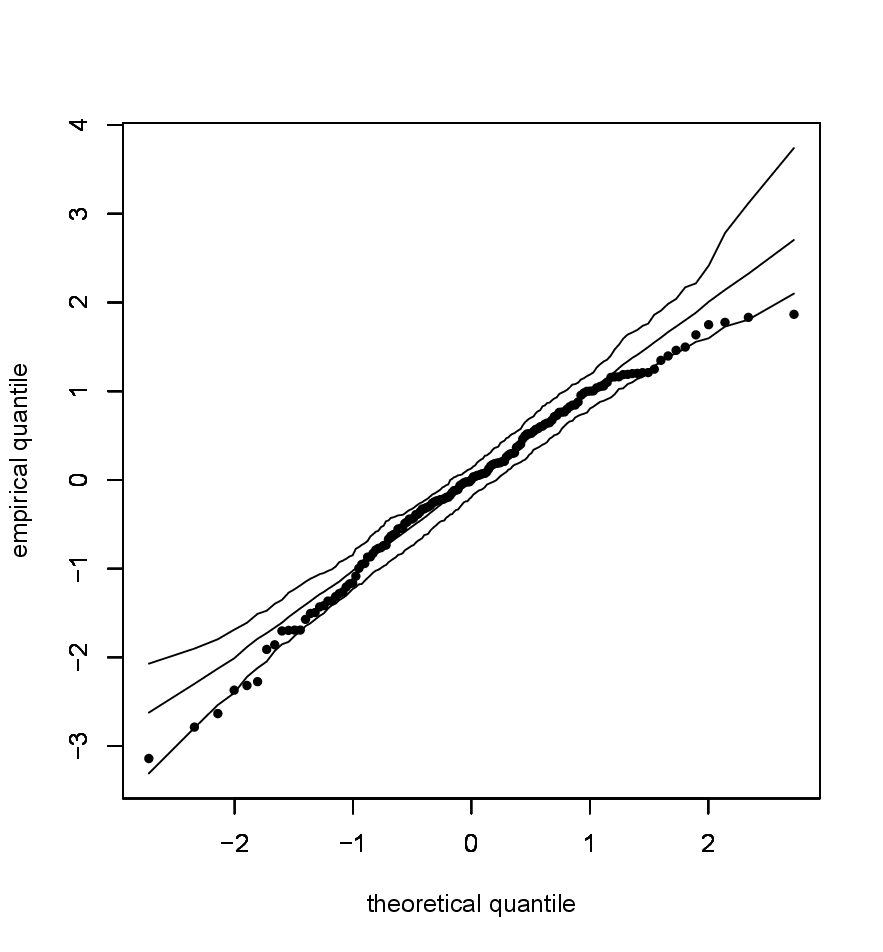}}
	\subfigure[$q = 0.75$]{\includegraphics[scale = 0.2]{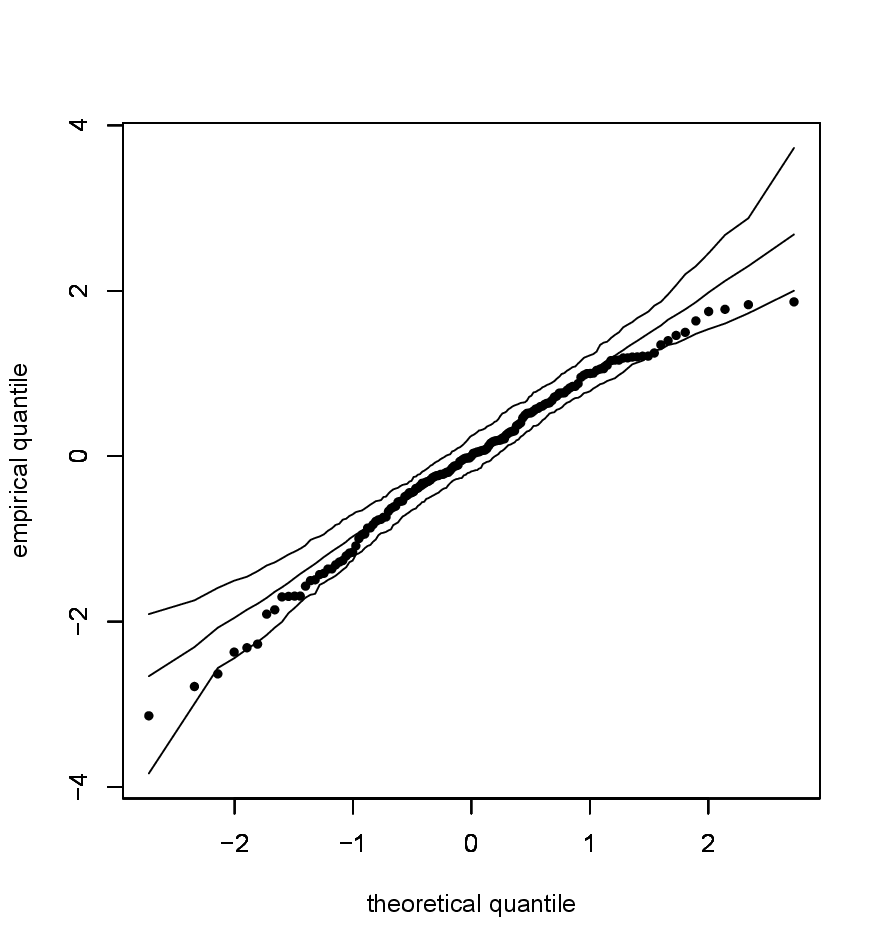}}
 \caption{\small {QQ plot and its envelope for the GCS {((a), (b) and (c))} and RQ { ((d), (e) and (f))} residuals for the reduced log-HP quantile regression model with $q = 0.25, 0.5, 0.75$.}}
	\label{fig:resid}
\end{figure}

\section{Concluding remarks}\label{sec:6}
  
We have introduce and analyzes a new class of quantile regression models based on a proposed reparameterization of log-symmetric distributions. The proposed models is a flexible alternative in the modeling of positive asymmetric data, in addition to being more informative, since it shows different effects of the covariates on the response along the quantiles of the response. Two Monte Carlo simulations were carried out to evaluate the behaviour of the maximum likelihood estimates, popular information criteria, generalized Cox-Snell and randomized quantile residuals, and Wald, score, likelihood ratio and gradient tests. The simulation results (a) have shown good performaces of the maximum likelihood estimates; (b) indicated that the success rates of AIC, BIC and AICc tend to increase with the increase in the sample size $n$, and that the rates are larger for the quantile regression models based on the log-normal and extended Birnbaum-Saunders distribution; (c) indicated that the generalized Cox-Snell and randomized quantile residuals conform well with their respective reference distributions; and (d) indicated that the Wald and likelihood ratio tests present null rejection rates closer to the corresponding nominal levels. Moreover, we observe that the power for testing three parameters simultaneously is greater for the score test, followed by the gradient test, and that the power for testing only one parameter with the score and gradient tests practically coincides. We have applied the proposed models to a real data set on movie industry. The application has shown the flexibility of the proposed models as different effects of the covariates on the response across the quantiles of the response can be studied.

\normalsize


\end{document}